\documentclass[a4paper]{article}%
\usepackage{amsfonts}
\usepackage{amsmath}
\usepackage{amssymb}
\usepackage{amstext}
\usepackage{graphicx}
\usepackage[a4paper]{geometry}
\usepackage[toc,page,title,titletoc,header]{appendix}%
\providecommand{\U}[1]{\protect\rule{.1in}{.1in}}
\newtheorem{theorem}{Theorem}[section]

\newtheorem{corollary}[theorem]{Corollary}

\newtheorem{definition}[theorem]{Definition}

\newtheorem{lemma}[theorem]{Lemma}

\newtheorem{proposition}[theorem]{Proposition}

\newenvironment{proof}[1][Proof]{\noindent\textbf{#1.} }{\ \rule{0.5em}{0.5em}}
\numberwithin{equation}{section}
\ifx\pdfoutput\relax\let\pdfoutput=\undefined\fi
\newcount\msipdfoutput
\ifx\pdfoutput\undefined\else
\ifcase\pdfoutput\else
\ifx\paperwidth\undefined\else
\ifdim\paperheight=0pt\relax\else\pdfpageheight\paperheight\fi
\ifdim\paperwidth=0pt\relax\else\pdfpagewidth\paperwidth\fi
\fi\fi\fi
\begin{document}

\title{Wave operators, similarity and dynamics for a class of Schr\"{o}dinger
operators with generic non-mixed interface conditions in 1D.}
\author{Andrea Mantile\thanks{Laboratoire de Math\'{e}matiques de Reims, EA4535 URCA,
F\'{e}d\'{e}ration de Recherche ARC Math\'{e}matiques, FR 3399 CNRS.}}
\date{}
\maketitle

\begin{abstract}
We consider a simple modification of the 1D-Laplacian where non-mixed
interface conditions occur at the boundaries of a finite interval. It has
recently been shown that Schr\"{o}dinger operators having this form allow a
new approach to the transverse quantum transport through resonant
heterostructures. In this perspective, it is important to control the
deformations effects introduced on the spectrum and on the time propagator by
this class of non-selfadjont perturbations. In order to obtain uniform-in-time
estimates of the perturbed semigroup, our strategy consists in constructing
stationary waves operators allowing to intertwine the modified non-selfadjoint
Schr\"{o}dinger operator with a 'physical' Hamiltonian. For small values of a
deformation parameter '$\theta$', this yields a dynamical comparison between
the two models showing that the distance between the corresponding semigroups
is dominated by $\left\vert \theta\right\vert $ uniformly in time in the
$L^{2}$-operator norm.

\end{abstract}

\begin{description}
\item[AMS Subject Classification: ] 81Q12, 47A40, 58J50
\end{description}

\section{Introduction.}

Schr\"{o}dinger operators with non-mixed interface conditions have been
recently considered in \cite{FMN2}, by introducing the modified 1D Laplacian
$\Delta_{\theta}$%
\begin{equation}
\left\{
\begin{array}
[c]{l}%
\medskip D(\Delta_{\theta})=\left\{  u\in H^{2}(\mathbb{R}\backslash\left\{
a,b\right\}  ):\left[
\begin{array}
[c]{l}%
\smallskip e^{-\frac{\theta}{2}}u(b^{+})=u(b^{-});\ e^{-\frac{3}{2}\theta
}u^{\prime}(b^{+})=u^{\prime}(b^{-})\\
e^{-\frac{\theta}{2}}u(a^{-})=u(a^{+});\ e^{-\frac{3}{2}\theta}u^{\prime
}(a^{-})=u^{\prime}(a^{+})
\end{array}
\right.  \right\}  \medskip\,,\\
\Delta_{\theta}u(x)=u^{\prime\prime}(x)\quad\text{for }x\in\mathbb{R}%
\backslash\left\{  a,b\right\}  \,.
\end{array}
\right.  \label{Laplacian_mod}%
\end{equation}
where $u(x^{\pm})$ respectively denote the right and left limit of the
function $u$ in $x$. For all $\theta\in\mathbb{C}\backslash\left\{  0\right\}
$, the operator $\Delta_{\theta}$ describes a singularly perturbed Laplacian,
with non-selfadjoint point interactions acting in the boundary points
$\left\{  a,b\right\}  $. It is worthwhile to notice that the boundary
conditions in (\ref{Laplacian_mod}) do not model an usual non-selfadjoint
point interaction (that is $\delta$ or $\delta^{\prime}$ type).

The interest in quantum models arising from $\Delta_{\theta}$ stands upon the
fact that a sharp exterior complex dilation, depending on $\theta=i\tau$ with
$\tau>0$, maps $-i\Delta_{\theta}$ into the accretive operator: $\left.
-ie^{-2\theta\,1_{\mathbb{R}\backslash\left(  a,b\right)  }(x)}\Delta
_{2\theta}\right.  $, where $1_{D}$ denotes the characteristic function of the
domain $D$ (e.g. in Lemma 3.1 in \cite{FMN2}). For a short-range potential
$\mathcal{V}$ (i.e.: $\mathcal{V}\in L^{1}$) compactly supported in $\left[
a,b\right]  $, the corresponding complex deformed Schr\"{o}dinger operator%
\begin{equation}
\mathcal{H}_{\theta}\left(  \mathcal{V},\theta\right)  =-e^{-2\theta
\,1_{\mathbb{R}\backslash\left(  a,b\right)  }(x)}\Delta_{2\theta
}+\mathcal{V\,}\,,\qquad\text{supp}\mathcal{V}=\left[  a,b\right]  \,,
\label{H_mod_def}%
\end{equation}
is the generator of a contraction semigroup and, in the case of time dependent
potentials, uniform-in-time estimates hold for the dynamical system. According
to the complex dilation technique (see \cite{AgCo}, \cite{BaCo}) the quantum
resonances of the undeformed operator $\mathcal{H}_{\theta}\left(
\mathcal{V}\right)  $%
\begin{equation}
\mathcal{H}_{\theta}\left(  \mathcal{V}\right)  =-\Delta_{\theta}%
+\mathcal{V}\,, \label{H_mod}%
\end{equation}
are detected by exterior complex dilations and identify with the spectral
points of $\mathcal{H}_{\theta}\left(  \mathcal{V},\theta\right)  $ in a
suitable sector of the second Riemann sheet. Then, the adiabatic evolution
problem for the resonant states of $\mathcal{H}_{\theta}\left(  \mathcal{V}%
\right)  $ rephrases, through an exterior complex dilation, as the adiabatic
evolution problem for the corresponding eigenstates of $\mathcal{H}_{\theta
}\left(  \mathcal{V},\theta\right)  $. In this framework, accounting the
contractivity property of the semigroup $e^{-it\mathcal{H}_{\theta}\left(
\mathcal{V},\theta\right)  }$ a 'standard' adiabatic theory can be developed
(e.g. in \cite{Nenciu}). This approach has been introduced in \cite{FMN2}
where an adiabatic theorem is obtained for shape resonances in the regime of
quantum wells in a semiclassical island. The purpose of this work is to
justify the use of Hamiltonians of the type $\mathcal{H}_{\theta}\left(
\mathcal{V}\right)  $ in the modelling of quantum systems.

The relevance of the artificial interface conditions (\ref{Laplacian_mod}),
stands upon the fact that they are expected to introduce small errors, w.r.t.
the selfadjoint case, controlled by $\left\vert \theta\right\vert $. The
quantum dynamics generated by $\Delta_{\theta}$ has been considered in
\cite{FMN2}. The explicit character of the model allows to obtain the
asymptotic expansion%
\begin{equation}
e^{-it\Delta_{\theta}}=e^{-it\Delta}+\mathcal{R}\left(  t,\theta\right)  \,,
\label{Exp_FMN1}%
\end{equation}
holding in a suitable neighbourhood: $\left\vert \theta\right\vert <\delta$
(see Proposition 2.2 in \cite{FMN2}). Here, the reminder $\mathcal{R}\left(
t,\theta\right)  $ is strongly continuous w.r.t. $t$ and $\theta$, exhibits
the group property w.r.t. the time variable and is such that%
\begin{equation}
\sup_{t\in\mathbb{R\,}}\left\Vert \mathcal{R}\left(  t,\theta\right)
\right\Vert _{\mathcal{L}\left(  L^{2}(\mathbb{R})\right)  }=\mathcal{O}%
\left(  \left\vert \theta\right\vert \right)  \,. \label{Exp_FMN1_rem}%
\end{equation}
Thus, for $\theta$ small enough, $\Delta_{\theta}$ generates a group, strongly
continuous both w.r.t. $t$ and $\theta$, and allowing uniform-in-time
estimates. In the perspective of modelling realistic physical situations
through the modified Schr\"{o}dinger operators $\mathcal{H}_{\theta}\left(
\mathcal{V}\right)  $, an important step would consists in extending to this
class of operators the expansion obtained in (\ref{Exp_FMN1}) for
$\mathcal{H}_{\theta}\left(  0\right)  =-\Delta_{\theta}$. A possible approach
considers $\mathcal{H}_{\theta}\left(  \mathcal{V}\right)  $ as a
(selfadjoint) perturbation of the modified Laplacian: $\mathcal{H}_{\theta
}\left(  0\right)  $; it is worthwhile to notice that this would give a weaker
result. For instance, implementing a Picard iteration on the Duhamel formula%
\begin{equation}
u_{t}(\theta)=e^{-it\Delta_{\theta}}u_{0}+i\int_{0}^{t}e^{-i\left(
t-s\right)  \Delta_{\theta}}\mathcal{V}u_{s}(\theta)\,ds\,, \label{Duhamel}%
\end{equation}
and making use of the expansion (\ref{Exp_FMN1}), yields, in the case of a
bounded potential $\mathcal{V}\in L^{\infty}$, the time dependent estimates%
\begin{align}
\left\Vert u_{t}(\theta)\right\Vert _{L^{2}(\mathbb{R})}  &  \leq
C_{1}\left\Vert u_{0}\right\Vert _{L^{2}(\mathbb{R})}\,e^{C_{2}\left\Vert
\mathcal{V}\right\Vert _{L^{\infty}(\mathbb{R})}t}\label{est5}\\
& \nonumber\\
\left\Vert u_{t}(\theta)-u_{t}(0)\right\Vert _{L^{2}(\mathbb{R})}  &  \leq
C_{3}\left\vert \theta\right\vert \,\left\Vert u_{0}\right\Vert _{L^{2}%
(\mathbb{R})}\,\,te^{C_{4}\left\Vert \mathcal{V}\right\Vert _{L^{\infty
}(\mathbb{R})}t} \label{est5_1}%
\end{align}
where $C_{i}$, $i=1,..4$, are suitable positive constants. It follows that,
for an initial state $u_{0}$, the corresponding mild solution to the quantum
evolution problem, $u_{t}(\theta)$, is Lipschitz continuous w.r.t. $\theta$,
with a Lipschitz constant bounded by an exponentially increasing function of
time. As an aside, we notice that the estimate (\ref{est5}) may also be
obtained as a consequence of the Hille-Yoshida-Phillips Theorem, by using the
second resolvent formula for $\left(  \mathcal{H}_{\theta}\left(
\mathcal{V}\right)  -z\right)  ^{-1}$ and resolvent estimates for $\left(
\mathcal{H}_{\theta}\left(  0\right)  -z\right)  ^{-1}$ arising from
(\ref{Exp_FMN1}).

The relation (\ref{est5_1}) yields a finite-time control, depending on
$\left\Vert \mathcal{V}\right\Vert _{L^{\infty}(\mathbb{R})}$, of the error
introduced on the quantum evolution by the interface conditions. However, when
$\mathcal{V}$ describes the (possibly non-linear) interactions involving
charge carriers in resonant heterostructures, its norm $\left\Vert
\mathcal{V}\right\Vert _{L^{\infty}(\mathbb{R})}$ is expected to be small
compared to the energy of the particles, while the quantum evolution of
relevant observables is characterized by a long time scale, corresponding to
the inverse of the imaginary part of the shape resonances (examples of this
mechanism are exhibited in \cite{PrSj} and \cite{FMN3}). In this framework,
the use of modified Hamiltonians of the type $\mathcal{H}_{\theta}\left(
\mathcal{V}\right)  $ would be justified by a stronger uniform-in-time
estimate for the error $\left\Vert u_{t}(\theta)-u_{t}(0)\right\Vert
_{L^{2}(\mathbb{R})}$ as $\theta\rightarrow0$.

Adopting a different approach, in what follows the operator $\mathcal{H}%
_{\theta}\left(  \mathcal{V}\right)  $ is considered as a non-selfadjoint
perturbation of the selfadjoint Hamiltonian $\mathcal{H}_{0}\left(
\mathcal{V}\right)  $. Non-selfadjoint perturbations of the type
$T(x)=T+xAB^{\ast}$ have been studied in \cite{Kato} where, under smoothness
assumptions on $A$ and $B$, the 'stationary' wave operators for the couple
$\left\{  T(x),T\right\}  $ are given and the corresponding similarity between
$T$ and $T(x)$ is exploited to define the dynamics generated by $-iT(x)$. This
strategy is adapted here to the case where $T=\mathcal{H}_{0}\left(
\mathcal{V}\right)  $, while the perturbation is determined by generic,
non-mixed, interface conditions occurring at the boundaries of the potential's
support. This is a larger class of operators, parametrized by a couple of
complex, which includes both the cases of $\mathcal{H}_{\theta}\left(
\mathcal{V}\right)  $ and $\left(  \mathcal{H}_{\theta}\left(  \mathcal{V}%
\right)  \right)  ^{\ast}$. From an accurate resolvent analysis and explicit
generalized eigenfunctions formulas, we deduce, in this extended framework, a
small-$\theta$ expansion of the 'stationary wave operators'. Then the quantum
evolution group generated by $-i\mathcal{H}_{\theta}\left(  \mathcal{V}%
\right)  $ is determined by conjugation from $e^{-it\mathcal{H}_{0}\left(
\mathcal{V}\right)  }$ and an uniform-in-time estimate for the 'distance'
between the two dynamics is obtained (see Theorem \ref{Theorem1} below).

Similarity transformations, from non-selfadjoint to similar selfadjoint
operators, have been recently studied in \cite{KreSieZel}, where the authors
focus on the particular case of 1D Schr\"{o}dinger operators defined with
non-selfadjoint Robin-type conditions occurring at the boundary of an
interval. In the case of parity and time-reversal symmetry ($\mathcal{PT}%
$-symmetry), the similarity of this model with a selfadjoint Hamiltonian is
derived. It is worth noticing that, when $\theta\in i\mathbb{R}$, the modified
Laplacian $\Delta_{\theta}$ actually exhibits the $\mathcal{PT}$-symmetry
(once the parity is defined with respect to the point $\left(  a+b\right)
/2$). However, the models introduced in the next sections are generically not
$\mathcal{PT}$-symmetric (see the definition (\ref{Q_teta}) below).

\subsection{Schr\"{o}dinger operators with non-mixed interface conditions.}

We consider the family of modified Schr\"{o}dinger operators $Q_{\theta
_{1},\theta_{2}}(\mathcal{V})$, depending on a couple of complex parameters,
$\left(  \theta_{1},\theta_{2}\right)  \in\mathbb{C}^{2}$, and on a
selfadjoint short-range potential, compactly supported over the interval
$\left[  a,b\right]  \subset\mathbb{R}$,%
\begin{equation}
\mathcal{V}\in L^{1}(\mathbb{R},\mathbb{R})\,,\qquad\text{supp }%
\mathcal{V}=\left[  a,b\right]  \,. \label{V}%
\end{equation}
The parameters $\theta_{1}$ and $\theta_{2}$ fix the interface conditions,%
\begin{equation}
\left\{
\begin{array}
[c]{ccc}%
e^{-\frac{\theta_{1}}{2}}u(b^{+})=u(b^{-})\,, &  & e^{-\frac{\theta_{2}}{2}%
}u^{\prime}(b^{+})=u^{\prime}(b^{-})\,,\\
&  & \\
e^{-\frac{\theta_{1}}{2}}u(a^{-})=u(a^{+})\,, &  & e^{-\frac{\theta_{2}}{2}%
}u^{\prime}(a^{-})=u(a^{+})\,,
\end{array}
\right.  \label{B_C_1}%
\end{equation}
occurring at the boundary of the potential's support and $Q_{\theta_{1}%
,\theta_{2}}(\mathcal{V})$ is defined as follows%
\begin{equation}
Q_{\theta_{1},\theta_{2}}(\mathcal{V}):\left\{
\begin{array}
[c]{l}%
D\left(  Q_{\theta_{1},\theta_{2}}(\mathcal{V})\right)  =\left\{  u\in
H^{2}\left(  \mathbb{R}\backslash\left\{  a,b\right\}  \right)  \,\left\vert
\ \text{(\emph{\ref{B_C_1}}) holds}\right.  \right\}  \,,\\
\\
\left(  Q_{\theta_{1},\theta_{2}}(\mathcal{V})\,u\right)  (x)=-u^{\prime
\prime}(x)+\mathcal{V}(x)\,u(x)\,,\qquad x\in\mathbb{R}\backslash\left\{
a,b\right\}  \,.
\end{array}
\right.  \label{Q_teta}%
\end{equation}
The set $\left\{  Q_{\theta_{1},\theta_{2}}(\mathcal{V})\,,\ \left(
\theta_{1},\theta_{2}\right)  \in\mathbb{C}^{2}\right\}  $ is closed w.r.t.
the adjoint operation: a direct computation shows that%
\begin{equation}
\left(  Q_{\theta_{1},\theta_{2}}(\mathcal{V})\right)  ^{\ast}=Q_{-\theta
_{2}^{\ast},-\theta_{1}^{\ast}}(\mathcal{V})\,. \label{Q_teta_adj}%
\end{equation}
The subset of selfadjoint operators in this class is identified by the
conditions: for $\theta_{j}=r_{j}e^{i\varphi_{j}}$, $j=1,2$,%
\begin{equation}
\left\{
\begin{array}
[c]{l}%
\varphi_{1}+\varphi_{2}=\pi+2\pi k\,,\quad k\in\mathbb{Z}\,,\\
r_{1}=r_{2}\,.
\end{array}
\right.  \label{Selafadj_cond}%
\end{equation}
When (\ref{Selafadj_cond}) are not satisfied, the corresponding operator
$Q_{\theta_{1},\theta_{2}}(\mathcal{V})$ is neither selfadjoint nor symmetric,
since in this case: $Q_{\theta_{1},\theta_{2}}(\mathcal{V})\not \subset
\left(  Q_{\theta_{1},\theta_{2}}(\mathcal{V})\right)  ^{\ast}$.

For each couple $\left\{  \theta_{1},\theta_{2}\right\}  $, $Q_{\theta
_{1},\theta_{2}}\left(  \mathcal{V}\right)  $ identifies with a, possibly
non-selfadjoint, extension of the Hermitian operator $Q^{0}(\mathcal{V})$%
\begin{equation}
D\left(  Q^{0}(\mathcal{V})\right)  =\left\{  u\in H^{2}\left(  \mathbb{R}%
\right)  \,\left\vert \ u(\alpha)=u^{\prime}(\alpha)=0\,,\ \alpha=a,b\right.
\right\}  \,,
\end{equation}
and defines an explicitly solvable model w.r.t. the selfadjoint Hamiltonian
$Q_{0,0}\left(  \mathcal{V}\right)  $. Non-selfadjoint models arising from
proper extensions of Hermitian operators with gaps have been already
considered in literature, for instance in \cite{DerMa}, \cite{Ryz1} (see also
\cite{Mamo1}-\cite{MaMo3} and \cite{BrMaNaWo} for the general case of adjoint
pairs of operators). In these works, the formalism of \emph{boundary triples
}(e.g. in \cite{LyaSto} for adjoint pairs) is adopted; this leads to
Krein-like formulas expressing the resolvent of an extended operator in terms
of the resolvent of a 'reference' extension plus a finite rank part depending
on the Weyl function of the triple. In Section \ref{Section_Resolvent_1} we
give Krein-like formulas for the difference $\left.  \left(  Q_{\theta
_{1},\theta_{2}}\left(  \mathcal{V}\right)  -z\right)  ^{-1}-\left(
Q_{0,0}\left(  \mathcal{V}\right)  -z\right)  ^{-1}\right.  $. Exploiting this
framework, we show that $Q_{\theta_{1},\theta_{2}}(\mathcal{V})$ is an
analytic family in the sense of Kato both w.r.t. the variables $\theta_{1}$
and $\theta_{2}$ and study its spectral profile depending on $\mathcal{V}$.
The result is exposed in the Proposition \ref{Proposition_spectrum}; in
particular, for a defined positive $\mathcal{V}$, we obtain: $\sigma\left(
Q_{\theta_{1},\theta_{2}}\left(  \mathcal{V}\right)  \right)  =\sigma
_{ac}\left(  Q_{0,0}\left(  \mathcal{V}\right)  \right)  =\mathbb{R}_{+}$,
provided that $\theta_{1}$ and $\theta_{2}$ are small enough.

Under the same assumptions, in Section \ref{Section_Similarity}, a family of
intertwining operators $\mathcal{W}_{\theta_{1},\theta_{2}}$\ for the
couple\newline$\left\{  Q_{\theta_{1},\theta_{2}}(\mathcal{V}),Q_{0,0}%
(\mathcal{V})\right\}  $ is introduced as the analogous of the usual
stationary wave operators in selfadjoint frameworks. Using the eigenfunctions
expansion obtained in Subsection \ref{Section_Resolvent_2}, we get a
small-$\theta_{i}$ expansion of $\mathcal{W}_{\theta_{1},\theta_{2}}$ allowing
to define the quantum evolution group $e^{-iQ_{\theta_{1},\theta_{2}}\left(
\mathcal{V}\right)  }$ from $e^{-iQ_{0,0}\left(  \mathcal{V}\right)  }$ by
conjugation. Then, we develop a quantitative comparison showing that $\left.
e^{-iQ_{\theta_{1},\theta_{2}}\left(  \mathcal{V}\right)  }-e^{-iQ_{0,0}%
\left(  \mathcal{V}\right)  }\right.  $ is controlled by $\left\vert
\theta_{i}\right\vert $, $i=1,2$, uniformly in time, in the $L^{2}$-operator
norm. The result is presented in the following theorem and the proof is given
in Subsection \ref{Section_Theorem1}. It can be adapted to the particular case
of $\mathcal{H}_{\theta}\left(  \mathcal{V}\right)  $, by noticing that:
$\mathcal{H}_{\theta}\left(  \mathcal{V}\right)  =Q_{\theta,3\theta
}(\mathcal{V})$.

\begin{theorem}
\label{Theorem1}Let $\mathcal{V}$ fulfills the conditions (\ref{V}),%
\begin{equation}
\left\langle u,\mathcal{V\,}u\right\rangle _{L^{2}(\left(  a,b\right)
)}>0\qquad\forall\,u\in L^{2}(\left(  a,b\right)  )\,,
\end{equation}
and assume $\left\vert \theta_{j}\right\vert <\delta$, $j=1,2$, with
$\delta>0$ small enough. Then $-iQ_{\theta_{1},\theta_{2}}(\mathcal{V})$
generates a strongly continuous group of bounded operators on $L^{2}\left(
\mathbb{R}\right)  $, $\left\{  e^{-itQ_{\theta_{1},\theta_{2}}(\mathcal{V}%
)}\right\}  _{t\in\mathbb{R}}$. For a fixed $t$, $e^{-itQ_{\theta_{1}%
,\theta_{2}}(\mathcal{V})}$ defines an analytic family of bounded operators
w.r.t. $\left(  \theta_{1},\theta_{2}\right)  $ and the expansion%
\begin{equation}
e^{-itQ_{\theta_{1},\theta_{2}}(\mathcal{V})}=e^{-itQ_{0,0}(\mathcal{V}%
)}+\mathcal{R}\left(  t,\theta_{1},\theta_{2}\right)  \,,
\end{equation}
holds with an uniformly bounded in time reminder s.t.%
\begin{equation}
\sup_{t\in\mathbb{R}}\left\Vert \mathcal{R}\left(  t,\theta_{1},\theta
_{2}\right)  \right\Vert _{\mathcal{L}\left(  L^{2}(\mathbb{R})\right)
}=\mathcal{O}\left(  \theta_{1}\right)  +\mathcal{O}\left(  \theta_{2}\right)
\,.
\end{equation}

\end{theorem}

In the Subsection \ref{Section_WO} the pair $\left\{  Q_{\theta_{1},\theta
_{2}}\left(  \mathcal{V}\right)  ,Q_{0,0}\left(  \mathcal{V}\right)  \right\}
$ is considered as a scattering system and we investigate the existence of
non-stationary wave operators. Under the assumptions of the Theorem
\ref{Theorem1}, it is shown that $\mathcal{W}_{\theta_{1},\theta_{2}}$
coincides with a 'physical' wave operator, according to the time dependent
definition (Lemma \ref{Lemma_WO}). Although exploited in our analysis, the
small perturbations condition does not seem to be necessary in the proof of
this result: a possible strategy for its extension to the case where $\left\{
\theta_{1},\theta_{2}\right\}  $ are not small is finally mentioned.

Further perspectives of this work, concerning the regime of quantum wells in a
semiclassical island, are discussed in the Section \ref{Sec_small_h}.

\subsection{\label{Sec_Notation}Notation}

In what follows, we make use of a generalization of the Landau notation,
$\mathcal{O}\left(  \cdot\right)  $, defined according to:

\begin{definition}
\label{Landau_Notation}Let be $X$ a metric space and $f,g:X\rightarrow
\mathbb{C}$. Then $f=\mathcal{O}\left(  g\right)  $ $\overset{def}%
{\Longleftrightarrow}$ $\forall\,x\in X$ it holds: $\left.
f(x)=p(x)g(x)\right.  $, being $p$ a bounded map $X\rightarrow\mathbb{C}$.
\end{definition}

The next notation are also adopted\smallskip\newline$1_{\Omega}(\cdot)$ is the
characteristic function of a domain $\Omega$.\smallskip\newline$\mathcal{B}%
_{\delta}(p)$ is the open ball of radius $\delta$ centered in a point
$p\in\mathbb{C}$.\smallskip\newline$\mathcal{C}_{x}^{k}(U)$ is the set of
$\mathcal{C}^{k}$-continuous functions w.r.t. $x\in U\subseteq\mathbb{R}%
$.\smallskip\newline$\mathcal{H}_{z}(D)$ is the set of holomorphic functions
w.r.t. $z\in D\subseteq\mathbb{C}$.\smallskip\newline$\partial_{j}f\left(
x_{1,}...x_{n}\right)  $, $j\in\left\{  1,..n\right\}  $, denotes the
derivative of $f$ w.r.t. the variable $x_{j}$.\smallskip\newline$S_{\eta}$
denotes the complex half-plane: $S_{\eta}=\left\{  z\in\mathbb{C\,}\left\vert
\ \operatorname{Im}z>-\eta\right.  \right\}  $. In particular, $\mathbb{C}%
^{+}$ coincides with $S_{0}$.\smallskip\newline The notation '$\lesssim$',
appearing in some of the proofs, denotes the inequality: '$\leq C$' being $C$
a suitable positive constant.

\section{\label{Section_Resolvent_1}Boundary triples and Krein-like resolvent
formulas.}

Point perturbation models, as $Q_{\theta_{1},\theta_{2}}(\mathcal{V})$, can be
described as restrictions of a \emph{larger} operator through linear relations
on an Hilbert space. Let introduce $Q(\mathcal{V})$%
\begin{equation}
\left\{
\begin{array}
[c]{l}%
D(Q(\mathcal{V}))=H^{2}\left(  \mathbb{R}\backslash\left\{  a,b\right\}
\right)  \,,\\
\\
\left(  Q(\mathcal{V})\,u\right)  (x)=-u^{\prime\prime}(x)+\mathcal{V}%
(x)\,u(x)\qquad\text{for }x\in\mathbb{R}\backslash\left\{  a,b\right\}  \,,
\end{array}
\right.  \label{Q}%
\end{equation}
with $\mathcal{V}$ defined according to (\ref{V}), and let $Q^{0}%
(\mathcal{V})$ be such that: $\left(  Q^{0}(\mathcal{V})\right)  ^{\ast
}=Q(\mathcal{V})$. Explicitly, $Q^{0}(\mathcal{V})$ identifies with the
symmetric restriction of $Q(\mathcal{V})$ to the domain
\begin{equation}
D\left(  Q^{0}(\mathcal{V})\right)  =\left\{  u\in D(Q(\mathcal{V}%
))\,\left\vert \ u(\alpha)=u^{\prime}(\alpha)=0\ \forall\,\alpha\in\left\{
a,b\right\}  \right.  \right\}  \,.
\end{equation}
The related defect spaces, $\mathcal{N}_{z}=\ker(Q(\mathcal{V})-z)$, are
$4$-dimensional subspaces of $D(Q(\mathcal{V}))$ generated, for $z\in
\mathbb{C}\backslash\mathbb{R}$, by the independent solutions to the problem%
\begin{equation}
\left\{
\begin{array}
[c]{l}%
(-\partial_{x}^{2}+\mathcal{V}-z)u(x)=0\,,\quad x\in\mathbb{R}\backslash
\left\{  a,b\right\}  \,,\\
\\
u\in D(Q(\mathcal{V}))\,.
\end{array}
\right.  \label{defect_fun}%
\end{equation}
A \emph{boundary triple} $\left\{  \mathbb{C}^{4},\Gamma_{0},\Gamma
_{1}\right\}  $ for $Q(\mathcal{V})$ is defined with two linear boundary maps
$\Gamma_{i=1,2}:D(Q(\mathcal{V}))\rightarrow\mathbb{C}^{4}$ fulfilling, for
any $\psi,\varphi\in D(Q(\mathcal{V}))$, the equation%
\begin{equation}
\left\langle \psi,Q(\mathcal{V})\varphi\right\rangle _{L^{2}(\mathbb{R}%
)}-\left\langle Q(\mathcal{V})\psi,\varphi\right\rangle _{L^{2}(\mathbb{R}%
)}=\left\langle \Gamma_{0}\psi,\Gamma_{1}\varphi\right\rangle _{\mathbb{C}%
^{4}}-\left\langle \Gamma_{1}\psi,\Gamma_{0}\varphi\right\rangle
_{\mathbb{C}^{4}}\,,\label{BVT_1}%
\end{equation}
and such that the transformation $\left(  \Gamma_{0},\Gamma_{1}\right)
:D(Q(\mathcal{V}))\rightarrow\mathbb{C}^{4}\times\mathbb{C}^{4}$ is
surjective. A proper\emph{ }extension $Q_{ext}$ of $Q^{0}(\mathcal{V})$ is
called \emph{almost solvable }if there exists a boundary triple $\left\{
\mathbb{C}^{4},\Gamma_{0},\Gamma_{1}\right\}  $ and a matrix $M\in
\mathbb{C}^{4,4}$ such that it coincides with the restriction of
$Q(\mathcal{V})$ to the domain: $\left\{  u\in D(Q(\mathcal{V}))\,\left\vert
\ M\Gamma_{0}u=\Gamma_{1}u\right.  \right\}  $. Using the notation
$Q_{ext}=Q_{M}(\mathcal{V})$, the characterization%
\begin{equation}
Q^{0}(\mathcal{V})\subset Q_{M}(\mathcal{V})\subset Q(\mathcal{V}%
)\,,\quad\left(  Q_{M}(\mathcal{V})\right)  ^{\ast}=Q_{M^{\ast}}%
(\mathcal{V})\,.\label{extensions}%
\end{equation}
holds (e.g. \cite{Ryz1}, Theorem 1.1). In what follows, $\tilde{Q}%
(\mathcal{V})$ denotes the particular restriction of $Q(\mathcal{V})$
associated with the conditions: $\Gamma_{0}u=0$, i.e.%
\begin{equation}
D\left(  \tilde{Q}(\mathcal{V})\right)  =\left\{  u\in D\left(  Q(\mathcal{V}%
)\right)  \,,\ \Gamma_{0}u=0\right\}  \,.
\end{equation}
According to the relation (\ref{BVT_1}), $\tilde{Q}(\mathcal{V})$ is
selfadjoint and $\mathbb{C}\backslash\mathbb{R}\subset\mathcal{\rho}\left(
\tilde{Q}(\mathcal{V})\right)  $. Let, $\gamma(z,\mathcal{V})$ and
$q(z,\mathcal{V})$ be the linear maps defined by%
\begin{equation}
\gamma(z,\mathcal{V})=\left(  \left.  \Gamma_{0}\right\vert _{\mathcal{N}_{z}%
}\right)  ^{-1}\,,\qquad q(z,\mathcal{V})=\Gamma_{1}\circ\gamma(z,\mathcal{V}%
)\,,\qquad z\in\mathcal{\rho}\left(  \tilde{Q}(\mathcal{V})\right)
\,,\label{gamma_q_def}%
\end{equation}
where $\left.  \Gamma_{0}\right\vert _{\mathcal{N}_{z}}$ is the restriction of
$\Gamma_{0}$ to $\mathcal{N}_{z}$. These define holomorphic families of
bounded operators in $\mathcal{L}\left(  \mathbb{C}^{4},L^{2}\left(
\mathbb{R}\right)  \right)  $ and $\mathcal{L}\left(  \mathbb{C}%
^{4},\mathbb{C}^{4}\right)  $ (e.g. in \cite{DerMa} and \cite{BruGeyPan}). The
maps $\gamma(\cdot,z,\mathcal{V})$ and $q(z,\mathcal{V})$ are respectively
referred to as the \emph{Gamma field }and the \emph{Weyl function} associated
with the triple $\left\{  \mathbb{C}^{4},\Gamma_{0},\Gamma_{1}\right\}  $.
With this formalism, a resolvent formula expresses the difference $\left.
\left(  Q_{M}(\mathcal{V})-z\right)  ^{-1}-\left(  \tilde{Q}(\mathcal{V}%
)-z\right)  ^{-1}\right.  $ in terms of finite rank operator with range
$\mathcal{N}_{z}$%
\begin{equation}
\left(  Q_{M}(\mathcal{V})-z\right)  ^{-1}-\left(  \tilde{Q}(\mathcal{V}%
)-z\right)  ^{-1}=\gamma(z,\mathcal{V})\left(  M-q(z,\mathcal{V})\right)
^{-1}\gamma^{\ast}(\bar{z},\mathcal{V})\,,\quad z\in\rho\left(  Q_{M}%
(\mathcal{V})\right)  \cap\mathcal{\rho}\left(  \tilde{Q}(\mathcal{V})\right)
\label{krein}%
\end{equation}
(e.g. in \cite{Ryz1}, Theorem 1.2). In many situations, the interface
conditions occurring in the points $\left\{  a,b\right\}  $ can also be
represented in the form: $A\Gamma_{0}u=B\Gamma_{1}u$, where $A,B\in
\mathbb{C}^{4,4}$. We denote with $Q_{A,B}(\mathcal{V})$ the corresponding
restriction%
\begin{equation}
\left\{
\begin{array}
[c]{l}%
D\left(  Q_{A,B}(\mathcal{V})\right)  =\left\{  u\in D(Q(\mathcal{V}%
))\,\left\vert \ A\Gamma_{0}u=B\Gamma_{1}u\right.  \right\}  \,,\\
\\
Q_{A,B}(\mathcal{V})\,u=Q(\mathcal{V})\,u\,.
\end{array}
\right.  \,.\label{restriction_def}%
\end{equation}
With this parametrization, we have: $\tilde{Q}(\mathcal{V})=Q_{1,0}%
(\mathcal{V})$, while the resolvent's formula rephrases as%
\begin{equation}
\left(  Q_{M}(\mathcal{V})-z\right)  ^{-1}-\left(  \tilde{Q}(\mathcal{V}%
)-z\right)  ^{-1}=-\gamma(z,\mathcal{V})\left[  \left(  Bq(z,\mathcal{V}%
)-A\right)  ^{-1}B\right]  \gamma^{\ast}(\bar{z},\mathcal{V})\,,\qquad
z\in\rho\left(  Q_{M}(\mathcal{V})\right)  \label{krein_AB}%
\end{equation}

In the perspective of a comparison between the quantum models arising from
$Q_{\theta_{1},\theta_{2}}(\mathcal{V})$ and $Q_{0,0}(\mathcal{V})$ a natural
choice is%
\begin{equation}%
\begin{array}
[c]{ccc}%
\Gamma_{0}u=%
\begin{pmatrix}
u^{\prime}(b^{-})-u^{\prime}(b^{+})\smallskip\medskip\\
u(b^{+})-u(b^{-})\smallskip\medskip\\
u^{\prime}(a^{-})-u^{\prime}(a^{+})\smallskip\medskip\\
u(a^{+})-u(a^{-})
\end{pmatrix}
\,, &  & \Gamma_{1}u=\frac{1}{2}%
\begin{pmatrix}
u(b^{+})+u(b^{-})\smallskip\medskip\\
u^{\prime}(b^{+})+u^{\prime}(b^{-})\smallskip\medskip\\
u(a^{+})+u(a^{-})\smallskip\medskip\\
u^{\prime}(a^{+})+u^{\prime}(a^{-})
\end{pmatrix}
\,,
\end{array}
\, \label{BVT}%
\end{equation}
which leads to: $\tilde{Q}(\mathcal{V})=Q_{0,0}(\mathcal{V})$. According to
the definitions (\ref{Q_teta}) and (\ref{Q}), the operator $Q_{\theta
_{1},\theta_{2}}(\mathcal{V})$ identifies with the restriction of
$Q(\mathcal{V})$ parametrized by the $\mathbb{C}^{4,4}$-block-diagonal
matrices%
\begin{equation}%
\begin{array}
[c]{ccc}%
A_{\theta_{1},\theta_{2}}=%
\begin{pmatrix}
a(\theta_{1},\theta_{2}) &  & \\
&  & \\
&  & a(-\theta_{1},-\theta_{2})
\end{pmatrix}
\,, &  & B_{\theta_{1},\theta_{2}}=%
\begin{pmatrix}
b(\theta_{1},\theta_{2}) &  & \\
&  & \\
&  & b(-\theta_{1},-\theta_{2})
\end{pmatrix}
\,,
\end{array}
\label{AB_teta1,2}%
\end{equation}
defined with%
\begin{equation}%
\begin{array}
[c]{ccc}%
a(\theta_{1},\theta_{2})=%
\begin{pmatrix}
1+e^{\frac{\theta_{2}}{2}} & 0\\
0 & 1+e^{\frac{\theta_{1}}{2}}%
\end{pmatrix}
\,, &  & b(\theta_{1},\theta_{2})=2%
\begin{pmatrix}
0 & 1-e^{\frac{\theta_{2}}{2}}\\
e^{\frac{\theta_{1}}{2}}-1 & 0
\end{pmatrix}
\,.
\end{array}
\label{ab_teta1,2}%
\end{equation}
Using (\ref{BVT}) and (\ref{ab_teta1,2})-(\ref{ab_teta1,2}), the linear
relations (\emph{\ref{B_C_1}}) rephrase as%
\begin{equation}
A_{\theta_{1},\theta_{2}}\Gamma_{0}u=B_{\theta_{1},\theta_{2}}\Gamma_{1}u\,,
\label{bc_teta1_teta_2}%
\end{equation}
which leads to the equivalent definition
\begin{equation}
Q_{\theta_{1},\theta_{2}}(\mathcal{V}):\left\{
\begin{array}
[c]{l}%
\mathcal{D}\left(  Q_{\theta_{1},\theta_{2}}(\mathcal{V})\right)  =\left\{
u\in D(Q(\mathcal{V}))\,\left\vert \ A_{\theta_{1},\theta_{2}}\Gamma
_{0}u=B_{\theta_{1},\theta_{2}}\Gamma_{1}u\right.  \right\}  \,,\\
\\
Q_{\theta_{1},\theta_{2}}(\mathcal{V})\,u=Q(\mathcal{V})\,u\,.
\end{array}
\right.  \label{Q_teta1_teta2}%
\end{equation}
In this framework, the relation (\ref{krein_AB}) explicitly writes as%
\begin{equation}
\left(  Q_{\theta_{1},\theta_{2}}(\mathcal{V})-z\right)  ^{-1}=\left(
Q_{0,0}(\mathcal{V})-z\right)  ^{-1}-\sum_{i,j=1}^{4}\left[  \left(
B_{\theta_{1},\theta_{2}}\,q(z,\mathcal{V})-A_{\theta_{1},\theta_{2}}\right)
^{-1}B_{\theta_{1},\theta_{2}}\right]  _{ij}\left\langle \gamma(e_{j},\bar
{z},\mathcal{V}),\cdot\right\rangle _{L^{2}(\mathbb{R})}\gamma(e_{i}%
,z,\mathcal{V})\,, \label{krein_1}%
\end{equation}
where $\left\{  e_{i}\right\}  _{i=1}^{4}$ is the standard basis in
$\mathbb{C}^{4}$, while $\gamma(v,z,\mathcal{V})$ denotes the action of
$\gamma(z,\mathcal{V})$ on the vector $v$. The corresponding integral kernel,
$\mathcal{G}_{\theta_{1},\theta_{2}}^{z}(x,y)$, is%
\begin{equation}
\mathcal{G}_{\theta_{1},\theta_{2}}^{z}(x,y)=\mathcal{G}_{0,0}^{z}%
(x,y)-\sum_{i,j=1}^{4}\left[  \left(  B_{\theta_{1},\theta_{2}}%
\,q(z,\mathcal{V})-A_{\theta_{1},\theta_{2}}\right)  ^{-1}B_{\theta_{1}%
,\theta_{2}}\right]  _{ij}\gamma(e_{j},y,z,\mathcal{V})\,\gamma(e_{i}%
,x,z,\mathcal{V})\,, \label{G_z_teta}%
\end{equation}

\subsection{The Jost's solutions.\label{Section_Jost}}

In order to obtain explicit representations of the operators $\gamma
(\cdot,z,\mathcal{V})$ and $q(z,\mathcal{V})$ appearing at the r.h.s. of
(\ref{krein_1}), it is necessary to define a particular basis of the defect
spaces $\mathcal{N}_{z}$. A possible choice is given in terms of the Green's
function of the operator $\left(  Q_{0,0}(\mathcal{V})-z\right)  $ and of
their derivatives. This motivates the forthcoming analysis, where the
properties of the functions in $\mathcal{N}_{z}$ are investigated by using the
Jost's solutions associated with $Q_{0,0}(\mathcal{V})$. Our aim is to provide
with explicit low and high energy asymptotic in the case of compactly
supported and defined positive potentials. We follow a standard approach
adapting arguments from one dimensional scattering to this particular case.
Detailed computations, for selfadjoint one-dimensional Schr\"{o}dinger
operators with generic short range potentials, are presented in \cite{Yafa}.

Consider the problem%
\begin{equation}
\left(  -\partial_{x}^{2}+\mathcal{V}\right)  u=\zeta^{2}u\,,\qquad\text{for
}x\in\mathbb{R}\ \text{and }\zeta\in\mathbb{C}^{+}\,. \label{Jost_eq}%
\end{equation}
The Jost solutions to (\ref{Jost_eq}), $\chi_{\pm}$, are respectively defined
by the exterior conditions%
\begin{equation}
\left.  \chi_{+}\right\vert _{x>b}=e^{i\zeta x}\,,\qquad\left.  \chi
_{-}\right\vert _{x<a}=e^{-i\zeta x}\,. \label{Jost_sol}%
\end{equation}
The next proposition resumes some properties of the functions $\chi_{\pm}$ in
the case of compactly supported potentials.

\begin{proposition}
\label{Proposition_Jost}Let $\mathcal{V}$ be defined according to (\ref{V}).
The solutions $\chi_{\pm}$ to the problem (\ref{Jost_eq})-(\ref{Jost_sol})
belong to $\mathcal{C}_{x}^{1}\left(  \mathbb{R},\,\mathcal{H}_{\zeta}\left(
\mathbb{C}^{+}\right)  \right)  $ having continuous extension to the real
axis. For $\zeta\in\overline{\mathbb{C}^{+}}$, the relations%
\begin{equation}
\chi_{\pm}\left(  x,\zeta\right)  =e^{\pm i\zeta x}\mathcal{O}\left(
1\right)  \,,\qquad\partial_{x}\chi_{\pm}\left(  x,\zeta\right)  =e^{\pm
i\zeta x}\mathcal{O}\left(  1+\left\vert \zeta\right\vert \right)  \,,
\label{Jost_sol_bound}%
\end{equation}
hold with $\mathcal{O(\cdot)}$ referred to the metric space $\mathbb{R}%
\times\overline{\mathbb{C}^{+}}$.
\end{proposition}

The proof follows, with slightly modifications, the one given in \cite{Yafa}
in the case of 1D short-range potentials: an integral setting for
(\ref{Jost_eq}) and explicit estimates for the corresponding integral kernel
are used to discuss the convergence of the solution developed as a Picard
series. To this aim, we need the following simple Lemma.

\begin{lemma}
\label{Lemma_Jost_est}Let $\mathcal{V}$ be defined according to (\ref{V}) and
$F(x)=\left\vert \int_{x}^{x_{0}}\left\vert \,\mathcal{V}(t)\right\vert
\,dt\right\vert $, with $x_{0}\neq x$. If $f$ is continuous and such that:
$\left\vert f(x)\right\vert \leq\,\frac{F^{n}(x)}{n!}$ for $n\in\mathbb{N}$,
then it results%
\begin{equation}
\left\vert \int_{x}^{x_{0}}f(t)\mathcal{V}(t)\,dt\right\vert \leq\frac
{F^{n+1}(x)}{(n+1)!}\,. \label{G_ineq}%
\end{equation}

\end{lemma}

\begin{proof}
For $x_{0}>x$ we have $F(x)=\int_{x}^{x_{0}}\left\vert \,\mathcal{V}%
(t)\right\vert \,dt$, we get: $\partial_{x}F(x)=-\left\vert \mathcal{V}%
(x)\right\vert $. Then%
\[
\left\vert \int_{x}^{x_{0}}f(t)\mathcal{V}(t)\,dt\right\vert \leq-\int
_{x}^{x_{0}}\left\vert f(t)\right\vert \,\partial_{t}F(t)\,dt\leq-\int
_{x}^{x_{0}}\frac{F^{n}(x)}{n!}\partial_{t}F(t)\,dt=-\int_{x}^{x_{0}}%
\frac{\partial_{t}F^{n+1}(x)}{\left(  n+1\right)  !}\,dt\,.
\]
For $x_{0}<x$, we have $F(x)=\int_{x_{0}}^{x}\left\vert \,\mathcal{V}%
(t)\right\vert \,dt$ and $\partial_{x}F(x)=\left\vert \mathcal{V}%
(x)\right\vert $. Then%
\[
\left\vert \int_{x_{0}}^{x}f(t)\mathcal{V}(t)\,dt\right\vert \leq\int_{x_{0}%
}^{x}\left\vert f(t)\right\vert \,\partial_{t}F(t)\,dt\leq\int_{x_{0}}%
^{x}\frac{F^{n}(x)}{n!}\partial_{t}F(t)\,dt=\int_{x_{0}}^{x}\frac{\partial
_{t}F^{n+1}(x)}{\left(  n+1\right)  !}\,dt\,.
\]
Both the above relations imply (\ref{G_ineq}) since, by definition,
$F(x_{0})=0$.
\end{proof}

\begin{proof}
[Proof of the Proposition \ref{Proposition_Jost}]Here we focus on the case of
$\chi_{+}$, while the problem for $\chi_{-}$ can be analyzed similarly. Using
the integral kernel: $\left.  -\left(  \zeta\right)  ^{-1}\sin\zeta\left(
t-x\right)  \right.  $, the equation (\ref{Jost_eq}) rephrases into the
equivalent integral form%
\begin{equation}
u(x,\zeta)=u_{0}(x,\zeta)-%
{\displaystyle\int\limits_{x_{0}}^{x}}
\frac{\sin\zeta\left(  t-x\right)  }{\zeta}\mathcal{V}(t)u(t,\zeta)\,dt\,.
\label{Int_Jost_eq}%
\end{equation}
In order to account for the conditions (\ref{Jost_sol}), we replace in
(\ref{Int_Jost_eq}): $x_{0}=b$, $x<b$ and $u_{0}=e^{i\zeta x}$. Then,
introducing the rescaled functions: $b_{+}=e^{-i\zeta x}\chi_{+}$ and%
\begin{equation}
\mathcal{K}_{+}\left(  t,x,\zeta\right)  =-e^{i\zeta(t-x)}\frac{\sin
\zeta\left(  t-x\right)  }{\zeta}\,, \label{Jost_eq_kernel}%
\end{equation}
we get the equation%
\begin{equation}
b_{+}(x,\zeta)=1-%
{\displaystyle\int\limits_{x}^{b}}
\mathcal{K}_{+}\left(  t,x,\zeta\right)  \mathcal{V}(t)b_{+}(t,\zeta
)\,dt\,,\quad\text{for }x<b\,, \label{Int_Jost_eq_rescaled}%
\end{equation}
while, for $x>b$ one has: $b_{+}=1$. The corresponding solution formally
writes as a Picard series: $\left.  b_{+}=\sum_{n=0}^{+\infty}b_{+,n}\right.
$ whose terms are defined according to%
\begin{equation}
b_{+,0}=1\,,\qquad b_{+,n}(x,\zeta)=-%
{\displaystyle\int\limits_{x}^{b}}
\mathcal{K}_{+}\left(  t,x,\zeta\right)  \mathcal{V}(t)b_{+,n-1}%
(t,\zeta)\,dt\,,\quad n\in\mathbb{N}^{\ast}\,,\quad x<b\,. \label{Jost_Picard}%
\end{equation}
Let $\eta>0$ and introduce the auxiliary domain $S_{\eta}$ (see the definition
given in the subsection \ref{Sec_Notation}). The rescaled kernel is a smooth
map of $t$ and $x$ with values in $\mathcal{H}_{\zeta}\left(  S_{\eta}\right)
$; then, using (\ref{Jost_Picard}), an induction over $n$ leads to: $\left.
b_{+,n}\in\mathcal{C}_{x}^{1}\left(  (-\infty,b),\,\mathcal{H}_{\zeta}\left(
S_{\eta}\right)  \right)  \right.  $. Next, the convergence of the sum
$\sum_{n=0}^{+\infty}b_{+,n}$ is considered, at first, in the case of a
bounded interval $x\in(c,b)$, then in the whole interval $\left(
-\infty,b\right)  $. Finally, the low and high-energy behaviour of $\chi_{+}$
are investigated to obtain the relations in (\ref{Jost_sol_bound}). In what
follows, we assume: $\eta>0$, $c<a$ and use the notation $U_{\eta
,c}=(c,b)\times S_{\eta}$; a direct computation yields%
\begin{equation}
1_{\left(  x,b\right)  }(t)1_{U_{\eta,c}}(x,\zeta)\mathcal{K}_{+}\left(
t,x,\zeta\right)  =\mathcal{O}\left(  \frac{1}{1+\left\vert \zeta\right\vert
}\right)  \,;\quad1_{\left(  x,b\right)  }(t)1_{U_{\eta,c}}(x,\zeta
)\partial_{x}\mathcal{K}_{+}(t,x,\zeta)=\mathcal{O}\left(  1\right)  \,,
\label{Jost_kernel_asympt1}%
\end{equation}
where the symbols $\mathcal{O}(\cdot)$, introduced in the Definition
\ref{Landau_Notation}, here refer to the metric space: $(c,b)^{2}\times
S_{\eta}$. According to (\ref{Jost_kernel_asympt1}), a positive constant
$C_{a,b,c,\eta}$, possibly depending on the data, exists such that%
\begin{equation}
\sup_{\substack{\left\{  x,\zeta\right\}  \in U_{c,\eta}\\t\in(x,b)}%
}\left\vert \mathcal{K}_{+}\left(  t,x,\zeta\right)  \,\right\vert
<C_{a,b,c,\eta}\,,
\end{equation}
Let introduce the rescaled potential: $\mathcal{\tilde{V}}\left(  x\right)
=C_{a,b,c,\eta}\mathcal{V}\left(  x\right)  $ and the function $F(x)=\int
_{x}^{b}\left\vert \,\mathcal{\tilde{V}}(t)\right\vert \,dt$. As a consequence
of Lemma \ref{Lemma_Jost_est}, we have:\newline$\left.  \left\vert
1_{U_{\eta,c}}\,b_{+,n+1}\right\vert \leq1_{U_{\eta,c}}\frac{F^{n+1}}%
{(n+1)!}\right.  $. Since $\left\Vert F\right\Vert _{L^{\infty}(c,b)}%
=C_{a,b,c,\eta}\left\Vert \mathcal{V}\right\Vert _{L^{1}(a,b)}$, this yields
the estimate%
\begin{equation}
\sup_{\left\{  x,\zeta\right\}  \in U_{c}^{+}}\left\vert 1_{U_{\eta,c}%
}\,b_{+,n+1}\right\vert \leq\frac{C_{a,b,c,\eta}^{n+1}\left\Vert
\mathcal{V}\right\Vert _{L^{1}(a,b)}^{n+1}}{(n+1)!}\,, \label{Jost_est_0}%
\end{equation}
and the Picard series uniformly converges to $b_{+}\in\mathcal{C}_{x}%
^{0}\left(  (c,b),\,\mathcal{H}_{\zeta}\left(  S_{\eta}\right)  \right)  $. In
particular, (\ref{Jost_est_0}) implies%
\begin{equation}
\sup_{\left\{  x,\zeta\right\}  \in U_{\eta,c}}\left\vert b_{+}\right\vert
\leq e^{C_{a,b,c,\eta}\left\Vert \mathcal{V}\right\Vert _{L^{1}(a,b)}%
}\,\Rightarrow1_{U_{\eta,c}}b_{+}=\mathcal{O}\left(  1\right)  \,,
\label{Jost_est_1}%
\end{equation}
and, taking into account the definition: $\chi_{+}=e^{i\zeta x}b_{+}$, it
follows%
\begin{equation}
1_{U_{\eta,c}}\left(  x,\zeta\right)  \chi_{+}\left(  x,\zeta\right)
=e^{i\zeta x}\mathcal{O}\left(  1\right)  \,. \label{Jost_est_1_1}%
\end{equation}
Next, consider $\partial_{x}b_{+}=b_{+}^{\prime}$. For $x\in(c,b)$, it
fulfills the equation%
\begin{equation}
b_{+}^{\prime}(x,\zeta)=-%
{\displaystyle\int\limits_{x}^{b}}
\partial_{x}\mathcal{K}_{+}\left(  t,x,\zeta\right)  \mathcal{V}%
(t)b_{+}(t,\zeta)\,dt\,,\quad x\in(c,b)\,. \label{Int_Jost_eq_rescaled1}%
\end{equation}
The regularity of the r.h.s. of (\ref{Int_Jost_eq_rescaled1}) is a consequence
of the properties of $b_{+}$ and of the kernel $\partial_{x}\mathcal{K}_{+}$.
In particular, making use of the above characterization of $b_{+}$, we get:
$b_{+}^{\prime}\in\mathcal{C}_{x}^{0}\left(  (c,b),\,\mathcal{H}_{\zeta
}\left(  S_{\eta}\right)  \right)  $. Moreover, being $b_{+}$ and
$\partial_{x}\mathcal{K}_{+}$ uniformly bounded for $\left\{  x,\zeta\right\}
\in U_{\eta,c}$ and $t\in(x,b)$, the r.h.s. of (\ref{Int_Jost_eq_rescaled1})
is $\mathcal{O}\left(  1\right)  $. Then, a direct computation shows that%
\begin{equation}
1_{U_{\eta,c}}\left(  x,\zeta\right)  \partial_{x}\chi_{+}\left(
x,\zeta\right)  =e^{i\zeta x}\mathcal{O}\left(  1+\left\vert \zeta\right\vert
\right)  \,. \label{Jost_est_2}%
\end{equation}

To discuss the case $x<c$, we notice that, when $x\in(-\infty,a)$, the
solution $b_{+}$ explicitly writes in the form%
\begin{equation}
b_{+}(x,\zeta)=B_{+}(\zeta)+B_{-}(\zeta)e^{-2i\zeta x}\,. \label{Jost_ext}%
\end{equation}
The condition $b_{+}\in\mathcal{C}_{x}^{1}\left(  (c,b),\,\mathcal{H}_{\zeta
}\left(  S_{\eta}\right)  \right)  $ compels the coefficients $B_{\pm}$ to be
holomorphic in $S_{\eta}$ with the only possible exception of the point
$\zeta=0$; this leads: $b_{+}\in\mathcal{C}_{x}^{1}\left(  (-\infty
,a),\,\mathcal{H}_{\zeta}\left(  S_{\eta}\backslash\left\{  0\right\}
\right)  \right)  $. In $\zeta=0$, the maps $\zeta\rightarrow B_{\pm}$ may
diverge, but, in such a case, a compensation between the different
contributions at the r.h.s of (\ref{Jost_ext}) take place to assure the
regularity of $\zeta\rightarrow1_{U_{\eta,c}}b_{+}$ in the origin. Therefore,
$B_{\pm}$ may have, at most, a simple pole in $\zeta=0$ and the conditions%
\begin{equation}
\lim_{\zeta\rightarrow0}\left(  B_{+}(\zeta)+B_{-}(\zeta)e^{-2i\zeta
x}\right)  =c_{0}\,,\qquad\lim_{\zeta\rightarrow0}-2i\zeta B_{-}%
(\zeta)e^{-2i\zeta x}=c_{1}\,, \label{Jost_lowenergy}%
\end{equation}
holds for any $x\in\left(  c,a\right)  $. Since these are independent of $x$,
the function $b_{+}$ can be extended to: $b_{+}\in\mathcal{C}_{x}^{1}\left(
(-\infty,a),\,\mathcal{H}_{\zeta}\left(  S_{\eta}\right)  \right)  $.

To conclude the proof, we need to extend the relations (\ref{Jost_est_1_1}),
(\ref{Jost_est_2}) to the case of $x\in\left(  -\infty,a\right)  $ and
$\zeta\in\overline{\mathbb{C}^{+}}$. According to (\ref{Jost_ext}) and
(\ref{Jost_lowenergy}), for any fixed $\zeta\in\overline{\mathbb{C}^{+}}$, the
functions $b_{+}$ and $b_{+}^{\prime}$ are uniformly bounded w.r.t.
$x\in\left(  -\infty,a\right)  $. In particular, in a neighbourhood
$\mathcal{B}_{1}\left(  0\right)  \cap\overline{\mathbb{C}^{+}}$ of $\zeta=0$
we have%
\begin{equation}
1_{\left(  -\infty,a\right)  }(x)1_{\mathcal{B}_{1}\left(  0\right)
\cap\overline{\mathbb{C}^{+}}}(\zeta)\partial_{x}^{i}b_{+}\left(
x,\zeta\right)  =\mathcal{O}\left(  1\right)  \,,\qquad i=0,1\,.
\label{Jost_in_bound}%
\end{equation}
To obtain estimates as $\left\vert \zeta\right\vert \rightarrow\infty$, the
high energy asymptotics of the coefficients $B_{\pm}$ is needed. The
$x$-derivative of (\ref{Jost_ext}) is%
\begin{equation}
b_{+}^{\prime}(x,\zeta)=-2i\zeta B_{-}(\zeta)e^{-2i\zeta x}\,.
\label{Jost_ext_1}%
\end{equation}
As it has been previously shown, for $\left\{  x,\zeta\right\}  \in U_{\eta
,c}$ it results $b_{+}^{\prime}(x,\zeta)=\mathcal{O}\left(  1\right)  $.
Taking $x\in\left(  c,a\right)  $, and using (\ref{Jost_ext_1}), this implies:
$\zeta B_{-}(\zeta)=\mathcal{O}\left(  1\right)  $ in the sense of the metric
space $S_{\eta}$; it follows%
\begin{equation}
\zeta B_{-}(\zeta)=\mathcal{O}\left(  1\right)  \qquad\text{in }\zeta
\in\overline{\mathbb{C}^{+}}\backslash\mathcal{B}_{1}\left(  0\right)  \,.
\label{Jost_coeff}%
\end{equation}
Similarly, since $b_{+}(x,\zeta)=\mathcal{O}\left(  1\right)  $ for $\left\{
x,\zeta\right\}  \in U_{\eta,c}$, taking $x\in\left(  c,a\right)  $ and using
the relations (\ref{Jost_ext}) and (\ref{Jost_coeff}), we get%
\begin{equation}
B_{+}(\zeta)=\mathcal{O}\left(  1\right)  \,\qquad\text{in }\zeta\in
\overline{\mathbb{C}^{+}}\backslash\mathcal{B}_{1}\left(  0\right)  \,.
\label{Jost_coeff1}%
\end{equation}
From these relations and the representations (\ref{Jost_ext}) and
(\ref{Jost_ext_1}), we obtain%
\begin{equation}
1_{\left(  -\infty,a\right)  }(x)1_{\overline{\mathbb{C}^{+}}\backslash
\mathcal{B}_{1}\left(  0\right)  }(\zeta)\partial_{x}^{i}b_{+}\left(
x,\zeta\right)  =\mathcal{O}\left(  1\right)  \,,\qquad i=1,2\,.
\label{Jost_ext_bound}%
\end{equation}
Then, taking into account the definition $\chi_{+}=e^{i\zeta x}b_{+}$, the
relations (\ref{Jost_sol_bound}) follows from (\ref{Jost_est_1_1}),
(\ref{Jost_est_2}), (\ref{Jost_ext_bound}) and (\ref{Jost_ext_bound}).
\end{proof}

The Jost function, denoted in the following with $w(\zeta)$, is defined as the
Wronskian associated with the couple $\left\{  \chi_{+}(\cdot,\zeta),\chi
_{-}(\cdot,\zeta)\right\}  $. Setting%
\begin{equation}
w(f,g)=fg^{\prime}-f^{\prime}g\,, \label{Wronskian}%
\end{equation}
we have%
\begin{equation}
w(\zeta)=\chi_{+}(\cdot,\zeta)\partial_{1}\chi_{-}(\cdot,\zeta)-\partial
_{1}\chi_{+}(\cdot,\zeta)\chi_{-}(\cdot,\zeta) \label{Jost_fun}%
\end{equation}
According to the definition of $\chi_{\pm}$, this function is independent of
the space variable, while due to the result of Proposition
\ref{Proposition_Jost}, $w(\zeta)$ is holomorphic w.r.t. $\zeta$ in an open
half complex plane including $\overline{\mathbb{C}^{+}}$. The point spectrum
of $Q_{0,0}(\mathcal{V})$ is defined by the solutions $z=\zeta^{2}$ to the
problem: $w(\zeta)=0\,$, $\zeta\in\mathbb{C}^{+}$ (e.g. in \cite{Yafa},
Chp.5). Since $Q_{0,0}(\mathcal{V})$ is a selfadjoint Schr\"{o}dinger operator
with a short range potential, the point spectrum is non-degenerate and located
on the negative real axis, while: $\sigma_{ac}\left(  Q_{0,0}(\mathcal{V}%
)\right)  =\left[  0,+\infty\right)  $. Then, $w\left(  \zeta\right)  $ does
not annihilates almost everywhere in the closed upper complex plane, with the
only possible exceptions of a discrete subset of the positive imaginary axis.
Next, consider $\zeta=k\in\mathbb{R}$ and let $w_{0}(k)$ be the Wronskian
associated with $\left\{  \chi_{+}(\cdot,-k),\chi_{-}(\cdot,k)\right\}  $%
;$\ $the behavior of $w(k)$ on the real axis follows by using the relations%
\begin{align}
\chi_{+}(\cdot,k)  &  =\frac{1}{2ik}\left(  w_{0}^{\ast}(k)\chi_{-}%
(\cdot,k)-w(k)\chi_{-}(\cdot,-k)\right)  \,,\label{Jost_sol_identity_1}\\
& \nonumber\\
\chi_{-}(\cdot,k)  &  =\frac{1}{2ik}\left(  w_{0}(k)\chi_{+}(\cdot
,k)-w(k)\chi_{+}(\cdot,-k)\right)  \,, \label{Jost_sol_identity_2}%
\end{align}
expressing the Jost's solutions $\chi_{\pm}(\cdot,k)$ in terms of the linearly
independent couples $\chi_{-}(\cdot,\pm k)$ and $\chi_{+}(\cdot,\pm k)$
respectively (e.g. in \cite{Yafa}, chp. 5). Plugging
(\ref{Jost_sol_identity_2}) into (\ref{Jost_sol_identity_1}), indeed, it
follows: $\left\vert w(k)\right\vert ^{2}=4k^{2}\,+\left\vert w_{0}\left(
k\right)  \right\vert ^{2}$, which entails%
\begin{equation}
\left\vert w(k)\right\vert ^{2}\geq4k^{2}\,. \label{w_ineq}%
\end{equation}

Let introduce the functions $\mathcal{G}^{z}(x,y)$ and $\mathcal{H}^{z}(x,y)$%
\begin{equation}
\mathcal{G}^{z}(x,y)=\frac{1}{w(\zeta)}\mathcal{\,}\left\{
\begin{array}
[c]{c}%
\chi_{+}(x,\zeta)\chi_{-}(y,\zeta)\,,\qquad x\geq y\,,\\
\\
\chi_{-}(x,\zeta)\chi_{+}(y,\zeta)\,,\qquad x<y\,,
\end{array}
\right.  \qquad z=\zeta^{2}\,, \label{G_z}%
\end{equation}%
\begin{equation}
\mathcal{H}^{z}(x,y)=-\frac{1}{w(\zeta)}\mathcal{\,}\left\{
\begin{array}
[c]{c}%
\chi_{+}(x,\zeta)\partial_{1}\chi_{-}(y,\zeta)\,,\qquad x\geq y\,,\\
\\
\chi_{-}(x,\zeta)\partial_{1}\chi_{+}(y,\zeta)\,,\qquad x<y\,,
\end{array}
\right.  \,,\qquad z=\zeta^{2}\,, \label{H_z}%
\end{equation}
Assume $\zeta\in\mathbb{C}^{+}$ to be such that $w(\zeta)\neq0$ and
$y\in\mathbb{R}$; from the equation (\ref{Jost_eq}) and the relations
(\ref{Jost_sol_bound}), it follows that the maps $x\rightarrow\mathcal{G}%
^{z}(\cdot,y)$ and $x\rightarrow\mathcal{H}^{z}(\cdot,y)$ are exponentially
decreasing as $\left\vert x-y\right\vert \rightarrow\infty$ (with a decreasing
rate depending on $\operatorname{Im}\zeta$) and fulfill the boundary condition
problems%
\begin{equation}
\left\{
\begin{array}
[c]{lll}%
\left(  -\partial_{x}^{2}+\mathcal{V-}\zeta^{2}\right)  \mathcal{G}^{z}%
(\cdot,y)=0 &  & \text{in }\mathbb{R}/\left\{  y\right\} \\
&  & \\
\mathcal{G}^{z}(y^{+},y)=\mathcal{G}^{z}(y^{-},y)\,, &  & \partial
_{1}\mathcal{G}^{z}(y^{+},y)-\partial_{1}\mathcal{G}^{z}(y^{-},y)=-1\,,
\end{array}
\right.  \label{Green_eq}%
\end{equation}
and%
\begin{equation}
\left\{
\begin{array}
[c]{lll}%
\left(  -\partial_{x}^{2}+\mathcal{V-}\zeta^{2}\right)  \mathcal{H}^{z}%
(\cdot,y)=0 &  & \text{in }\mathbb{R}/\left\{  y\right\} \\
&  & \\
\mathcal{H}^{z}(y^{+},y)-\mathcal{H}^{z}(y^{-},y)=1\,, &  & \partial
_{1}\mathcal{H}^{z}(y^{+},y)=\partial_{1}\mathcal{H}^{z}(y^{-},y)\,,
\end{array}
\right.  \label{D_Green_eq}%
\end{equation}
For $z=\zeta^{2}$ s.t. $w(\zeta)\neq0$ and $y\in\left\{  a,b\right\}  $, the
functions $\mathcal{G}^{z}(\cdot,y)$, $\mathcal{H}^{z}(\cdot,y)$ form a basis
of the defect space $\mathcal{N}_{z}$, which writes as%
\begin{equation}
\mathcal{N}_{z}=l.c.\left\{  \mathcal{G}^{z}(x,b)\,,\ \mathcal{H}%
^{z}(x,b)\,,\ \mathcal{G}^{z}(x,a)\,,\ \mathcal{H}^{z}(x,a)\right\}  \,.
\label{defect1}%
\end{equation}
According to the equation (\ref{Green_eq}), $\mathcal{G}^{z}(\cdot,y)$
identifies with the integral kernel of $\left(  Q_{0,0}(\mathcal{V})-z\right)
^{-1}$,while, as a consequence of the definitions (\ref{G_z})-(\ref{H_z}) and
the results of the Proposition \ref{Proposition_Jost}, the maps $z\rightarrow
\mathcal{G}^{z}(x,y)$, $z\rightarrow\mathcal{H}^{z}(x,y)$ are meromorphic in
$\mathbb{C}\backslash\mathbb{R}_{+}$ with a branch cut along the positive real
axis and poles, corresponding to the points in $\sigma_{p}\left(
Q_{0,0}(\mathcal{V})\right)  $, located on the negative real axis. In
particular, due to the inequality (\ref{w_ineq}), these functions continuously
extend up to the branch cut, both in the limits: $z\rightarrow k^{2}\pm i0$,
with the only possible exception of the point $z=0$.

In the case of defined positive potentials, it is possible to obtain uniform
estimates of $\mathcal{G}^{z}(x,y)$\ and $\mathcal{H}^{z}(x,y)$ up to the
whole branch cut. Next, we assume $\mathcal{V}$ to fulfill the additional
condition%
\begin{equation}
\left\langle u,\mathcal{V\,}u\right\rangle _{L^{2}(a,b)}>0\qquad\forall\,u\in
L^{2}(\mathbb{R})\,, \label{V_pos}%
\end{equation}
and introduce, for $\zeta\in\mathbb{C}^{+}$ and $z=\zeta^{2}$, the functions%
\begin{equation}
G^{\zeta}(x,y)=\mathcal{G}^{\zeta^{2}}(x,y)\,;\qquad\partial_{1}^{i}H^{\zeta
}(x,y)=\partial_{1}^{i}\mathcal{H}^{\zeta^{2}}(x,y)\,, \label{GH_zeta}%
\end{equation}
where the notation $\partial^{0}u=u$ is adopted. These are characterized as follows.

\begin{lemma}
\label{Lemma_Green_ker}Let $\mathcal{V}$ fulfill (\ref{V}) and (\ref{V_pos}).
For all $\left(  x,y\right)  \in\mathbb{R}^{2}$, $x\neq y$, the maps
$\zeta\rightarrow G^{\zeta}(x,y)$ and $\zeta\rightarrow\partial_{1}%
^{i}H^{\zeta}(x,y)$, $i=0,1$, defined according to (\ref{GH_zeta}) are
holomorphic in $\mathbb{C}^{+}$ and continuously extend to $\overline
{\mathbb{C}^{+}}$. In particular, for $\zeta=k\in\mathbb{R}$, it results:
$G^{k}\left(  \cdot,y\right)  ,H^{k}\left(  \cdot,y\right)  \in\mathcal{C}%
_{x}^{1}\left(  \mathbb{R}\backslash\left\{  y\right\}  \,,\ \mathcal{C}%
_{k}^{0}\left(  \mathbb{R}\right)  \right)  $, while the relations%
\begin{equation}
G^{\zeta}(x,y)=e^{i\zeta\left\vert x-y\right\vert }\mathcal{O}\left(  \frac
{1}{1+\left\vert \zeta\right\vert }\right)  \,,\quad H^{\zeta}(x,y)=e^{i\zeta
\left\vert x-y\right\vert }\mathcal{O}\left(  1\right)  \,,\quad\partial
_{1}H^{\zeta}(x,y)=e^{i\zeta\left\vert x-y\right\vert }\mathcal{O}\left(
1+\left\vert \zeta\right\vert \right)  \,. \label{Green_ker_bound}%
\end{equation}
hold with $\mathcal{O}\left(  \cdot\right)  $ referred to the metric space
$\mathbb{R}^{2}\times\overline{\mathbb{C}^{+}}$.
\end{lemma}

\begin{proof}
The conditions (\ref{V}), (\ref{V_pos}) and the relation (\ref{w_ineq})
prevent $w\left(  \zeta\right)  $ to have zeroes in $\overline{\mathbb{C}^{+}%
}\backslash\left\{  0\right\}  $. Computing $w\left(  \zeta\right)  $, we have%
\begin{equation}
w(\zeta)=\chi_{+}(a,\zeta)\partial_{1}\chi_{-}(a,\zeta)-\partial_{1}\chi
_{+}(a,\zeta)\chi_{-}(a,\zeta)\,. \label{w_explicit_a}%
\end{equation}
Using the exterior conditions (\ref{Jost_sol}), the coefficients $\partial
_{1}^{j}\chi_{-}(a,\zeta)$, $j=0,1$, are explicitly given by%
\begin{equation}
\chi_{-}(a,\zeta)=e^{-i\zeta a}\,,\qquad\partial_{1}\chi_{-}(a,\zeta)=-i\zeta
e^{-i\zeta a}\,, \label{w_explicit_a1}%
\end{equation}
For $x<b$, the function $\chi_{+}(\cdot,\zeta)$ writes as $\chi_{+}%
(x,\zeta)=e^{i\zeta x}b_{+}(x,\zeta)$, where $b_{+}(\cdot,\zeta)$ solves the
equation (\ref{Int_Jost_eq_rescaled}) and can be represented as the sum of the
Picard series: $\left.  b_{+}(\cdot,\zeta)=\sum_{n=0}^{+\infty}b_{+,n}%
(\cdot,\zeta)\right.  $ whose terms are defined by a recurrence relation given
in (\ref{Jost_Picard}). Under the condition (\ref{V}), it has been shown that
this series uniformly converges to $b_{+}\in\mathcal{C}_{x}^{1}\left(  \left(
c,b\right)  ,\,\mathcal{H}_{\zeta}\left(  S_{\eta}\right)  \right)  $, being
$\left(  c,b\right)  $ any interval including the point $a$ (see the proof of
the Proposition \ref{Proposition_Jost}); in particular the relations:
$\partial_{1}^{j}b_{+}(\cdot,\zeta)=\mathcal{O}\left(  1\right)  $, $j=0,1$,
hold with the symbols $\mathcal{O}(\cdot)$ referring to the metric space:
$(c,b)\times S_{\eta}$. Let $\zeta=0$; the relations (\ref{Jost_Picard}) write
as%
\begin{equation}
b_{+,n}(x,0)=%
{\displaystyle\int\limits_{x}^{b}}
(t-x)\mathcal{V}(t)b_{+,n-1}(t,0)\,dt\,,\quad n\in\mathbb{N}^{\ast}\,,\ x<b\,.
\end{equation}
Using the conditions: $\left\langle u,\mathcal{V}u\right\rangle _{L^{2}%
(a,b)}>0$ and $b_{+,0}=1$, an induction argument leads to: $b_{+,n}(x,0)\geq0$
and $b_{+}(x,0)\geq1$. Taking the limit of (\ref{Int_Jost_eq_rescaled1}) as
$\zeta\rightarrow0$, we get%
\begin{equation}
\partial_{x}b_{+}(x,0)=-%
{\displaystyle\int\limits_{x}^{b}}
\mathcal{V}(t)b_{+}(t,0)\,dt\,,\qquad x<b\,.
\end{equation}
Since $b_{+}>0$ and $\mathcal{V}>0$, at least in a subset of $(a,b)$, we have:
$\partial_{1}b_{+}(a,0)<0\,$. With the notation introduced above, the equation
(\ref{w_explicit_a}) rephrases as%
\begin{equation}
w(\zeta)=-\partial_{1}b_{+}(a,\zeta)-2i\zeta\,b_{+}(a,\zeta)\,.
\label{w_explicit}%
\end{equation}
Then, according to the conditions: $\partial_{1}b_{+}(a,0)\neq0$, and
$b_{+}(\cdot,\zeta)=\mathcal{O}\left(  1\right)  $, we have: $w(\zeta
)=-\partial_{1}b_{+}(a,\zeta)+\mathcal{O}(\zeta)$ which implies $w\left(
0\right)  \neq0$.

As a consequence, the function $p\left(  \zeta\right)  $ defined by%
\begin{equation}
p\left(  \zeta\right)  =\frac{1+\zeta}{w\left(  \zeta\right)  }\,,
\end{equation}
is bounded in any bounded set $\mathcal{B}_{R}\left(  0\right)  \cap
\overline{\mathbb{C}^{+}}$, $R>0$. Moreover, using the equation
(\ref{Int_Jost_eq_rescaled}), we have%
\begin{equation}
b_{+}(a,\zeta)=1+%
{\displaystyle\int\limits_{a}^{b}}
\frac{e^{2i\zeta(t-x)}-1}{\zeta}\mathcal{V}(t)b_{+}(t,\zeta)\,dt\,;
\end{equation}
since $b_{+}(\cdot,\zeta)$ is uniformly bounded for $\zeta\in\overline
{\mathbb{C}^{+}}$, and $\left\vert 1/\zeta\right\vert \left\vert
e^{2i\zeta(t-x)}-1\right\vert \leq2/\left\vert \zeta\right\vert $, it follows:
\newline$\lim\nolimits_{\zeta\rightarrow\infty\,,\ \zeta\in\overline
{\mathbb{C}^{+}}}b_{+}(a,\zeta)=1$. Set $M=\sup_{\zeta\in\overline
{\mathbb{C}^{+}}}\left\vert \partial_{1}b_{+}(a,\zeta)\right\vert $, and let
$\tilde{R}>0$ be such that $\left\vert \zeta\right\vert \left\vert
b_{+}(a,\zeta)\right\vert /M>1$ for any $\zeta\in\overline{\mathbb{C}^{+}%
}\backslash\mathcal{B}_{\tilde{R}}\left(  0\right)  $. From the representation
(\ref{w_explicit}) it follows
\begin{equation}
\sup_{\zeta\in\overline{\mathbb{C}^{+}}\backslash\mathcal{B}_{\tilde{R}%
}\left(  0\right)  }\left\vert p\left(  \zeta\right)  \right\vert \leq\frac
{1}{M}\frac{1+\left\vert \zeta\right\vert }{2\left\vert \zeta\right\vert
\frac{\left\vert b_{+}(a,\zeta)\right\vert }{M}-1\,}\lesssim1\,.
\end{equation}
Then, $p\left(  \zeta\right)  $ results uniformly bounded as $\zeta
\in\overline{\mathbb{C}^{+}}$ and we can write%
\begin{equation}
\left(  w\left(  \zeta\right)  \right)  ^{-1}=\mathcal{O}\left(  \frac
{1}{1+\left\vert \zeta\right\vert }\right)  \,,\label{w_inverse_bound}%
\end{equation}
in the sense of the metric space $\overline{\mathbb{C}^{+}}$ (see Definition
\ref{Landau_Notation}).

From the definitions (\ref{G_z})-(\ref{H_z}), (\ref{GH_zeta}) and the result
of Proposition \ref{Proposition_Jost}, the functions $\zeta\rightarrow
G^{\zeta}(x,y)$ and $\zeta\rightarrow\partial_{1}^{i}H^{\zeta}(x,y)$, $i=0,1$,
are meromorphic in $\mathbb{C}^{+}$, while, the previous result implies that,
in our assumptions, these maps have no poles in $\mathbb{C}^{+}$ and
continuously extend to the whole real axis. The relations
(\ref{Green_ker_bound}) follows from (\ref{Jost_sol_bound}) by using
(\ref{w_inverse_bound}).
\end{proof}

\subsection{Resolvent analysis.}

The results of the previous Sections and, in particular, the Krein's-like
formula given in (\ref{krein_1}), allow a detailed resolvent analysis for the
operators $Q_{\theta_{1},\theta_{2}}(\mathcal{V})$. At this concern, we recall
that the maps $z\rightarrow q(z,\mathcal{V})$ and $z\rightarrow\gamma
(e_{i},z,\mathcal{V})$, appearing at the r.h.s. of (\ref{krein_1}), are
holomorphic in $\mathbb{C}\backslash\sigma\left(  Q_{0,0}(\mathcal{V})\right)
$, while, from the definitions (\ref{ab_teta1,2})-(\ref{ab_teta1,2}), the
matrix coefficients in $A_{\theta_{1},\theta_{2}}$ and $B_{\theta_{1}%
,\theta_{2}}$ are holomorphic functions of the parameters $\left(  \theta
_{1},\theta_{2}\right)  $ in the whole $\mathbb{C}^{2}$. Then%
\begin{equation}
d(z,\theta_{1},\theta_{2})=\det\left(  B_{\theta_{1},\theta_{2}}%
\,q(z,\mathcal{V})-A_{\theta_{1},\theta_{2}}\right)  \,,\label{d_z_teta}%
\end{equation}
defines an holomorphic function of the variables $\left(  z,\theta_{1}%
,\theta_{2}\right)  $ in $\mathbb{C}\backslash\sigma\left(  Q_{0,0}%
(\mathcal{V})\right)  \times\mathbb{C}^{2}$. Moreover, for any couple $\left(
\theta_{1},\theta_{2}\right)  $, the set of singular points%
\begin{equation}
\mathcal{S}_{\theta_{1},\theta_{2}}=\left\{  z\in\mathbb{C\,}\left\vert
\ d(z,\theta_{1},\theta_{2})=0\right.  \right\}  \,,\label{S_teta}%
\end{equation}
is discrete. As a consequence, the representation (\ref{krein_1}) makes sense
in the dense open set \newline$\mathbb{C}\backslash\left(  \sigma\left(
Q_{0,0}(\mathcal{V})\right)  \cup\mathcal{S}_{\theta_{1},\theta_{2}}\right)
$. Let us fix $z\in\mathbb{C}\backslash\left(  \sigma\left(  Q_{0,0}%
(\mathcal{V})\right)  \cup\mathcal{S}_{\tilde{\theta}_{1},\tilde{\theta}_{2}%
}\right)  $, for a given couple $\left(  \tilde{\theta}_{1},\tilde{\theta}%
_{2}\right)  \in\mathbb{C}^{2}$; using the expansion%
\begin{equation}
d(z,\theta_{1},\theta_{2})=d(z,\tilde{\theta}_{1},\tilde{\theta}%
_{2})+\mathcal{O}\left(  \theta_{1}-\tilde{\theta}_{1}\right)  +\mathcal{O}%
\left(  \theta_{2}-\tilde{\theta}_{2}\right)  \,,
\end{equation}
it results: $d(z,\theta_{1},\theta_{2})\neq0$ for all $\left(  \theta
_{1},\theta_{2}\right)  $ in a suitable neighbourhood of $\left(
\tilde{\theta}_{1},\tilde{\theta}_{2}\right)  $. This implies that, for any
couple of parameters $\left(  \tilde{\theta}_{1},\tilde{\theta}_{2}\right)  $,
there exists $z\in\mathcal{\rho}\left(  Q_{\tilde{\theta}_{1},\tilde{\theta
}_{2}}(\mathcal{V})\right)  $ and a positive constant $\delta$, possibly
depending on $\left(  \tilde{\theta}_{1},\tilde{\theta}_{2}\right)  $, such
that: $z\in\mathcal{\rho}\left(  Q_{\theta_{1},\theta_{2}}(\mathcal{V}%
)\right)  $ for all $\left(  \theta_{1},\theta_{2}\right)  \in\mathcal{B}%
_{\delta}\left(  \left(  \tilde{\theta}_{1},\tilde{\theta}_{2}\right)
\right)  $. Next, for such a $z$, consider the map $\left(  \theta_{1}%
,\theta_{2}\right)  \rightarrow\left(  Q_{\theta_{1},\theta_{2}}%
(\mathcal{V})-z\right)  ^{-1}$ defined for $\left(  \theta_{1},\theta
_{2}\right)  \in\mathcal{B}_{\delta}\left(  \left(  \tilde{\theta}_{1}%
,\tilde{\theta}_{2}\right)  \right)  $. Since $z\notin\mathcal{S}_{\theta
_{1},\theta_{2}}$, the coefficients of the finite rank part at the r.h.s. of
(\ref{krein_1}) are holomorphic w.r.t. $\left(  \theta_{1},\theta_{2}\right)
$ and $\left(  Q_{\theta_{1},\theta_{2}}(\mathcal{V})-z\right)  ^{-1}$ forms
an analytic family in $\mathcal{L}\left(  L^{2}\left(  \mathbb{R}\right)
\right)  $. Then, $Q_{\theta_{1},\theta_{2}}(\mathcal{V})$ is analytic in the
sense of Kato, w.r.t. the parameters $\left(  \theta_{1},\theta_{2}\right)  $.

As this result suggests, when $\left(  \theta_{1},\theta_{2}\right)  $ is
close to the origin of $\mathbb{C}^{2}$, a part of the point spectrum
$\sigma_{p}\left(  Q_{\theta_{1},\theta_{2}}(\mathcal{V})\right)  $ is formed
by non-degenerate eigenvalues holomorphically dependent on $\left(  \theta
_{1},\theta_{2}\right)  $ and converging, in the limit $\left(  \theta
_{1},\theta_{2}\right)  \rightarrow\left(  0,0\right)  $, to the corresponding
points of $\sigma_{p}\left(  Q_{0,0}(\mathcal{V})\right)  $ (see the point
$(ii)$ in the next Proposition \ref{Proposition_spectrum}). As an aside we
notice that, for generic compactly supported potentials, new spectral points
(not converging to $\sigma_{p}\left(  Q_{0,0}(\mathcal{V})\right)  $) may
eventually arise in a complex neighbourhood of the origin, due to the
interface conditions. Nevertheless, if the additional assumption of positive
potentials (\ref{V_pos}) is adopted, it is possible to prove the identity
$\sigma\left(  Q_{\theta_{1},\theta_{2}}(\mathcal{V})\right)  =\sigma\left(
Q_{0,0}(\mathcal{V})\right)  $ provided that $\theta_{i=1,2}$ are small
enough. To fix this point, we need appropriate estimates for the coefficients
of the finite rank part in (\ref{krein_1}).

The relations (\ref{krein_1}) and (\ref{G_z_teta}) can be made explicit by
computing the matrix representation of $\left(  B_{\theta_{1},\theta_{2}%
}\,q(z,\mathcal{V})-A_{\theta_{1},\theta_{2}}\right)  $ w.r.t. the basis
$\left\{  e_{j}\right\}  _{j=1}^{4}$ and (\ref{defect1}). Making use of the
definition (\ref{gamma_q_def}), a direct computation yields%
\begin{equation}
\gamma(\cdot,z,\mathcal{V})=%
\begin{pmatrix}
1 &  &  & \\
& -1 &  & \\
&  & 1 & \\
&  &  & -1
\end{pmatrix}
\,,\quad\text{with:\ }\left\{
\begin{array}
[c]{l}%
\gamma(e_{1},z,\mathcal{V})=\mathcal{G}^{z}(x,b)\,;\quad\gamma(e_{2}%
,z,\mathcal{V})=-\mathcal{H}^{z}(x,b)\,;\\
\\
\gamma(e_{3},z,\mathcal{V})=\mathcal{G}^{z}(x,a)\,;\quad\gamma(e_{4}%
,z,\mathcal{V})=-\mathcal{H}^{z}(x,a)\,.
\end{array}
\right.  \label{gamma_z}%
\end{equation}
The matrix coefficients of $q(z,\mathcal{V})$ are related to the boundary
values of the functions $\gamma(e_{i},z,\mathcal{V})$, $i=1...4$, as
$x\rightarrow b^{\pm}$ or $x\rightarrow a^{\pm}$. Using the definitions
(\ref{G_z}) and (\ref{H_z}), it follows: $\partial_{1}\mathcal{G}%
^{z}(x,y)=\mathcal{H}^{z}(y,x)\,$; in particular, the boundary values at
$x\rightarrow y^{\pm}$ are related by%
\begin{equation}
\partial_{1}\mathcal{G}^{z}(y^{\pm},y)=\mathcal{H}^{z}(y^{\mp},y)\,.
\end{equation}
Using these relations and the boundary conditions in (\ref{Green_eq}%
)-(\ref{D_Green_eq}), a direct computation yields%
\begin{equation}
q(z,\mathcal{V})=%
\begin{pmatrix}
\mathcal{G}^{z}(b,b)\medskip & \frac{1}{2}-\mathcal{H}^{z}(b^{+},b) &
\mathcal{G}^{z}(b,a) & -\mathcal{H}^{z}(b,a)\\
\mathcal{H}^{z}(b^{+},b)\medskip-\frac{1}{2} & -\partial_{1}\mathcal{H}%
^{z}(b,b) & \mathcal{H}^{z}(a,b) & -\partial_{1}\mathcal{H}^{z}(b,a)\\
\mathcal{G}^{z}(a,b)\medskip & -\mathcal{H}^{z}(a,b) & \mathcal{G}^{z}(a,a) &
-\left(  \frac{1}{2}+\mathcal{H}^{z}(a^{-},a)\right) \\
\mathcal{H}^{z}(b,a) & -\partial_{1}\mathcal{H}^{z}(a,b) & \mathcal{H}%
^{z}(a^{-},a)+\frac{1}{2} & -\partial_{1}\mathcal{H}^{z}(a,a)
\end{pmatrix}
\,. \label{q_z}%
\end{equation}

\begin{lemma}
\label{Lemma_Krein_coeff}Let the matrix $\left(  B_{\theta_{1},\theta_{2}%
}\,q(z,\mathcal{V})-A_{\theta_{1},\theta_{2}}\right)  $ be defined according
to the relations (\ref{q_z}), (\ref{AB_teta1,2})-(\ref{ab_teta1,2}) and assume
$\mathcal{V}$ to fulfill the conditions (\ref{V}), (\ref{V_pos}). There exists
$\delta>0$ such that, for all $\left(  \theta_{1},\theta_{2}\right)
\in\mathcal{B}_{\delta}\left(  \left(  0,0\right)  \right)  $, $\left(
B_{\theta_{1},\theta_{2}}\,q(z,\mathcal{V})-A_{\theta_{1},\theta_{2}}\right)
$ is invertible in $z\in\mathbb{C}/\mathbb{R}_{+}$. The coefficients of the
inverse matrix are holomorphic w.r.t. $\left(  z,\theta_{1},\theta_{2}\right)
$ with: $\left(  \theta_{1},\theta_{2}\right)  \in\mathcal{B}_{\delta}\left(
\left(  0,0\right)  \right)  $, $z\in\mathbb{C}/\mathbb{R}_{+}$; they have
continuous extensions to the branch cut both in the limits $z=k^{2}%
+i\varepsilon$, $\varepsilon\rightarrow0^{\pm}$.
\end{lemma}

\begin{proof}
Using the notation introduced in (\ref{GH_zeta}), for $z=\zeta^{2}$, $\zeta
\in\mathbb{C}^{+}$, a direct computation leads to%
\begin{gather}
\left(  B_{\theta_{1},\theta_{2}}\,q(z,\mathcal{V})-A_{\theta_{1},\theta_{2}%
}\right)  =\qquad\qquad\qquad\qquad\qquad\qquad\qquad\qquad\qquad\qquad
\qquad\qquad\qquad\qquad\qquad\qquad\qquad\qquad\qquad\label{Krein_coeff_0}\\%
\begin{pmatrix}
\mathcal{\beta}\left(  \theta_{2}\right)  \left(  H^{\zeta}(b^{+},b)-\frac
{1}{2}\right)  & -\mathcal{\beta}\left(  \theta_{2}\right)  \partial
_{1}H^{\zeta}(b,b) & \mathcal{\beta}\left(  \theta_{2}\right)  H^{\zeta
}(a,b) & -\mathcal{\beta}\left(  \theta_{2}\right)  \partial_{1}H^{\zeta
}(b,a)\\
\mathcal{\beta}\left(  \theta_{1}\right)  G^{\zeta}(b,b) & \mathcal{\beta
}\left(  \theta_{1}\right)  \left(  \frac{1}{2}-H^{\zeta}(b^{+},b)\right)  &
\mathcal{\beta}\left(  \theta_{1}\right)  G^{\zeta}(b,a) & -\mathcal{\beta
}\left(  \theta_{1}\right)  H^{\zeta}(b,a)\\
\mathcal{\beta}\left(  -\theta_{2}\right)  H^{\zeta}(b,a) & -\mathcal{\beta
}\left(  -\theta_{2}\right)  \partial_{1}H^{\zeta}(a,b) & \mathcal{\beta
}\left(  -\theta_{2}\right)  \left(  H^{\zeta}(a^{-},a)+\frac{1}{2}\right)  &
-\mathcal{\beta}\left(  -\theta_{2}\right)  \partial_{1}H^{\zeta}(a,a)\\
\mathcal{\beta}\left(  -\theta_{1}\right)  G^{\zeta}(a,b)\medskip &
-\mathcal{\beta}\left(  -\theta_{1}\right)  H^{\zeta}(a,b) & \mathcal{\beta
}\left(  -\theta_{1}\right)  G^{\zeta}(a,a) & -\mathcal{\beta}\left(
-\theta_{1}\right)  \left(  \frac{1}{2}+H^{\zeta}(a^{-},a)\right)
\end{pmatrix}
\nonumber\\
\qquad\qquad\qquad\qquad\qquad\qquad-%
\begin{pmatrix}
\alpha\left(  \theta_{2}\right)  &  &  & \\
& \alpha\left(  \theta_{1}\right)  &  & \\
&  & \alpha\left(  -\theta_{2}\right)  & \\
&  &  & \alpha\left(  -\theta_{1}\right)
\end{pmatrix}
\nonumber
\end{gather}
where $\alpha\left(  \theta\right)  $ and $\mathcal{\beta}\left(
\theta\right)  $ are defined by%
\begin{equation}
\mathcal{\alpha}\left(  \theta\right)  =1+e^{\frac{\theta}{2}}\,,\qquad
\mathcal{\beta}\left(  \theta\right)  =1-e^{\frac{\theta}{2}}\,.
\label{alpha_beta}%
\end{equation}
As consequence of the Lemma \ref{Lemma_Green_ker}, for defined positive
potentials the above relation rephrases as%
\begin{align}
&  \left(  B_{\theta_{1},\theta_{2}}\,q(z,\mathcal{V})-A_{\theta_{1}%
,\theta_{2}}\right)  =\label{Krein_coeff}\\
&
\begin{pmatrix}
\mathcal{\beta}\left(  \theta_{2}\right)  \mathcal{O}\left(  1\right)
-\alpha\left(  \theta_{2}\right)  & \mathcal{\beta}\left(  \theta_{2}\right)
\mathcal{O}\left(  1+\left\vert \zeta\right\vert \right)  & \mathcal{\beta
}\left(  \theta_{2}\right)  e^{i\zeta\left(  b-a\right)  }\mathcal{O}\left(
1\right)  & \mathcal{\beta}\left(  \theta_{2}\right)  e^{i\zeta\left(
b-a\right)  }\mathcal{O}\left(  1+\left\vert \zeta\right\vert \right) \\
\mathcal{\beta}\left(  \theta_{1}\right)  \mathcal{O}\left(  \frac
{1}{1+\left\vert \zeta\right\vert }\right)  & \mathcal{\beta}\left(
\theta_{1}\right)  \mathcal{O}\left(  1\right)  -\alpha\left(  \theta
_{1}\right)  & \mathcal{\beta}\left(  \theta_{1}\right)  e^{i\zeta\left(
b-a\right)  }\mathcal{O}\left(  \frac{1}{1+\left\vert \zeta\right\vert
}\right)  & \mathcal{\beta}\left(  \theta_{1}\right)  e^{i\zeta\left(
b-a\right)  }\mathcal{O}\left(  1\right) \\
\mathcal{\beta}\left(  -\theta_{2}\right)  e^{i\zeta\left(  b-a\right)
}\mathcal{O}\left(  1\right)  & \mathcal{\beta}\left(  -\theta_{2}\right)
e^{i\zeta\left(  b-a\right)  }\mathcal{O}\left(  1+\left\vert \zeta\right\vert
\right)  & \mathcal{\beta}\left(  -\theta_{2}\right)  \mathcal{O}\left(
1\right)  -\alpha\left(  -\theta_{2}\right)  & \mathcal{\beta}\left(
-\theta_{2}\right)  \mathcal{O}\left(  1+\left\vert \zeta\right\vert \right)
\\
\mathcal{\beta}\left(  -\theta_{1}\right)  e^{i\zeta\left(  b-a\right)
}\mathcal{O}\left(  \frac{1}{1+\left\vert \zeta\right\vert }\right)  &
\mathcal{\beta}\left(  -\theta_{1}\right)  e^{i\zeta\left(  b-a\right)
}\mathcal{O}\left(  1\right)  & \mathcal{\beta}\left(  -\theta_{1}\right)
\mathcal{O}\left(  \frac{1}{1+\left\vert \zeta\right\vert }\right)  &
\mathcal{\beta}\left(  -\theta_{1}\right)  \mathcal{O}\left(  1\right)
-\alpha\left(  -\theta_{1}\right)
\end{pmatrix}
\,,\nonumber
\end{align}
being the symbols $\mathcal{O}\left(  \cdot\right)  $ referred to the metric
space $\overline{\mathbb{C}^{+}}$ and defining holomorphic functions of
$\zeta\in\mathbb{C}^{+}$ with continuous extension the real axis. Due to the
definition of $\mathcal{\alpha}\left(  \theta\right)  $, $\mathcal{\beta
}\left(  \theta\right)  $, the coefficients of $\left(  B_{\theta_{1}%
,\theta_{2}}\,q(z,\mathcal{V})-A_{\theta_{1},\theta_{2}}\right)  $ result
separately w.r.t. $\left(  \theta_{1},\theta_{2}\right)  \in\mathbb{C}^{2}$,
$z\in\mathbb{C}/\mathbb{R}_{+}$, and admit, for each couple $\left(
\theta_{1},\theta_{2}\right)  $, continuous extensions to the branch cut. In
particular, setting $\zeta=k\in\mathbb{R}_{\pm}$ at the r.h.s. of
(\ref{Krein_coeff}) corresponds to consider the limits of $\left(
B_{\theta_{1},\theta_{2}}\,q(z,\mathcal{V})-A_{\theta_{1},\theta_{2}}\right)
_{ij}$ for $z\rightarrow k^{2}\pm i0$ respectively. Making use of this
expression, and taking into account (\ref{alpha_beta}), a determinant's
expansion follows%
\begin{equation}
d(z,\theta_{1},\theta_{2})=4\left(  1+\cosh\frac{\theta_{1}}{2}\right)
\left(  1+\cosh\frac{\theta_{2}}{2}\right)  +\mathcal{O}\left(  \theta
_{1}\right)  +\mathcal{O}\left(  \theta_{2}\right)  \,, \label{d_exp}%
\end{equation}
where $\mathcal{O}\left(  \theta_{i}\right)  $, being referred to the metric
space $\mathcal{B}_{1}\left(  \left(  0,0\right)  \right)  \times\mathbb{C}$,
defines holomorphic functions w.r.t. $\left(  \theta_{1},\theta_{2}\right)
\in\mathcal{B}_{1}\left(  \left(  0,0\right)  \right)  $ and $z\in
\mathbb{C}/\mathbb{R}_{+}$, allowing continuous extensions to the branch cut
in the above-specified sense. According to the Definition
\ref{Landau_Notation}, $\mathcal{O}\left(  \theta_{i}\right)  $ writes as%
\begin{equation}
\mathcal{O}\left(  \theta_{i}\right)  =\theta_{i}\,p(z,\theta_{1},\theta
_{2})\,,
\end{equation}
with $p(z,\theta_{1},\theta_{2})$ uniformly bounded in $\mathcal{B}_{1}\left(
\left(  0,0\right)  \right)  \times\mathbb{C}$. Therefore $\delta>0$ exists
such that, when $\left(  \theta_{1},\theta_{2}\right)  \in\mathcal{B}_{\delta
}\left(  \left(  0,0\right)  \right)  $, it results $\left\vert d(z,\theta
_{1},\theta_{2})\right\vert >1$ for all $z\in\mathbb{C}/\mathbb{R}_{+}$. In
these conditions, the matrix $\left(  B_{\theta_{1},\theta_{2}}%
\,q(z,\mathcal{V})-A_{\theta_{1},\theta_{2}}\right)  $ is invertible and the
coefficients $\left(  B_{\theta_{1},\theta_{2}}\,q(z,\mathcal{V}%
)-A_{\theta_{1},\theta_{2}}\right)  _{ij}^{-1}$ are separately holomorphic
w.r.t. $\left(  \theta_{1},\theta_{2}\right)  \in\mathcal{B}_{\delta}\left(
\left(  0,0\right)  \right)  $, $z\in\mathbb{C}/\mathbb{R}_{+}$, having
continuous extensions to the whole branch cut, both in the limits
$z=k^{2}+i\varepsilon$, $\varepsilon\rightarrow0^{\pm}$.
\end{proof}

We are now in the position to develop the spectral analysis for the operators
$Q_{\theta_{1},\theta_{2}}(\mathcal{V})$ starting from the resolvent's formula
(\ref{krein_1}).

\begin{proposition}
\label{Proposition_spectrum}Let $Q_{\theta_{1},\theta_{2}}(\mathcal{V})$ be
defined according to (\ref{V}), (\ref{Q_teta}). The operator's spectrum
characterizes as follows\medskip:\newline$i)$ For any $\left(  \theta
_{1},\theta_{2}\right)  \in\mathbb{C}^{2}$, the essential part of the spectrum
is $\sigma_{ess}\left(  Q_{\theta_{1},\theta_{2}}(\mathcal{V})\right)
=\mathbb{R}_{+}$.\medskip\newline$ii)$ Let $E_{0}$ be an eigenvalue of
$Q_{0,0}(\mathcal{V})$ and assume $\varepsilon_{0}>0$ small enough; for any
fixed $\varepsilon\in\left(  0,\varepsilon_{0}\right)  $ it exists
$\delta_{\varepsilon}>0$ depending on $\varepsilon$ s.t.: for any $\left(
\theta_{1},\theta_{2}\right)  \in\mathcal{B}_{\delta_{\varepsilon}}\left(
\left(  0,0\right)  \right)  \cap\mathbb{C}^{2}$ it exists an unique
nondegenerate and discrete eigenvalue $E\left(  \theta_{1},\theta_{2}\right)
\in\sigma\left(  Q_{\theta_{1},\theta_{2}}(\mathcal{V})\right)  \cap
\mathcal{B}_{\varepsilon}(E_{0})$. Moreover, the function $E\left(  \theta
_{1},\theta_{2}\right)  $ is holomorphic w.r.t. $\left(  \theta_{1},\theta
_{2}\right)  $ in $\mathcal{B}_{\delta_{\varepsilon}}\left(  \left(
0,0\right)  \right)  $.

If, in addition, $\mathcal{V}$ is assumed to be defined positive, fulfilling
(\ref{V_pos}), then:\newline$iii)$ It exists $\delta>0$ s.t., for all $\left(
\theta_{1},\theta_{2}\right)  \in\mathcal{B}_{\delta}\left(  \left(
0,0\right)  \right)  $, $\sigma\left(  Q_{\theta_{1},\theta_{2}}%
(\mathcal{V})\right)  $ is purely absolutely continuous and coincide with the
positive real axis.
\end{proposition}

\begin{proof}
With the notation introduced above, let $z\in\mathbb{C}\backslash\left(
\sigma\left(  Q_{0,0}(\mathcal{V})\right)  \cup\mathcal{S}_{\theta_{1}%
,\theta_{2}}\right)  $; the representation (\ref{krein_1}) implies that the
difference $\left(  Q_{\theta_{1},\theta_{2}}(\mathcal{V})-z\right)
^{-1}-\left(  Q_{0,0}(\mathcal{V})-z\right)  ^{-1}$ is a finite rank operator.
Then, the first statement of the Proposition follows by adapting the Weyl's
theorem to the non-selfadjoint framework (for this point, we refer to
\cite{ReSi4}, Sec.~XIII.4, Lemma~3 and the strong spectral mapping theorem).

In the unperturbed case, $Q_{0,0}(\mathcal{V})$ is a 1D Schr\"{o}dinger
operator with a short range potential. Its spectrum has a purely absolutely
continuous part on the positive real axis, and possible non-degenerate
eigenvalues located on the negative real axis, without accumulation points.
Then, the second statement is a direct consequence of the Kato-Rellich
theorem, since $Q_{\theta_{1},\theta_{2}}(\mathcal{V})$ is Kato-analytic
w.r.t. the parameters.

When $\mathcal{V}$ is a defined positive potential, $\sigma\left(
Q_{0,0}(\mathcal{V})\right)  =\sigma_{ac}\left(  Q_{0,0}(\mathcal{V})\right)
=\mathbb{R}_{+}$, while the point spectrum is empty. The spectrum
$\sigma\left(  Q_{\theta_{1},\theta_{2}}(\mathcal{V})\right)  $ corresponds to
the subset of the complex plane where the map $z\rightarrow\mathcal{G}%
_{\theta_{1},\theta_{2}}^{z}(x,y)$ (defined in eq. (\ref{G_z_teta})) is not
holomorphic. According to the result of the Lemma \ref{Lemma_Green_ker}, the
functions $\mathcal{G}^{z}(x,y)$ and $\mathcal{H}^{z}(x,y)$, appearing at the
r.h.s. of (\ref{G_z_teta}) are $z$-holomorphic in $\mathbb{C}/\mathbb{R}_{+}$
and continuously extend to the whole branch cut both in the limits
$z=k^{2}+i\varepsilon$, $\varepsilon\rightarrow0^{\pm}$. As shown in Lemma
\ref{Lemma_Krein_coeff}, the same hold for the coefficients of $\left(
B_{\theta_{1},\theta_{2}}\,q(z,\mathcal{V})-A_{\theta_{1},\theta_{2}}\right)
^{-1}$, provided that $\left(  \theta_{1},\theta_{2}\right)  $ is close enough
to the origin in $\mathbb{C}^{2}$. In particular, these are $z$-holomorphic in
$\mathbb{C}/\mathbb{R}_{+}$ and have continuous extensions to the whole branch
cut. Then, for $\mathcal{V}$ defined positive, the map $z\rightarrow
\mathcal{G}_{\theta_{1},\theta_{2}}^{z}(x,y)$ is holomorphic in $\mathbb{C}%
/\mathbb{R}_{+}$ and have continuous extensions as $z\rightarrow\mathbb{R}%
_{+}$, both in the limits $z=k^{2}+i\varepsilon$, $\varepsilon\rightarrow
0^{\pm}$. This yields: $\sigma\left(  Q_{\theta_{1},\theta_{2}}(\mathcal{V}%
)\right)  =\sigma_{ac}\left(  Q_{\theta_{1},\theta_{2}}(\mathcal{V})\right)
=\mathbb{R}_{+}$.
\end{proof}

\subsection{\label{Section_Resolvent_2}Generalized eigenfunctions expansion.}

Let $\psi_{-}(\cdot,k,\theta_{1},\theta_{2})$ denote the generalized
eigenfunction of the operator $Q_{\theta_{1},\theta_{2}}\left(  \mathcal{V}%
\right)  $, describing an incoming wave function of momentum $k$; this is a
solution to the boundary value problem%
\begin{equation}
\left\{
\begin{array}
[c]{lll}%
\left(  -\partial_{x}^{2}+\mathcal{V}\right)  u=k^{2}u\,, &  & \text{for }%
x\in\mathbb{R}\backslash\left\{  a,b\right\}  \,,\ k\in\mathbb{R}\,,\\
&  & \\
e^{-\frac{\theta_{1}}{2}}u(b^{+},\zeta,\theta_{1},\theta_{2})=u(b^{-}%
,\zeta,\theta_{1},\theta_{2})\,, &  & e^{-\frac{\theta_{2}}{2}}u^{\prime
}(b^{+},\zeta,\theta_{1},\theta_{2})=u^{\prime}(b^{-},\zeta,\theta_{1}%
,\theta_{2})\,,\\
&  & \\
e^{-\frac{\theta_{1}}{2}}u(a^{-},\zeta,\theta_{1},\theta_{2})=u(a^{+}%
,\zeta,\theta_{1},\theta_{2})\,, &  & e^{-\frac{\theta_{2}}{2}}u^{\prime
}(a^{-},\zeta,\theta_{1},\theta_{2})=u(a^{+},\zeta,\theta_{1},\theta_{2})\,,
\end{array}
\right.  \label{Jost_eq_teta}%
\end{equation}
fulfilling the exterior conditions%
\begin{equation}
\psi_{-}(x,k,\theta_{1},\theta_{2})\left\vert _{\substack{x<a\\k>0}}\right.
=e^{ikx}+R(k,\theta_{1},\theta_{2})e^{-ikx}\,,\quad\psi_{-}(x,k,\theta
_{1},\theta_{2})\left\vert _{\substack{x>b\\k>0}}\right.  =T(k,\theta
_{1},\theta_{2})e^{ikx}\,, \label{gen_eigenfun_ext1}%
\end{equation}
and%
\begin{equation}
\psi_{-}(x,k,\theta_{1},\theta_{2})\left\vert _{\substack{x<a\\k<0}}\right.
=T(k,\theta_{1},\theta_{2})e^{ikx}\,,\quad\psi_{-}(x,k,\theta_{1},\theta
_{2})\left\vert _{\substack{x>b\\k<0}}\right.  =e^{ikx}+R(k,\theta_{1}%
,\theta_{2})e^{-ikx}\,, \label{gen_eigenfun_ext2}%
\end{equation}
where $R$ and $T$ are the reflection and transmission coefficients. In the
case $\left(  \theta_{1},\theta_{2}\right)  =\left(  0,0\right)  $, $\psi
_{-}(\cdot,k,0,0)$ is a generalized eigenfunction of the selfadjoint model
$Q_{0,0}\left(  \mathcal{V}\right)  $. In what follows we adopt the simplified
notation: $\psi_{-}(\cdot,k,0,0)=\psi_{-}(\cdot,k)$. These functions are
expressed in terms of the corresponding Jost's solutions as%
\begin{equation}
\psi_{-}(x,k)=\left\{
\begin{array}
[c]{lll}%
-\frac{2ik}{w\left(  k\right)  }\chi_{+}(x,k)\,, &  & \text{for }k\geq0\,,\\
&  & \\
\frac{2ik}{w\left(  -k\right)  }\chi_{-}(x,-k)\,, &  & \text{for }k<0\,,
\end{array}
\right.  \label{gen_eigenfun_0}%
\end{equation}
(e.g. in \cite{Yafa}). In the case of defined positive potentials, an approach
similar to the one leading to the Krein-like resolvent formula (\ref{krein_AB}%
) allows to obtain an expansion for the difference: $\left.  \psi_{-}%
(\cdot,k,\theta_{1},\theta_{2})-\psi_{-}(\cdot,k)\right.  $ for $\left.
\left(  \theta_{1},\theta_{2}\right)  \rightarrow\left(  0,0\right)  \right.
$. To this aim, we need an explicit expression of the finite rank terms,
appearing at the r.h.s. of (\ref{krein_1}), in the limits where $z$ approaches
the branch cut. This can be done by using the results of Lemmas
\ref{Lemma_Green_ker} and \ref{Lemma_Krein_coeff}. Adopting the notation
introduced in (\ref{GH_zeta}), let us define%
\begin{equation}
\left\{  g(e_{i},\zeta,\mathcal{V})\right\}  _{i=1}^{4}=\left\{  G^{\zeta
}(\cdot,b)\,,\ -H^{\zeta}(\cdot,b)\,,\ G^{\zeta}(\cdot,a)\,,\ -H^{\zeta}%
(\cdot,a)\right\}  \,; \label{g_zeta}%
\end{equation}
we get, for $\zeta\in\mathbb{C}^{+}$ and $z=\zeta^{2}$, the identity:
$\gamma(e_{i},z,\mathcal{V})=g(e_{i},\zeta,\mathcal{V})$; due to Lemma
\ref{Lemma_Green_ker}, the limits of $g(e_{i},\zeta,\mathcal{V})$ as
$\zeta\rightarrow k\in\mathbb{R}_{\pm}$ exist and corresponds to the limits of
$\gamma(e_{i},z,\mathcal{V})$ as $z\rightarrow k^{2}\pm i0$ respectively.
Namely, we have%
\begin{equation}
\lim_{z\rightarrow k^{2}\pm i0}\gamma(e_{i},z,\mathcal{V})=\left\{
\begin{array}
[c]{c}%
\left.  g(e_{i},k,\mathcal{V})\right\vert _{k\in\mathbb{R}_{+}}\,,\\
\\
\left.  g(e_{i},k,\mathcal{V})\right\vert _{k\in\mathbb{R}_{-}}\,.
\end{array}
\right.  \label{g_k}%
\end{equation}
The coefficients $\left(  B_{\theta_{1},\theta_{2}}\,q(z,\mathcal{V}%
)-A_{\theta_{1},\theta_{2}}\right)  _{ij}^{-1}$ have been considered in the
Lemma \ref{Lemma_Krein_coeff} where their regularity w.r.t. the $z$ and the
extensions to the branch cut have been investigated. To get further insights
on the structure of the inverse matrix, we use the explicit form of $\left(
B_{\theta_{1},\theta_{2}}\,q(z,\mathcal{V})-A_{\theta_{1},\theta_{2}}\right)
$ given in (\ref{Krein_coeff_0})-(\ref{Krein_coeff}). In what follows,
$\mathcal{M}\left(  \zeta,\theta_{1},\theta_{2}\right)  $ denotes the r.h.s.
of (\ref{Krein_coeff_0})%
\begin{equation}
\left(  B_{\theta_{1},\theta_{2}}\,q(z,\mathcal{V})-A_{\theta_{1},\theta_{2}%
}\right)  =\mathcal{M}\left(  \zeta,\theta_{1},\theta_{2}\right)
\,,\qquad\zeta\in\mathbb{C}^{+},\ z=\zeta^{2}\,. \label{M_zeta_teta}%
\end{equation}
In the assumption (\ref{V_pos}), the matrix $\mathcal{M}\left(  \zeta
,\theta_{1},\theta_{2}\right)  $ continuously extends extends to $\zeta
\in\overline{\mathbb{C}^{+}}$ and taking its limits for $\zeta\rightarrow
k\in\mathbb{R}_{\pm}$ corresponds to consider the limits of $\left(
B_{\theta_{1},\theta_{2}}\,q(z,\mathcal{V})-A_{\theta_{1},\theta_{2}}\right)
_{ij}$ as $z\rightarrow k^{2}\pm i0$ respectively. This yields%
\begin{equation}
\lim_{z\rightarrow k^{2}\pm i0}\left(  B_{\theta_{1},\theta_{2}}%
\,q(z,\mathcal{V})-A_{\theta_{1},\theta_{2}}\right)  =\left\{
\begin{array}
[c]{c}%
\left.  \mathcal{M}\left(  k,\theta_{1},\theta_{2}\right)  \right\vert
_{k\in\mathbb{R}_{+}^{\ast}}\,,\\
\\
\left.  \mathcal{M}\left(  k,\theta_{1},\theta_{2}\right)  \right\vert
_{k\in\mathbb{R}_{-}}\,.
\end{array}
\right.  \label{krein_coeff_1}%
\end{equation}
In particular, making use of (\ref{Krein_coeff}), we have%
\begin{align}
\mathcal{M}\left(  k,\theta_{1},\theta_{2}\right)   &  =\label{M_k_teta}\\
&
\begin{pmatrix}
\mathcal{\beta}\left(  \theta_{2}\right)  \mathcal{O}\left(  1\right)
-\alpha\left(  \theta_{2}\right)  & \mathcal{\beta}\left(  \theta_{2}\right)
\mathcal{O}\left(  1+\left\vert k\right\vert \right)  & \mathcal{\beta}\left(
\theta_{2}\right)  \mathcal{O}\left(  1\right)  & \mathcal{\beta}\left(
\theta_{2}\right)  \mathcal{O}\left(  1+\left\vert k\right\vert \right) \\
\mathcal{\beta}\left(  \theta_{1}\right)  \mathcal{O}\left(  \frac
{1}{1+\left\vert k\right\vert }\right)  & \mathcal{\beta}\left(  \theta
_{1}\right)  \mathcal{O}\left(  1\right)  -\alpha\left(  \theta_{1}\right)  &
\mathcal{\beta}\left(  \theta_{1}\right)  \mathcal{O}\left(  \frac
{1}{1+\left\vert k\right\vert }\right)  & \mathcal{\beta}\left(  \theta
_{1}\right)  \mathcal{O}\left(  1\right) \\
\mathcal{\beta}\left(  -\theta_{2}\right)  \mathcal{O}\left(  1\right)  &
\mathcal{\beta}\left(  -\theta_{2}\right)  \mathcal{O}\left(  1+\left\vert
k\right\vert \right)  & \mathcal{\beta}\left(  -\theta_{2}\right)
\mathcal{O}\left(  1\right)  -\alpha\left(  -\theta_{2}\right)  &
\mathcal{\beta}\left(  -\theta_{2}\right)  \mathcal{O}\left(  1+\left\vert
k\right\vert \right) \\
\mathcal{\beta}\left(  -\theta_{1}\right)  \mathcal{O}\left(  \frac
{1}{1+\left\vert k\right\vert }\right)  & \mathcal{\beta}\left(  -\theta
_{1}\right)  \mathcal{O}\left(  1\right)  & \mathcal{\beta}\left(  -\theta
_{1}\right)  \mathcal{O}\left(  \frac{1}{1+\left\vert k\right\vert }\right)  &
\mathcal{\beta}\left(  -\theta_{1}\right)  \mathcal{O}\left(  1\right)
-\alpha\left(  -\theta_{1}\right)
\end{pmatrix}
\,.\nonumber
\end{align}
From the Lemma \ref{Lemma_Krein_coeff}, this matrix is invertible whenever
$\left(  \theta_{1},\theta_{2}\right)  \in\mathcal{B}_{\delta}\left(  \left(
0,0\right)  \right)  $ with $\delta$ small enough; under such a condition,
indeed, the determinant's expansion%
\[
\left.
\begin{array}
[c]{l}%
\det\mathcal{M}\left(  k,\theta_{1},\theta_{2}\right)  =\det\left(
-A_{\theta_{1},\theta_{2}}\right)  +\mathcal{O}\left(  \theta_{1}\right)
+\mathcal{O}\left(  \theta_{2}\right)  \,,\\
\\
\det\left(  -A_{\theta_{1},\theta_{2}}\right)  =4\left(  1+\cosh\frac
{\theta_{1}}{2}\right)  \left(  1+\cosh\frac{\theta_{2}}{2}\right)  \,,
\end{array}
\right.
\]
(see the relation (\ref{d_exp})) implies: $\left\vert \det\mathcal{M}\left(
k,\theta_{1},\theta_{2}\right)  \right\vert >1$. Then, a direct computation
leads to%
\begin{equation}
\mathcal{M}^{-1}\left(  k,\theta_{1},\theta_{2}\right)  =\frac{1}%
{\det\mathcal{M}\left(  k,\theta_{1},\theta_{2}\right)  }\left[  \det\left(
A_{\theta_{1},\theta_{2}}\right)  \,diag\left\{  \lambda_{i}\right\}
\,+M(k,\theta_{1},\theta_{2})\right]  \,, \label{krein_coeff_inv}%
\end{equation}
where $diag\left(  \lambda_{i}\right)  $, the main term in
(\ref{krein_coeff_inv}), is the $\mathbb{C}^{4,4}$ diagonal matrix defined by
the coefficients%
\begin{equation}
\left\{  \lambda_{i}\right\}  _{i=1}^{4}=\left\{  \frac{-1}{\alpha\left(
\theta_{2}\right)  }\,,\ \frac{-1}{\alpha\left(  \theta_{1}\right)
}\,,\ \frac{-1}{\alpha\left(  -\theta_{2}\right)  }\,,\ \frac{-1}%
{\alpha\left(  -\theta_{1}\right)  }\right\}  \,, \label{coeff}%
\end{equation}
while the remainder is%
\begin{equation}
M(k,\theta_{1},\theta_{2})=%
\begin{pmatrix}
\mathcal{O}\left(  \theta_{1}\right)  +\mathcal{O}\left(  \theta_{2}\right)  &
\mathcal{O}\left(  \theta_{2}(1+\left\vert k\right\vert )\right)  &
\mathcal{O}\left(  \theta_{2}\right)  & \mathcal{O}\left(  \theta
_{2}(1+\left\vert k\right\vert )\right) \\
\mathcal{O}\left(  \frac{\theta_{1}}{1+\left\vert k\right\vert }\right)  &
\mathcal{O}\left(  \theta_{1}\right)  +\mathcal{O}\left(  \theta_{2}\right)  &
\mathcal{O}\left(  \frac{\theta_{1}}{1+\left\vert k\right\vert }\right)  &
\mathcal{O}\left(  \theta_{1}\right) \\
\mathcal{O}\left(  \theta_{2}\right)  & \mathcal{O}\left(  \theta
_{2}(1+\left\vert k\right\vert )\right)  & \mathcal{O}\left(  \theta
_{1}\right)  +\mathcal{O}\left(  \theta_{2}\right)  & \mathcal{O}\left(
\theta_{2}(1+\left\vert k\right\vert )\right) \\
\mathcal{O}\left(  \frac{\theta_{1}}{1+\left\vert k\right\vert }\right)  &
\mathcal{O}\left(  \theta_{1}\right)  & \mathcal{O}\left(  \frac{\theta_{1}%
}{1+\left\vert k\right\vert }\right)  & \mathcal{O}\left(  \theta_{1}\right)
+\mathcal{O}\left(  \theta_{2}\right)
\end{pmatrix}
\,, \label{krein_coeff_inv1}%
\end{equation}
Here, the symbols $\mathcal{O}\left(  \cdot\right)  $ are referred to the
metric space $\mathcal{B}_{\delta}\left(  \left(  0,0\right)  \right)
\times\mathbb{R}$ and, being obtained from the calculus of the inverse matrix
$\left(  B_{\theta_{1},\theta_{2}}\,q(z,\mathcal{V})-A_{\theta_{1},\theta_{2}%
}\right)  ^{-1}$, denotes polynomial expressions depending on the functions:
$\alpha\left(  \pm\theta_{i}\right)  $, $\mathcal{\beta}\left(  \pm\theta
_{i}\right)  $, $G^{\zeta}\left(  x,y\right)  $, $\partial_{1}^{i}H^{\zeta
}\left(  x,y\right)  $, with $x,y\in\left\{  a,b\right\}  $ and $i=0,1$. Then,
as a consequence of Lemma \ref{Lemma_Green_ker}, these terms are holomorphic
w.r.t. the parameters $\left(  \theta_{1},\theta_{2}\right)  \in
\mathcal{B}_{\delta}\left(  \left(  0,0\right)  \right)  $ and continuous
w.r.t. $k\in\mathbb{R}$.

\begin{proposition}
\label{Proposition_Krein_gen_eigenfun}Assume $\left(  \theta_{1},\theta
_{2}\right)  \in\mathcal{B}_{\delta}\left(  \left(  0,0\right)  \right)  $
with $\delta>0$ small enough, and let $\mathcal{V}$ be defined according to
(\ref{V}), (\ref{V_pos}). The solutions $\psi_{-}(\cdot,k,\theta_{1}%
,\theta_{2})$ to the generalized eigenfunctions problem (\ref{Jost_eq_teta}),
(\ref{gen_eigenfun_ext1})-(\ref{gen_eigenfun_ext2}) allow the representation%
\begin{equation}
\psi_{-}(\cdot,k,\theta_{1},\theta_{2})=\left\{
\begin{array}
[c]{lll}%
\psi_{-}(\cdot,k)-\sum_{i,j=1}^{4}\left[  \mathcal{M}^{-1}\left(  k,\theta
_{1},\theta_{2}\right)  B_{\theta_{1},\theta_{2}}\right]  _{ij}\left[
\Gamma_{1}\psi_{-}(\cdot,k)\right]  _{j}\,g(e_{i},k,\mathcal{V})\,, &  &
\text{for }k\geq0\,,\\
&  & \\
\psi_{-}(\cdot,k)-\sum_{i,j=1}^{4}\left[  \mathcal{M}^{-1}\left(
-k,\theta_{1},\theta_{2}\right)  B_{\theta_{1},\theta_{2}}\right]
_{ij}\left[  \Gamma_{1}\psi_{-}(\cdot,k)\right]  _{j}\,g(e_{i},-k,\mathcal{V}%
)\,, &  & \text{for }k<0\,.
\end{array}
\right.  \label{gen_eigenfun_Krein}%
\end{equation}
The functions $\psi_{-}(x,k,\theta_{1},\theta_{2})$ are $\mathcal{C}^{1}%
$-continuous w.r.t. $x\in\mathbb{R}/\left\{  a,b\right\}  $, $k$-continuous in
$\mathbb{R}$ and holomorphic w.r.t. the parameters $\left(  \theta_{1}%
,\theta_{2}\right)  $ in $\mathcal{B}_{\delta}\left(  \left(  0,0\right)
\right)  $.
\end{proposition}

\begin{proof}
We start considering the case $k\geq0$. According to the definition of
$\psi_{-}(\cdot,k)$ and $\,g(e_{i},k,\mathcal{V})$, the function at the r.h.s.
of (\ref{gen_eigenfun_Krein}) solves the equation%
\begin{equation}
\left(  -\partial_{x}^{2}+\mathcal{V}\right)  u=k^{2}u\,,\qquad\text{for }%
x\in\mathbb{R}\backslash\left\{  a,b\right\}  \,,\ k\in\mathbb{R}\,,
\end{equation}
and fulfills the conditions (\ref{gen_eigenfun_ext1}) and
(\ref{gen_eigenfun_ext2}). Set: $\psi_{-}(\cdot,k,\theta_{1},\theta_{2}%
)=\phi-\psi$ with
\begin{align}
\phi &  =\psi_{-}(\cdot,k)\,,\label{phi}\\
& \nonumber\\
\psi &  =\sum_{i,j=1}^{4}\left[  \mathcal{M}^{-1}\left(  k,\theta_{1}%
,\theta_{2}\right)  B_{\theta_{1},\theta_{2}}\right]  _{ij}\left[  \Gamma
_{1}\phi\right]  _{j}\,g(e_{i},k,\mathcal{V})\,.\label{psi}%
\end{align}
The function $\psi$ can be pointwise approximated by elements of the defect
spaces $\mathcal{N}_{z}$ as $z\rightarrow k^{2}+i0$. With the notation
introduced in (\ref{M_zeta_teta}) and (\ref{g_zeta}), let $\psi_{z}$ be
defined by%
\begin{equation}
\psi_{z}=\sum_{i,j=1}^{4}\left[  \mathcal{M}^{-1}\left(  \zeta,\theta
_{1},\theta_{2}\right)  B_{\theta_{1},\theta_{2}}\right]  _{ij}\left[
\Gamma_{1}\phi\right]  _{j}\,g(e_{i},\zeta,\mathcal{V})\,,\quad\zeta
\in\mathbb{C}^{+}\,,\ z=\zeta^{2}\,;\label{psi_zeta}%
\end{equation}
it results $\psi_{z}\in\mathcal{N}_{z}$ and $\lim_{z\rightarrow k^{2}+i0}%
\psi_{z}=\psi$. Since $\psi_{-}(\cdot,k)$ is $\mathcal{C}_{x}^{1}$-continuous
in $\mathbb{R}$, we have: $\Gamma_{0}\phi=0$ and the following relation holds%
\begin{align}
\mathcal{M}\left(  k,\theta_{1},\theta_{2}\right)  \Gamma_{0}\left(  \phi
-\psi\right)   &  =-\mathcal{M}\left(  k,\theta_{1},\theta_{2}\right)
\Gamma_{0}\psi=-\lim_{z\rightarrow k^{2}+i0}\left(  B_{\theta_{1},\theta_{2}%
}\,q(z,\mathcal{V})-A_{\theta_{1},\theta_{2}}\right)  \Gamma_{0}\psi
_{z}\nonumber\\
&  =-\lim_{z\rightarrow k^{2}+i0}\left(  B_{\theta_{1},\theta_{2}}\,\Gamma
_{1}\gamma(\cdot,z,\mathcal{V})-A_{\theta_{1},\theta_{2}}\right)  \Gamma
_{0}\psi_{z}\nonumber\\
&  =-\lim_{z\rightarrow k^{2}+i0}\left(  B_{\theta_{1},\theta_{2}}\,\Gamma
_{1}-A_{\theta_{1},\theta_{2}}\Gamma_{0}\right)  \psi_{z}=\left(
-B_{\theta_{1},\theta_{2}}\,\Gamma_{1}+A_{\theta_{1},\theta_{2}}\Gamma
_{0}\right)  \psi\,.\label{one_k}%
\end{align}
The $n$-th component of the vector at the l.h.s. of (\ref{one_k}) writes as%
\begin{multline*}
\left[  \smallskip\mathcal{M}\left(  k,\theta_{1},\theta_{2}\right)
\Gamma_{0}\left(  \phi-\psi\right)  \right]  _{n}=\left[  \smallskip
-\mathcal{M}\left(  k,\theta_{1},\theta_{2}\right)  \Gamma_{0}\psi\right]
_{n}=\\
-\sum_{i,j=1}^{4}\left[  \left[  \smallskip\mathcal{M}\left(  k,\theta
_{1},\theta_{2}\right)  \Gamma_{0}\left(  \phi-\psi\right)  \right]
_{n}\Gamma_{0}g(e_{i},k,\mathcal{V})\right]  _{n}\left[  \mathcal{M}%
^{-1}\left(  k,\theta_{1},\theta_{2}\right)  B_{\theta_{1},\theta_{2}}\right]
_{ij}\left[  \Gamma_{1}\phi\right]  _{j}\,.
\end{multline*}
Recalling that $\Gamma_{0}g(e_{i},k,\mathcal{V})=e_{i}$, we get
\begin{multline*}
\left[  \smallskip\mathcal{M}\left(  k,\theta_{1},\theta_{2}\right)
\Gamma_{0}\left(  \phi-\psi\right)  \right]  _{n}=\\
-\sum_{i,j=1}^{4}\left(  \mathcal{M}\left(  k,\theta_{1},\theta_{2}\right)
\right)  _{ni}\left[  \mathcal{M}^{-1}\left(  k,\theta_{1},\theta_{2}\right)
B_{\theta_{1},\theta_{2}}\right]  _{ij}\left[  \Gamma_{1}\phi\right]
_{j}=-\sum_{i,j=1}^{4}B_{nj}\left[  \Gamma_{1}\phi\right]  _{j}\,,
\end{multline*}
which implies%
\begin{equation}
\mathcal{M}\left(  k,\theta_{1},\theta_{2}\right)  \Gamma_{0}\left(  \phi
-\psi\right)  =-B\Gamma_{1}\phi\,.\label{two_k}%
\end{equation}
From (\ref{one_k}) and (\ref{two_k}), the interface conditions%
\begin{equation}
A_{\theta_{1},\theta_{2}}\Gamma_{0}\psi_{-}(\cdot,k,\theta_{1},\theta
_{2})=B_{\theta_{1},\theta_{2}}\Gamma_{1}\psi_{-}(\cdot,k,\theta_{1}%
,\theta_{2})\,,\label{gen_eigenfun_bc}%
\end{equation}
follow. Since these are equivalent to the ones assigned in the equation
(\ref{Jost_eq_teta}), the function defined in (\ref{gen_eigenfun_Krein}) is a
solution to the problem (\ref{Jost_eq_teta}), (\ref{gen_eigenfun_ext1}%
)-(\ref{gen_eigenfun_ext2}). The case $k<0$ can be treated by a suitable
adaptation of the above arguments.

The regularity of the generalized eigenfunctions w.r.t. to the variables
$\left\{  x,k,\theta_{1},\theta_{2}\right\}  $ is a consequence of the
representation (\ref{gen_eigenfun_Krein}) and of the properties of the maps
$\psi_{-}(\cdot,k)$, $\mathcal{M}^{-1}\left(  k,\theta_{1},\theta_{2}\right)
$, $B_{\theta_{1},\theta_{2}}$, and $g(e_{i},k,\mathcal{V})$ (for this point,
we refer to the corresponding definitions and to the results of the
Proposition \ref{Proposition_Jost} and of the Lemmata \ref{Lemma_Green_ker}%
-\ref{Lemma_Krein_coeff}).
\end{proof}

As a consequence of the above result, an expansion of $\psi_{-}(\cdot
,k,\theta_{1},\theta_{2})$ for small values of $\theta_{i}$ follows.

\begin{corollary}
\label{Corollary_eigenfun_exp}Let $\psi_{-}(\cdot,k,\theta_{1},\theta_{2})$
denotes a solution to the generalized eigenfunctions problem
(\ref{Jost_eq_teta}), (\ref{gen_eigenfun_ext1})-(\ref{gen_eigenfun_ext2}). In
the assumptions of the Proposition \ref{Proposition_Krein_gen_eigenfun}, the
expansion%
\begin{equation}
\psi_{-}(\cdot,k,\theta_{1},\theta_{2})-\psi_{-}(\cdot,k)=\mathcal{O}\left(
\theta_{2}k\right)  G^{\sigma k}(\cdot,b)+\mathcal{O}\left(  \frac{\theta
_{1}k}{1+\left\vert k\right\vert }\right)  H^{\sigma k}(\cdot,b)+\mathcal{O}%
\left(  \theta_{2}k\right)  G^{\sigma k}(\cdot,a)+\mathcal{O}\left(
\frac{\theta_{1}k}{1+\left\vert k\right\vert }\right)  H^{\sigma k}%
(\cdot,a)\,. \label{gen_eigenfun_exp}%
\end{equation}
holds with: $\sigma=\frac{k}{\left\vert k\right\vert }$. The symbols
$\mathcal{O}\left(  \cdot\right)  $ are defined in the sense of the metric
space $\mathbb{R\times}\mathcal{B}_{\delta}\left(  \left(  0,0\right)
\right)  $.
\end{corollary}

\begin{proof}
As already noticed, the assumption of positive potentials (\ref{V_pos})
prevents the Jost's function $w(k)$ to have zeroes on the real axis. In
particular, a consequence of the definition (\ref{gen_eigenfun_0}) and of the
relations (\ref{Jost_sol_bound}) is%
\begin{equation}
\psi_{-}(x,k)=\mathcal{O}\left(  \frac{k}{1+\left\vert k\right\vert }\right)
\,;\qquad\partial_{x}\psi_{-}(x,k)=\mathcal{O}\left(  k\right)  \,,
\label{gen_eigenfun_0_est}%
\end{equation}
and a direct computation yields%
\begin{equation}
B_{\theta_{1},\theta_{2}}\Gamma_{1}\psi_{-}(\cdot,k)=\left\{  \,\mathcal{O}%
\left(  \theta_{2}k\right)  ,\ \mathcal{O}\left(  \frac{\theta_{1}%
k}{1+\left\vert k\right\vert }\right)  \,,\ \mathcal{O}\left(  \theta
_{2}k\right)  \,,\ \mathcal{O}\left(  \frac{\theta_{1}k}{1+\left\vert
k\right\vert }\right)  \right\}  \,.
\end{equation}
where the symbols $\mathcal{O}\left(  \cdot\right)  $ are referred to the
metric space $\mathbb{R\times}\mathcal{B}_{\delta}\left(  \left(  0,0\right)
\right)  $. Making use of this expression and of the relations
(\ref{krein_coeff_inv})-(\ref{krein_coeff_inv1}), we get%
\begin{equation}
\mathcal{M}^{-1}\left(  \sigma k,\theta_{1},\theta_{2}\right)  B_{\theta
_{1},\theta_{2}}\left[  \Gamma_{1}\psi_{-}(\cdot,k)\right]  =\left\{
\,\mathcal{O}\left(  \theta_{2}k\right)  ,\ \mathcal{O}\left(  \frac
{\theta_{1}k}{1+\left\vert k\right\vert }\right)  \,,\ \mathcal{O}\left(
\theta_{2}k\right)  \,,\ \mathcal{O}\left(  \frac{\theta_{1}k}{1+\left\vert
k\right\vert }\right)  \right\}  \,. \label{coeff1}%
\end{equation}
Then, the expansion (\ref{gen_eigenfun_exp}) follows from the formula
(\ref{gen_eigenfun_Krein}) by taking into account (\ref{coeff1}) and the
definition (\ref{g_zeta}).
\end{proof}

\section{\label{Section_Similarity}Similarity and uniform-in-time estimates
for the dynamical system.}

In what follows, $\mathcal{V}$ is a positive short-range potential. With this
assumption, $Q_{0,0}(\mathcal{V})$ has a purely absolutely continuous spectrum
and the related generalized Fourier transform $\mathcal{F}_{\mathcal{V}}$%
\begin{equation}
\left(  \mathcal{F}_{\mathcal{V}}\varphi\right)  (k)=\int_{\mathbb{R}}%
\frac{dx}{\left(  2\pi\right)  ^{1/2}}\,\psi_{-}^{\ast}(x,k)\varphi
(x)\,,\qquad\varphi\in L^{2}(\mathbb{R})\,,\label{gen_Fourier}%
\end{equation}
is a unitary map with range $R\left(  \mathcal{F}_{\mathcal{V}}\right)
\mathcal{=}$ $L^{2}(\mathbb{R})$ and an inverse map $\mathcal{F}_{\mathcal{V}%
}^{-1}$ acting as%
\begin{equation}
\left(  \mathcal{F}_{\mathcal{V}}^{-1}f\right)  (x)=\int\frac{dk}{\left(
2\pi\right)  ^{1/2}}\,\psi_{-}(x,k)f(k)\,,
\end{equation}
for all $f\in L^{2}(\mathbb{R})$. Assume in addition the parameters
$\theta_{1},\theta_{2}$ to be close enough to the origin, so that the
expansion (\ref{gen_eigenfun_exp}) hold, and consider the operator
$\mathcal{W}_{\theta_{1},\theta_{2}}$ defined by the integral kernel%
\begin{equation}
\mathcal{W}_{\theta_{1},\theta_{2}}(x,y)=\int_{\mathbb{R}}\frac{dk}{2\pi
}\,\psi_{-}(x,k,\theta_{1},\theta_{2})\psi_{-}^{\ast}%
(y,k)\,.\label{W_teta_ker}%
\end{equation}
The next Proposition shows that $\mathcal{W}_{\theta_{1},\theta_{2}}$ form an
analytic family of bounded operators w.r.t. $\left(  \theta_{1},\theta
_{2}\right)  $, while, for fixed values of the parameters, $\mathcal{W}%
_{\theta_{1},\theta_{2}}$ induces a similarity between $Q_{\theta_{1}%
,\theta_{2}}(\mathcal{V})$ and $Q_{0,0}(\mathcal{V})$.

\begin{proposition}
\label{Proposition_W_cont}Let $\mathcal{V}$ satisfy the conditions (\ref{V}),
(\ref{V_pos}) and assume $\left(  \theta_{1},\theta_{2}\right)  \in
\mathcal{B}_{\delta}(\left(  0,0\right)  )$ with $\delta>0$ small enough.
Then, the set $\left\{  \mathcal{W}_{\theta_{1},\theta_{2}}\,,\ \left(
\theta_{1},\theta_{2}\right)  \in\mathcal{B}_{\delta}(\left(  0,0\right)
)\right\}  $ forms an analytic family of bounded operators in $L^{2}%
(\mathbb{R})$, w.r.t. $\left(  \theta_{1},\theta_{2}\right)  $, and the
expansion%
\begin{equation}
\mathcal{W}_{\theta_{1},\theta_{2}}=1+\mathcal{O}\left(  \theta_{1}\right)
+\mathcal{O}\left(  \theta_{2}\right)  \,, \label{W_teta_exp}%
\end{equation}
holds in the $\mathcal{L}\left(  L^{2}(\mathbb{R}),L^{2}(\mathbb{R})\right)  $
operator norm. The couple $Q_{\theta_{1},\theta_{2}}(\mathcal{V})$,
$Q_{0,0}(\mathcal{V})$ is intertwined through $\mathcal{W}_{\theta_{1}%
,\theta_{2}}$ by%
\begin{equation}
Q_{\theta_{1},\theta_{2}}(\mathcal{V})\mathcal{W}_{\theta_{1},\theta_{2}%
}=\mathcal{W}_{\theta_{1},\theta_{2}}Q_{0,0}(\mathcal{V})\,.
\label{W_teta_inter}%
\end{equation}

\end{proposition}

\begin{proof}
Let consider the action of $\mathcal{W}_{\theta_{1},\theta_{2}}$ on
$\varphi\in L^{2}(\mathbb{R})$; making use of (\ref{gen_Fourier}) and
(\ref{W_teta_ker}), this writes as
\begin{equation}
\mathcal{W}_{\theta_{1},\theta_{2}}\varphi=\int_{\mathbb{R}}\frac{dk}{\left(
2\pi\right)  ^{1/2}}\,\psi_{-}(\cdot,k,\theta_{1},\theta_{2})\left(
\mathcal{F}_{\mathcal{V}}\varphi\right)  (k)\,, \label{W_teta_act}%
\end{equation}
and, expressing $\psi_{-}(x,k,\theta_{1},\theta_{2})$ through the expansion
(\ref{gen_eigenfun_exp}), we get%
\begin{align}
\mathcal{W}_{\theta_{1},\theta_{2}}\varphi &  =\int_{\mathbb{R}}\frac
{dk}{\left(  2\pi\right)  ^{1/2}}\,\psi_{-}(\cdot,k)\left(  \mathcal{F}%
_{\mathcal{V}}\varphi\right)  (k)+\int_{\mathbb{R}}\frac{dk}{\left(
2\pi\right)  ^{1/2}}\,\left[  \mathcal{O}\left(  \theta_{2}k\right)
G^{\left\vert k\right\vert }(\cdot,b)+\mathcal{O}\left(  \theta_{2}k\right)
G^{\left\vert k\right\vert }(\cdot,a)\right]  \left(  \mathcal{F}%
_{\mathcal{V}}\varphi\right)  (k)\nonumber\\
& \nonumber\\
&  +\int_{\mathbb{R}}\frac{dk}{\left(  2\pi\right)  ^{1/2}}\,\left[
\mathcal{O}\left(  \frac{\theta_{1}k}{1+\left\vert k\right\vert }\right)
H^{\left\vert k\right\vert }(\cdot,b)+\mathcal{O}\left(  \frac{\theta_{1}%
k}{1+\left\vert k\right\vert }\right)  H^{\left\vert k\right\vert }%
(\cdot,a)\right]  \left(  \mathcal{F}_{\mathcal{V}}\varphi\right)  (k)\,,
\label{W_teta_exp1}%
\end{align}
where, it is important to remark, the symbols $\mathcal{O}\left(
\cdot\right)  $ here denote functions depending only on $k$, $\theta_{1}$ and
$\theta_{2}$, but independent of $x$. Since $\int\frac{dk}{\left(
2\pi\right)  ^{1/2}}\,\psi_{-}(\cdot,k)\left(  \mathcal{F}_{\mathcal{V}%
}\varphi\right)  (k)=\mathcal{F}_{\mathcal{V}}^{-1}\left(  \mathcal{F}%
_{\mathcal{V}}\varphi\right)  $, this equation yields: $\left(  \mathcal{W}%
_{\theta_{1},\theta_{2}}-\mathbb{I}\right)  \varphi=I+II$, where%
\begin{equation}
I(\varphi)=\int_{\mathbb{R}}\frac{dk}{\left(  2\pi\right)  ^{1/2}}\,\left[
\mathcal{O}\left(  \theta_{2}k\right)  G^{\left\vert k\right\vert }%
(\cdot,b)+\mathcal{O}\left(  \theta_{2}k\right)  G^{\left\vert k\right\vert
}(\cdot,a)\right]  \left(  \mathcal{F}_{\mathcal{V}}\varphi\right)  (k)\,,
\label{I}%
\end{equation}
and%
\begin{equation}
II(\varphi)=\int_{\mathbb{R}}\frac{dk}{\left(  2\pi\right)  ^{1/2}}\,\left[
\mathcal{O}\left(  \frac{\theta_{1}k}{1+\left\vert k\right\vert }\right)
H^{\left\vert k\right\vert }(\cdot,b)+\mathcal{O}\left(  \frac{\theta_{1}%
k}{1+\left\vert k\right\vert }\right)  H^{\left\vert k\right\vert }%
(\cdot,a)\right]  \left(  \mathcal{F}_{\mathcal{V}}\varphi\right)  (k)\,.
\label{II}%
\end{equation}

In order to obtain the expansion (\ref{W_teta_exp}), $L^{2}$-norm estimates of
the maps defined in (\ref{I}) and (\ref{II}) are needed. We consider at first
the case of $I(\varphi)$; let define $\phi_{\alpha}$ as%
\begin{equation}
\phi_{\alpha}(x)=\int_{\mathbb{R}}\frac{dk}{\left(  2\pi\right)  ^{1/2}%
}\,\mathcal{O}\left(  k\right)  G^{\left\vert k\right\vert }(x,\alpha)\left(
\mathcal{F}_{\mathcal{V}}\varphi\right)  (k)\,,\quad\alpha\in\left\{
a,b\right\}  \,, \label{Phi_alpha}%
\end{equation}
being $\mathcal{O}\left(  \cdot\right)  $ depending only on $k$. The $L^{2}%
$-norm of $\phi_{\alpha}$ is bounded by%
\begin{equation}
\left\Vert \phi_{\alpha}\right\Vert _{L^{2}\left(  \mathbb{R}\right)  }%
\leq\left\Vert 1_{\left\{  x\leq a\right\}  }\phi_{\alpha}\right\Vert
_{L^{2}\left(  \mathbb{R}\right)  }+\left\Vert 1_{\left(  a,b\right)  }%
\phi_{\alpha}\right\Vert _{L^{2}\left(  \mathbb{R}\right)  }+\left\Vert
1_{\left\{  x\geq b\right\}  }\phi_{\alpha}\right\Vert _{L^{2}\left(
\mathbb{R}\right)  }\,. \label{Phi_alpha_bound}%
\end{equation}
For $\alpha=b$, making use of the explicit form of $G^{k}(x,b)$, given by
(\ref{G_z}) for $\zeta=k$, and exploiting the relations (\ref{Jost_sol_bound})
and%
\begin{equation}
1_{\left\{  x\leq a\right\}  }\chi_{-}(x,k)=e^{-ikx}\,,\qquad1_{\left\{  x\geq
b\right\}  }\chi_{+}(x,k)=e^{-ikx}\,, \label{Jost_ext2}%
\end{equation}
we have%
\begin{align}
1_{\left\{  x\leq a\right\}  }(x)\mathcal{O}\left(  k\right)  G^{\left\vert
k\right\vert }(x,b)  &  =1_{\left\{  x\leq a\right\}  }(x)\tau_{1}\left(
k\right)  e^{-i\left\vert k\right\vert x}\label{Green_ext1}\\
& \nonumber\\
1_{\left\{  x\geq b\right\}  }(x)\mathcal{O}\left(  k\right)  G^{\left\vert
k\right\vert }(x,b)  &  =1_{\left\{  x\geq b\right\}  }(x)\tau_{2}\left(
k\right)  e^{i\left\vert k\right\vert x} \label{Green_ext2}%
\end{align}
with $\tau_{1},\tau_{2}\in L_{k}^{\infty}\left(  \mathbb{R}\right)  $. In the
following, $\mathcal{P}$ denotes be the parity operator: $\mathcal{P}%
u(t)=u(-t)$; from (\ref{Green_ext1}), we get%
\begin{align}
1_{\left\{  x\leq a\right\}  }(x)\phi_{b}(x)  &  =1_{\left\{  x\leq a\right\}
}(x)\int_{\mathbb{R}}\frac{dk}{\left(  2\pi\right)  ^{1/2}}\,\tau_{1}\left(
k\right)  e^{-i\left\vert k\right\vert x}\left(  \mathcal{F}_{\mathcal{V}%
}\varphi\right)  (k)\nonumber\\
&  =1_{\left\{  x\leq a\right\}  }(x)\left(  \mathcal{F}_{0}^{-1}\left(
1_{k<0}\tau_{1}\mathcal{F}_{\mathcal{V}}\varphi+\mathcal{P}\left(  1_{k>0}%
\tau_{1}\mathcal{F}_{\mathcal{V}}\varphi\right)  \right)  \right)  (x)\,,
\end{align}
where, according to the notation introduced in (\ref{gen_Fourier}),
$\mathcal{F}_{0}$ is the standard Fourier transform. Thus, $1_{\left\{  x\leq
a\right\}  }\phi_{b}$ is estimated by
\begin{equation}
\left\Vert 1_{\left\{  x\leq a\right\}  }\phi_{b}\right\Vert _{L^{2}\left(
\mathbb{R}\right)  }=\left\Vert \mathcal{F}_{0}^{-1}\left(  1_{k<0}\tau
_{1}\mathcal{F}_{\mathcal{V}}\varphi\right)  \right\Vert _{L^{2}\left(
\mathbb{R}\right)  }+\left\Vert \mathcal{F}_{0}^{-1}\mathcal{P}\left(
1_{k>0}\tau_{1}\mathcal{F}_{\mathcal{V}}\varphi\right)  \right\Vert
_{L^{2}\left(  \mathbb{R}\right)  }\lesssim\left\Vert \varphi\right\Vert
_{L^{2}\left(  \mathbb{R}\right)  }\,, \label{est1}%
\end{equation}
while, for $1_{\left\{  x\geq b\right\}  }\phi_{b}$, a similar inequality
follows by using (\ref{Green_ext2})%
\begin{equation}
\left\Vert 1_{\left\{  x\geq b\right\}  }\phi_{b}\right\Vert _{L^{2}\left(
\mathbb{R}\right)  }=\left\Vert \mathcal{F}_{0}^{-1}\mathcal{P}\left(
1_{k<0}\tau_{2}\mathcal{F}_{\mathcal{V}}\varphi\right)  \right\Vert
_{L^{2}\left(  \mathbb{R}\right)  }+\left\Vert \mathcal{F}_{0}^{-1}\left(
1_{k>0}\tau_{2}\mathcal{F}_{\mathcal{V}}\varphi\right)  \right\Vert
_{L^{2}\left(  \mathbb{R}\right)  }\lesssim\left\Vert \varphi\right\Vert
_{L^{2}\left(  \mathbb{R}\right)  }\,. \label{est2}%
\end{equation}
According to definition of $G^{k}(x,b)$ for $x<b$, the term $1_{\left(
a,b\right)  }\phi_{b}$ writes as%
\begin{equation}
1_{\left(  a,b\right)  }(x)\phi_{b}(x)=1_{\left(  a,b\right)  }(x)\int
_{\mathbb{R}}dk\,\frac{\mathcal{O}\left(  k\right)  \chi_{+}(b,\left\vert
k\right\vert )}{w(\left\vert k\right\vert )}\chi_{-}(x,\left\vert k\right\vert
)\left(  \mathcal{F}_{\mathcal{V}}\varphi\right)  (k)=1_{\left(  a,b\right)
}(x)\int_{\mathbb{R}}dk\,\chi_{-}(x,\left\vert k\right\vert )\tau
_{3}(k)\left(  \mathcal{F}_{\mathcal{V}}\varphi\right)  (k)\,, \label{est3_0}%
\end{equation}
where $\tau_{3}\in L_{k}^{\infty}\left(  \mathbb{R}\right)  $ is: $\tau
_{3}(k)=\frac{\mathcal{O}\left(  k\right)  \chi_{+}(b,\left\vert k\right\vert
)}{w(\left\vert k\right\vert )}$. Using the definition (\ref{gen_eigenfun_0})
and the identities%
\begin{equation}
\chi_{\pm}(\cdot,-k)=\chi_{\pm}^{\ast}(\cdot,k)\,,\qquad w(-k)=w^{\ast}(k)\,,
\end{equation}
it follows%
\begin{align}
1_{k<0}(k)\chi_{-}(x,-k)  &  =1_{k<0}(k)\frac{w\left(  -k\right)  }{2ik}%
\psi_{-}(x,k)\,,\label{Jost_gen_eigenfun1}\\
& \nonumber\\
1_{k\geq0}(k)\chi_{-}(x,k)  &  =-1_{k\geq0}(k)\frac{w\left(  k\right)  }%
{2ik}\psi_{-}(x,-k)\,. \label{Jost_gen_eigenfun2}%
\end{align}
Take $\tilde{\tau}_{3}(k)=\tau_{3}(k)\frac{w\left(  -k\right)  }{2ik}$ and
$\hat{\tau}_{3}(k)=\tau_{3}(k)\frac{w\left(  k\right)  }{-2ik}$; it results:
$\tilde{\tau}_{3},\hat{\tau}_{3}\in L_{k}^{\infty}\left(  \mathbb{R}\right)  $
and the r.h.s. of (\ref{est3_0}) rephrases as%
\begin{equation}
1_{\left(  a,b\right)  }(x)\phi_{b}(x)=1_{\left(  a,b\right)  }(x)\left[
\int_{k<0}dk\,\psi_{-}(x,k)\tilde{\tau}_{3}(k)\left(  \mathcal{F}%
_{\mathcal{V}}\varphi\right)  (k)+\int_{k>0}dk\,\psi_{-}(x,-k)\left(
\hat{\tau}_{3}(k)\left(  \mathcal{F}_{\mathcal{V}}\varphi\right)  (k)\right)
\right]  \,. \label{est3_1}%
\end{equation}
The first term identifies with the inverse Fourier transform of $1_{k<0}%
\tilde{\tau}_{3}\mathcal{F}_{\mathcal{V}}\varphi$,%
\begin{equation}
\int_{k<0}dk\,\psi_{-}(\cdot,k)\tilde{\tau}_{3}(k)\left(  \mathcal{F}%
_{\mathcal{V}}\varphi\right)  (k)=\mathcal{F}_{\mathcal{V}}^{-1}\left(
1_{k<0}\tilde{\tau}_{3}\mathcal{F}_{\mathcal{V}}\varphi\right)  \,,
\label{est3_2}%
\end{equation}
while, for the second term, we have%
\begin{equation}
\int_{k>0}dk\,\psi_{-}(\cdot,-k)\left(  \hat{\tau}_{3}(k)\left(
\mathcal{F}_{\mathcal{V}}\varphi\right)  (k)\right)  =-\int_{k<0}dk\,\psi
_{-}(\cdot,k)\left(  \mathcal{P}\left(  \hat{\tau}_{3}\mathcal{F}%
_{\mathcal{V}}\varphi\right)  \right)  (k)=-\mathcal{F}_{\mathcal{V}}%
^{-1}\mathcal{P}\left(  1_{k>0}\hat{\tau}_{3}\mathcal{F}_{\mathcal{V}}%
\varphi\right)  \,. \label{est3_3}%
\end{equation}
The above relations yield the estimate%
\begin{equation}
\left\Vert 1_{\left(  a,b\right)  }\phi_{b}\right\Vert _{L^{2}\left(
\mathbb{R}\right)  }=\left\Vert 1_{\left(  a,b\right)  }\mathcal{F}%
_{\mathcal{V}}^{-1}\left(  1_{k<0}\tilde{\tau}_{3}\mathcal{F}_{\mathcal{V}%
}\varphi\right)  ^{\ast}\right\Vert _{L^{2}\left(  \mathbb{R}\right)
}+\left\Vert 1_{\left(  a,b\right)  }\mathcal{F}_{\mathcal{V}}^{-1}%
\mathcal{P}\left(  1_{k>0}\tilde{\tau}_{3}\mathcal{F}_{\mathcal{V}}%
\varphi\right)  \right\Vert _{L^{2}\left(  \mathbb{R}\right)  }\lesssim
\left\Vert \varphi\right\Vert _{L^{2}\left(  \mathbb{R}\right)  }\,.
\label{est3_4}%
\end{equation}
As a consequence of (\ref{est1}), (\ref{est2}) and (\ref{est3_4}) we get%
\begin{equation}
\left\Vert \phi_{b}\right\Vert _{L^{2}\left(  \mathbb{R}\right)  }%
\lesssim\left\Vert \varphi\right\Vert _{L^{2}\left(  \mathbb{R}\right)  }\,,
\end{equation}
and a similar computation in the case of $\phi_{a}$ leads to: $\left\Vert
\phi_{a}\right\Vert _{L^{2}\left(  \mathbb{R}\right)  }\lesssim\left\Vert
\varphi\right\Vert _{L^{2}\left(  \mathbb{R}\right)  }$. From the definitions
(\ref{I}) and (\ref{Phi_alpha}), it follows%
\begin{equation}
\left\Vert I(\varphi)\right\Vert _{L^{2}\left(  \mathbb{R}\right)  }%
\lesssim\left\vert \theta_{2}\right\vert \left(  \left\Vert \phi
_{a}\right\Vert _{L^{2}\left(  \mathbb{R}\right)  }+\left\Vert \phi
_{b}\right\Vert _{L^{2}\left(  \mathbb{R}\right)  }\right)  \lesssim\left\vert
\theta_{2}\right\vert \,\left\Vert \varphi\right\Vert _{L^{2}\left(
\mathbb{R}\right)  }\,. \label{I_est}%
\end{equation}

For the map $II(\varphi)$, we introduce $\psi_{\alpha}$ defined as%
\begin{equation}
\psi_{\alpha}(x)=\int_{\mathbb{R}}dk\,\mathcal{O}\left(  \frac{k}{1+\left\vert
k\right\vert }\right)  H^{\left\vert k\right\vert }(x,\alpha)\left(
\mathcal{F}_{\mathcal{V}}\varphi\right)  (k)\,,\quad\alpha\in\left\{
a,b\right\}  \,, \label{Psi_alpha}%
\end{equation}
where $\mathcal{O}\left(  \cdot\right)  $ depends only on $k$. For $\alpha=b$,
the explicit form of $H^{k}(x,b)$, given by (\ref{H_z}) for $\zeta=k$, and the
relations (\ref{Jost_sol_bound}), (\ref{Jost_ext2}), yield%
\begin{align}
1_{\left\{  x\leq a\right\}  }(x)\mathcal{O}\left(  \frac{k}{1+\left\vert
k\right\vert }\right)  H^{\left\vert k\right\vert }(x,b)  &  =1_{\left\{
x\leq a\right\}  }(x)\eta_{1}\left(  k\right)  e^{-i\left\vert k\right\vert
x}\label{H_ext1}\\
& \nonumber\\
1_{\left(  a,b\right)  }(x)\mathcal{O}\left(  \frac{k}{1+\left\vert
k\right\vert }\right)  H^{\left\vert k\right\vert }(x,b)  &  =1_{\left(
a,b\right)  }(x)\eta_{3}(k)\chi_{-}(x,\left\vert k\right\vert )\label{H_in}\\
& \nonumber\\
1_{\left\{  x\geq b\right\}  }(x)\mathcal{O}\left(  \frac{k}{1+\left\vert
k\right\vert }\right)  H^{\left\vert k\right\vert }(x,b)  &  =1_{\left\{
x\geq b\right\}  }(x)\eta_{2}\left(  k\right)  e^{i\left\vert k\right\vert x}
\label{H_ext2}%
\end{align}
where $\eta_{i=1,2,3}\in L_{k}^{\infty}\left(  \mathbb{R}\right)  $ are
described by $\mathcal{O}\left(  \frac{k}{1+\left\vert k\right\vert }\right)
$. Setting: $\tilde{\eta}_{3}(k)=\eta_{3}(k)\frac{w\left(  -k\right)  }{2ik}$
and $\hat{\eta}_{3}(k)=\eta_{3}(k)\frac{w\left(  k\right)  }{-2ik}$ (which,
according to the characterization of $\eta_{3}$, still implies: $\tilde{\eta
}_{3},\hat{\eta}_{3}\in L_{k}^{\infty}\left(  \mathbb{R}\right)  $), and
proceeding as before, we obtain the decomposition%
\begin{align}
\psi_{b}  &  =1_{\left\{  x\leq a\right\}  }\left[  \mathcal{F}_{0}%
^{-1}\left(  1_{k<0}\eta_{1}\mathcal{F}_{\mathcal{V}}\varphi+\mathcal{P}%
\left(  1_{k>0}\eta_{1}\mathcal{F}_{\mathcal{V}}\varphi\right)  \right)
\right] \nonumber\\
&  +1_{\left\{  x\geq b\right\}  }\left[  \mathcal{F}_{0}^{-1}\mathcal{P}%
\left(  1_{k<0}\eta_{2}\mathcal{F}_{\mathcal{V}}\varphi\right)  +\mathcal{F}%
_{0}^{-1}\left(  1_{k>0}\eta_{2}\mathcal{F}_{\mathcal{V}}\varphi\right)
\right] \nonumber\\
&  +1_{\left(  a,b\right)  }\left[  \mathcal{F}_{\mathcal{V}}^{-1}\left(
1_{k<0}\tilde{\eta}_{3}\mathcal{F}_{\mathcal{V}}\varphi\right)  -\mathcal{F}%
_{\mathcal{V}}^{-1}\mathcal{P}\left(  1_{k>0}\hat{\eta}_{3}\mathcal{F}%
_{\mathcal{V}}\varphi\right)  \right]  \,. \label{psi_be}%
\end{align}
This entails: $\left\Vert \psi_{b}\right\Vert _{L^{2}\left(  \mathbb{R}%
\right)  }\lesssim\left\Vert \varphi\right\Vert _{L^{2}\left(  \mathbb{R}%
\right)  }$, while, with similar computations, the corresponding estimate in
the case of $\psi_{a}$ is obtained. From the definitions (\ref{II}) and
(\ref{Psi_alpha}), follows%
\begin{equation}
\left\Vert II\right\Vert _{L^{2}\left(  \mathbb{R}\right)  }\lesssim\left\vert
\theta_{1}\right\vert \left(  \left\Vert \psi_{a}\right\Vert _{L^{2}\left(
\mathbb{R}\right)  }+\left\Vert \psi_{b}\right\Vert _{L^{2}\left(
\mathbb{R}\right)  }\right)
\end{equation}
Then, the above estimates imply%
\begin{equation}
\left\Vert II\right\Vert _{L^{2}\left(  \mathbb{R}\right)  }\lesssim\left\vert
\theta_{1}\right\vert \,\left\Vert \varphi\right\Vert _{L^{2}\left(
\mathbb{R}\right)  }\,. \label{II_est}%
\end{equation}
The expansion (\ref{W_teta_exp}) is a consequence of (\ref{I_est}) and
(\ref{II_est}). Since the symbols in (\ref{I})-(\ref{II}) are holomorphic in
$\left(  \theta_{1},\theta_{2}\right)  $, the operators $\mathcal{W}%
_{\theta_{1},\theta_{2}}$ form an analytic family w.r.t. the parameters.

Next, we consider the relation (\ref{W_teta_inter}). Let $\varphi\in D\left(
Q_{0,0}(\mathcal{V})\right)  $, using the functional calculus of
$Q_{0,0}(\mathcal{V})$, we have: $\left(  \mathcal{F}_{\mathcal{V}}\left(
Q_{0,0}(\mathcal{V})\varphi\right)  \right)  (k)=k^{2}\left(  \mathcal{F}%
_{\mathcal{V}}\varphi\right)  (k)$ and, according to (\ref{W_teta_act}), the
r.h.s. of (\ref{W_teta_inter}) writes as%
\begin{equation}
\mathcal{W}_{\theta_{1},\theta_{2}}Q_{0,0}(\mathcal{V})\varphi=\int
_{\mathbb{R}}dk\,\psi_{-}(\cdot,k,\theta_{1},\theta_{2})k^{2}\left(
\mathcal{F}_{\mathcal{V}}\varphi\right)  (k)\,.\label{W_teta_inter1}%
\end{equation}
To discuss the action of $Q_{\theta_{1},\theta_{2}}(\mathcal{V})\mathcal{W}%
_{\theta_{1},\theta_{2}}$ over $D\left(  Q_{0,0}(\mathcal{V})\right)  $, we
use the expansion%
\begin{equation}
\mathcal{W}_{\theta_{1},\theta_{2}}\varphi=\varphi+I(\varphi)+II(\varphi
)\,.\label{W_teta_exp0}%
\end{equation}
From the above results, the map $I(\varphi)+II(\varphi)$ can be represented as%
\begin{align}
I(\varphi)+II(\varphi) &  =1_{\left\{  x\leq a\right\}  }\left[
\mathcal{F}_{0}^{-1}\left(  1_{k<0}\mu_{1}\mathcal{F}_{\mathcal{V}}%
\varphi+\mathcal{P}\left(  1_{k>0}\mu_{1}\mathcal{F}_{\mathcal{V}}%
\varphi\right)  \right)  \right]  \nonumber\\
&  +1_{\left\{  x\geq b\right\}  }\left[  \mathcal{F}_{0}^{-1}\mathcal{P}%
\left(  1_{k<0}\mu_{2}\mathcal{F}_{\mathcal{V}}\varphi\right)  +\mathcal{F}%
_{0}^{-1}\left(  1_{k>0}\mu_{2}\mathcal{F}_{\mathcal{V}}\varphi\right)
\right]  \nonumber\\
&  +1_{\left(  a,b\right)  }\left[  \mathcal{F}_{\mathcal{V}}^{-1}\left(
1_{k<0}\mu_{3}\mathcal{F}_{\mathcal{V}}\varphi\right)  -\mathcal{F}%
_{\mathcal{V}}^{-1}\left(  1_{k<0}\left(  \mathcal{P}\left(  \mu
_{4}\mathcal{F}_{\mathcal{V}}\varphi\right)  \right)  \right)  \right]
\,,\label{W_teta_exp2}%
\end{align}
where $\mu_{i}\in L_{k}^{\infty}\left(  \mathbb{R}\right)  $, $i=1,..4$, are
suitable bounded functions of $k$. Let $u\in L_{k}^{\infty}\left(
\mathbb{R}\right)  $ and $\mathcal{V},\mathcal{V}^{\prime}$ any couple of
potentials fulfilling the assumptions; the operators $\mathcal{F}%
_{\mathcal{V}^{\prime}}^{-1}u\mathcal{F}_{\mathcal{V}}$ and $\mathcal{F}%
_{\mathcal{V}^{\prime}}^{-1}\mathcal{P}u\mathcal{F}_{\mathcal{V}}$ map
$D\left(  Q_{0,0}(\mathcal{V})\right)  $ into itself. Then, as a consequence
of (\ref{W_teta_exp0}), (\ref{W_teta_exp2}), the operator $\mathcal{W}%
_{\theta_{1},\theta_{2}}$ maps $D\left(  Q_{0,0}(\mathcal{V})\right)  $ into
$D\left(  Q(\mathcal{V})\right)  $, while, according to (\ref{gen_eigenfun_bc}%
) and (\ref{W_teta_inter1}), $\mathcal{W}_{\theta_{1},\theta_{2}}\varphi$
fulfills the interface conditions (\ref{B_C_1}) for all $\varphi\in D\left(
Q_{0,0}(\mathcal{V})\right)  $; we obtain: $\left.  \mathcal{W}_{\theta
_{1},\theta_{2}}\in\mathcal{L}\left(  D\left(  Q_{0,0}(\mathcal{V})\right)
,D\left(  Q_{\theta_{1},\theta_{2}}(\mathcal{V})\right)  \right)  \right.  $.
Moreover, from the relation:\newline $\left.  \left(  Q_{\theta_{1},\theta
_{2}}(\mathcal{V})-k^{2}\right)  \psi_{-}(\cdot,k,\theta_{1},\theta
_{2})=0\right.  $, it follows%
\begin{equation}
Q_{\theta_{1},\theta_{2}}(\mathcal{V})\mathcal{W}_{\theta_{1},\theta_{2}%
}\varphi=\int_{\mathbb{R}}dk\,\psi_{-}(\cdot,k,\theta_{1},\theta_{2}%
)k^{2}\left(  \mathcal{F}_{\mathcal{V}}\varphi\right)
(k)\,,\label{W_teta_inter2}%
\end{equation}
which leads to (\ref{W_teta_inter}).
\end{proof}

\subsection{Proof of the Theorem\textbf{ \ref{Theorem1}%
.\label{Section_Theorem1}}}

When the parameters $\theta_{1},\theta_{2}$ are chosen in a suitably small
neighbourhood of the origin, the expansion (\ref{W_teta_exp}) yields%
\begin{equation}
\mathcal{W}_{\theta_{1},\theta_{2}}^{-1}=1+\mathcal{O}\left(  \theta
_{1}\right)  +\mathcal{O}\left(  \theta_{2}\right)  \,, \label{W_teta_exp3}%
\end{equation}
and $Q_{\theta_{1},\theta_{2}}(\mathcal{V})$ expresses as the conjugated
operator%
\begin{equation}
Q_{\theta_{1},\theta_{2}}(\mathcal{V})=\mathcal{W}_{\theta_{1},\theta_{2}%
}Q_{0,0}(\mathcal{V})\mathcal{W}_{\theta_{1},\theta_{2}}^{-1}\,.
\label{conjugation}%
\end{equation}
Let us introduce%
\begin{equation}
U_{\theta_{1},\theta_{2}}(t)=\mathcal{W}_{\theta_{1},\theta_{2}}%
U_{0,0}(t)\mathcal{W}_{\theta_{1},\theta_{2}}^{-1}\,, \label{semigroup}%
\end{equation}
being $U_{0,0}(t)=e^{-itQ_{0,0}(\mathcal{V})}$ the unitary propagator
associated with $-iQ_{0,0}(\mathcal{V})$. Due to the properties of
$\mathcal{W}_{\theta_{1},\theta_{2}}$, $U_{\theta_{1},\theta_{2}}(t)$ is
holomorphic w.r.t. $\left(  \theta_{1},\theta_{2}\right)  $, while, for fixed
values of the parameters, the family $\left\{  U_{\theta_{1},\theta_{2}%
}(t)\right\}  _{t\in\mathbb{R}}$ forms a strongly continuous group on
$L^{2}(\mathbb{R})$ and, according to (\ref{conjugation}), (\ref{semigroup}),
we have%
\begin{equation}
i\partial_{t}\left(  U_{\theta_{1},\theta_{2}}(t)u\right)  =Q_{\theta
_{1},\theta_{2}}(\mathcal{V})U_{\theta_{1},\theta_{2}}(t)u\,,
\end{equation}
for all $u\in L^{2}(\mathbb{R})$. This allows to identify $U_{\theta
_{1},\theta_{2}}(t)$ with the quantum dynamical system generated by
$-iQ_{\theta_{1},\theta_{2}}(\mathcal{V})$. Making use of (\ref{W_teta_exp})
and (\ref{W_teta_exp3}), we get%
\begin{equation}
U_{\theta_{1},\theta_{2}}(t)=U_{0,0}(t)+\mathcal{R}\left(  t,\theta_{1}%
,\theta_{2}\right)  \,,
\end{equation}
where the remainder term is strongly continuous and uniformly bounded w.r.t.
$t$ in the $L^{2}$-operator norm, allowing the representation: $\mathcal{R}%
\left(  t,\theta_{1},\theta_{2}\right)  =\mathcal{O}\left(  \theta_{1}\right)
+\mathcal{O}\left(  \theta_{2}\right)  $.

\subsection{\label{Section_WO}Time dependent wave operators and scattering
systems.}

So far, we have investigated the continuity of the dynamical system generated
by $-iQ_{\theta_{1},\theta_{2}}(\mathcal{V})$ w.r.t. the parameters
$\theta_{i=1,2}$. This has been analyzed by using small-$\theta_{i}$
expansions of the 'stationary wave operators' $\mathcal{W}_{\theta_{1}%
,\theta_{2}}$ defined in (\ref{W_teta_ker}). In what follows we consider the
scattering problem for the pair $\left\{  Q_{\theta_{1},\theta_{2}%
}(\mathcal{V}),Q_{0,0}(\mathcal{V})\right\}  $ and show that $\mathcal{W}%
_{\theta_{1},\theta_{2}}$ coincides with a wave operator of this couple. The
next Lemma discusses this point under the assumptions of the Proposition
\ref{Proposition_W_cont}.

\begin{lemma}
\label{Lemma_WO}Let $\mathcal{V}$ fulfills the conditions (\ref{V}),
(\ref{V_pos}) and assume $\left(  \theta_{1},\theta_{2}\right)  \in
\mathcal{B}_{\delta}(\left(  0,0\right)  )$ with $\delta>0$ small enough. Then%
\begin{equation}
\text{\emph{s-}}\lim_{t\rightarrow-\infty}e^{itQ_{\theta_{1},\theta_{2}%
}(\mathcal{V})}e^{-itQ_{0,0}(\mathcal{V})}=\mathcal{W}_{\theta_{1},\theta_{2}%
}\,. \label{wave_op}%
\end{equation}

\end{lemma}

\begin{proof}
Let introduce the modified transform $\mathcal{F}_{\mathcal{V},\theta
_{1},\theta_{2}}$ defined by%
\begin{equation}
\mathcal{F}_{\mathcal{V},\theta_{1},\theta_{2}}^{-1}f=\int_{\mathbb{R}%
}dk\,\psi_{-}(x,k,\theta_{1},\theta_{2})\,f(k),\qquad f\in L^{2}%
(\mathbb{R})\,. \label{gen_Fourier_teta_inverse}%
\end{equation}
The action of $\mathcal{W}_{\theta_{1},\theta_{2}}$ can be expressed in terms
of $\mathcal{F}_{\mathcal{V},\theta_{1},\theta_{2}}$ and $\mathcal{F}%
_{\mathcal{V}}$ as $\left.  \mathcal{W}_{\theta_{1},\theta_{2}}=\mathcal{F}%
_{\mathcal{V},\theta_{1},\theta_{2}}^{-1}\mathcal{F}_{\mathcal{V}}\right.  $,
from which we get: $\left.  \mathcal{F}_{\mathcal{V},\theta_{1},\theta_{2}%
}=\mathcal{F}_{\mathcal{V}}\mathcal{W}_{\theta_{1},\theta_{2}}^{-1}\right.  $.
Due to the expansion (\ref{W_teta_exp}), it results%
\begin{equation}
\mathcal{F}_{\mathcal{V},\theta_{1},\theta_{2}}=\mathcal{F}_{\mathcal{V}%
}\left(  1+\mathcal{O}\left(  \theta_{1}\right)  +\mathcal{O}\left(
\theta_{2}\right)  \right)  \label{gen_Fourier_teta_exp}%
\end{equation}
in the $L^{2}$-operator norm sense. Making use of the intertwining property,
we have%
\begin{equation}
\mathcal{W}_{\theta_{1},\theta_{2}}^{\ast}\left(  Q_{\theta_{1},\theta_{2}%
}(\mathcal{V})\right)  ^{\ast}=Q_{0,0}(\mathcal{V})\mathcal{W}_{\theta
_{1},\theta_{2}}^{\ast}\,.
\end{equation}
Since $\left.  \mathcal{W}_{\theta_{1},\theta_{2}}^{\ast}=\mathcal{F}%
_{\mathcal{V}}^{-1}\left(  \mathcal{F}_{\mathcal{V},\theta_{1},\theta_{2}%
}^{-1}\right)  ^{\ast}\right.  $ and $\left.  \left(  Q_{\theta_{1},\theta
_{2}}(\mathcal{V})\right)  ^{\ast}=Q_{-\theta_{2}^{\ast},-\theta_{1}^{\ast}%
}(\mathcal{V})\right.  $ (see eq. (\ref{Q_teta_adj})), it follows%
\begin{equation}
\mathcal{F}_{\mathcal{V}}^{-1}\left(  \mathcal{F}_{\mathcal{V},\theta
_{1},\theta_{2}}^{-1}\right)  ^{\ast}Q_{-\theta_{2}^{\ast},-\theta_{1}^{\ast}%
}(\mathcal{V})=Q_{0,0}(\mathcal{V})\mathcal{F}_{\mathcal{V}}^{-1}\left(
\mathcal{F}_{\mathcal{V},\theta_{1},\theta_{2}}^{-1}\right)  ^{\ast}\,.
\end{equation}
Let us denote with $A$ the operator of multiplication by $k^{2}$. Using the
functional calculus for $Q_{0,0}(\mathcal{V})$, this operator is represented
by: $A=\mathcal{F}_{\mathcal{V}}Q_{0,0}(\mathcal{V})\mathcal{F}_{\mathcal{V}%
}^{-1}$, and the previous relation rephrases as%
\begin{equation}
\left(  \mathcal{F}_{\mathcal{V},\theta_{1},\theta_{2}}^{-1}\right)  ^{\ast
}Q_{-\theta_{2}^{\ast},-\theta_{1}^{\ast}}(\mathcal{V})=A\left(
\mathcal{F}_{\mathcal{V},\theta_{1},\theta_{2}}^{-1}\right)  ^{\ast}\,.
\end{equation}
Then, taking the adjoint, yields%
\begin{equation}
Q_{\theta_{1},\theta_{2}}(\mathcal{V})\mathcal{F}_{\mathcal{V},\theta
_{1},\theta_{2}}^{-1}=\mathcal{F}_{\mathcal{V},\theta_{1},\theta_{2}}^{-1}A\,.
\label{gen_Fourier_teta_1}%
\end{equation}
To identify $\mathcal{W}_{\theta_{1},\theta_{2}}$ with the wave operator,
according to the time dependent definition%
\begin{equation}
W_{-}\left(  Q_{\theta_{1},\theta_{2}}(\mathcal{V}),Q_{0,0}(\mathcal{V}%
)\right)  =\text{s-}\lim_{t\rightarrow-\infty}e^{itQ_{\theta_{1},\theta_{2}%
}(\mathcal{V})}e^{-itQ_{0,0}(\mathcal{V})}\,, \label{W_}%
\end{equation}
it is enough to prove that%
\begin{equation}
\lim_{t\rightarrow-\infty}\left\Vert \left(  e^{itQ_{\theta_{1},\theta_{2}%
}(\mathcal{V})}e^{-itQ_{0,0}(\mathcal{V})}-\mathcal{W}_{\theta_{1},\theta_{2}%
}\right)  u\right\Vert _{L^{2}\left(  \mathbb{R}\right)  }=0\,.
\label{wave_op_condition}%
\end{equation}
Explicitly, the function in (\ref{wave_op_condition}) reads as%
\begin{equation}
\left(  e^{itQ_{\theta_{1},\theta_{2}}(\mathcal{V})}e^{-itQ_{0,0}%
(\mathcal{V})}-\mathcal{F}_{\mathcal{V},\theta_{1},\theta_{2}}^{-1}%
\mathcal{F}_{\mathcal{V}}\right)  u\,. \label{wave_op_condition1}%
\end{equation}
Setting $g=\mathcal{F}_{\mathcal{V}}u$, we have: $e^{-itQ_{0,0}(\mathcal{V}%
)}u=\mathcal{F}_{\mathcal{V}}^{-1}e^{-itA}g$, and (\ref{wave_op_condition1})
rephrases as%
\begin{equation}
e^{itQ_{\theta_{1},\theta_{2}}(\mathcal{V})}\left(  \mathcal{F}_{\mathcal{V}%
}^{-1}e^{-itA}-e^{-itQ_{\theta_{1},\theta_{2}}(\mathcal{V})}\mathcal{F}%
_{\mathcal{V},\theta_{1},\theta_{2}}^{-1}\right)  g\,.
\label{wave_op_condition2}%
\end{equation}
Then, using (\ref{gen_Fourier_teta_1}) and the definitions (\ref{gen_Fourier}%
), (\ref{gen_Fourier_teta_inverse}), we get%
\begin{gather}
\left(  e^{itQ_{\theta_{1},\theta_{2}}(\mathcal{V})}e^{-itQ_{0,0}%
(\mathcal{V})}-\mathcal{W}_{\theta_{1},\theta_{2}}\right)  u=e^{itQ_{\theta
_{1},\theta_{2}}(\mathcal{V})}\left(  \mathcal{F}_{\mathcal{V}}^{-1}%
-\mathcal{F}_{\mathcal{V},\theta_{1},\theta_{2}}^{-1}\right)  e^{-itA}%
g\,\medskip\qquad\qquad\qquad\qquad\nonumber\\
=e^{itQ_{\theta_{1},\theta_{2}}(\mathcal{V})}\int_{\mathbb{R}}dk\,\left(
\psi_{-}(\cdot,k)-\psi_{-}(\cdot,k,\theta_{1},\theta_{2})\right)  e^{-itk^{2}%
}g(k),\qquad g\in\mathcal{F}_{\mathcal{V}}u\,. \label{wave_op_condition3}%
\end{gather}
Under our assumptions, the result of the Corollary
\ref{Corollary_eigenfun_exp} applies and the r.h.s. of
(\ref{wave_op_condition3}) can be further developed through the expansion
(\ref{gen_eigenfun_exp}). This yields%
\begin{gather}
\left(  e^{itQ_{\theta_{1},\theta_{2}}(\mathcal{V})}e^{-itQ_{0,0}%
(\mathcal{V})}-\mathcal{W}_{\theta_{1},\theta_{2}}\right)  u=\medskip
\qquad\qquad\qquad\qquad\qquad\qquad\qquad\qquad\qquad\qquad\nonumber\\
\int_{\mathbb{R}}dk\,\mathcal{O}\left(  \theta_{2}k\right)  G^{\sigma k}%
(\cdot,b)e^{-itk^{2}}g(k)+\int_{\mathbb{R}}dk\,\mathcal{O}\left(  \frac
{\theta_{1}k}{1+\left\vert k\right\vert }\right)  H^{\sigma k}(\cdot
,b)e^{-itk^{2}}g(k)\medskip\nonumber\\
+\int_{\mathbb{R}}dk\,\mathcal{O}\left(  \theta_{2}k\right)  G^{\sigma
k}(\cdot,a)e^{-itk^{2}}g(k)+\int_{\mathbb{R}}dk\,\mathcal{O}\left(
\frac{\theta_{1}k}{1+\left\vert k\right\vert }\right)  H^{\sigma k}%
(\cdot,a)e^{-itk^{2}}g(k)\,. \label{wave_op_exp}%
\end{gather}
where $\sigma=\frac{k}{\left\vert k\right\vert }$, while the functions
$\mathcal{O}\left(  \theta_{2}k\right)  $ and $\mathcal{O}\left(  \frac
{\theta_{1}k}{1+\left\vert k\right\vert }\right)  $ are independent of $x$. To
obtain (\ref{wave_op}), it is enough to show that, for all $g\in
\mathcal{C}_{0}^{\infty}\left(  \mathbb{R}\right)  $, limits of the type
(\ref{wave_op_condition}) are available for each term at the r.h.s. of
(\ref{wave_op_exp}). In what follows, we consider the first contribution of
(\ref{wave_op_exp}), while the other terms can be treated within the same
approach. Since we work with fixed $\left(  \theta_{1},\theta_{2}\right)  $,
the dependence from these parameters is next omitted. Our aim is to prove%
\[
\lim_{\left\vert t\right\vert \rightarrow\infty}\left\Vert \int_{\mathbb{R}%
}dk\,\mathcal{O}\left(  k\right)  G^{\sigma k}(\cdot,b)e^{-itk^{2}%
}g(k)\right\Vert _{L^{2}(\mathbb{R})}=0\,,
\]
when $g\in\mathcal{C}_{0}^{\infty}\left(  \mathbb{R}\right)  $. We use the
definitions (\ref{G_z}), (\ref{GH_zeta}), with $\zeta=k$, to write
\begin{equation}
\left\Vert \int_{\mathbb{R}}dk\,\mathcal{O}\left(  k\right)  G^{\sigma
k}(\cdot,b)e^{-itk^{2}}g(k)\right\Vert _{L^{2}(\mathbb{R})}=I+II
\label{wave_op_exp1}%
\end{equation}
where%
\begin{align}
I  &  =\int_{-\infty}^{b}dx\left\vert \int_{\mathbb{R}}dk\,\frac
{\mathcal{O}\left(  k\right)  }{w(\left\vert k\right\vert )}\chi
_{-}(x,\left\vert k\right\vert )\chi_{+}(b,\left\vert k\right\vert
)e^{-itk^{2}}g(k)\right\vert ^{2}\,,\label{wave_op_exp1_1}\\
& \nonumber\\
II  &  =\int_{b}^{+\infty}dx\left\vert \int_{\mathbb{R}}dk\,\frac
{\mathcal{O}\left(  k\right)  }{w(\left\vert k\right\vert )}\chi
_{+}(x,\left\vert k\right\vert )\chi_{-}(b,\left\vert k\right\vert
)e^{-itk^{2}}g(k)\right\vert ^{2}\,. \label{wave_op_exp1_2}%
\end{align}
Recall that $\chi_{\pm}$ are defined through the equations%
\begin{align}
1_{x<b}(x)\chi_{+}(x,k)  &  =e^{ikx}+%
{\displaystyle\int\limits_{x}^{b}}
\frac{\sin k(t-x)}{k}\mathcal{V}(t)\chi_{+}(t,k)\,dt\,,\qquad1_{x\geq
b}(x)\chi_{+}(x,k)=e^{ikx}\,,\label{Int_Jost_eq_0}\\
& \nonumber\\
1_{x>a}(x)\chi_{-}(x,k)  &  =e^{-ikx}-%
{\displaystyle\int\limits_{a}^{x}}
\frac{\sin k(t-x)}{k}\mathcal{V}(t)\chi_{-}(t,k)\,dt\,,\qquad1_{x\leq
a}(x)\chi_{-}(x,k)=e^{-ikx}\,, \label{Int_Jost_eq_0.1}%
\end{align}
and introduce the functions%
\begin{equation}
\gamma_{+}(x)=%
{\displaystyle\int\limits_{x}^{b}}
\left\vert \mathcal{V}(t)\right\vert \,dt\,,\qquad\gamma_{-}(x)=%
{\displaystyle\int\limits_{a}^{x}}
\left\vert \mathcal{V}(t)\right\vert \,dt\,. \label{gamma_+-}%
\end{equation}
For compactly supported short-range potentials (see (\ref{V})), it results:
$\gamma_{+}\in L^{2}\left(  \left(  a,+\infty\right)  \right)  $ and
$\gamma_{-}\in L^{2}\left(  \left(  -\infty,b\right)  \right)  $. As stated in
the Proposition \ref{Proposition_Jost}, $\chi_{\pm}$ are uniformly bounded
w.r.t. $x,k\in\mathbb{R}$ and the relations (\ref{Int_Jost_eq_0}),
(\ref{Int_Jost_eq_0.1}) rephrase as%
\begin{align}
1_{x<b}(x)\chi_{+}(x,k)  &  =e^{ikx}+\ell_{+}(x,k)\,,\qquad1_{x\geq b}%
(x)\chi_{+}(x,k)=e^{ikx}\,,\label{Jost_exp_plus}\\
& \nonumber\\
1_{x>a}(x)\chi_{-}(x,k)  &  =e^{-ikx}+\ell_{-}(x,k)\,,\qquad1_{x\leq a}%
(x)\chi_{-}(x,k)=e^{-ikx}\,, \label{Jost_exp_min}%
\end{align}
with $\ell_{\pm}$ s.t.:%
\begin{equation}
\left\vert \ell_{\pm}(x,k)\right\vert <\frac{1}{k}\left\Vert \chi_{\pm
}\right\Vert _{L_{x,k}^{\infty}(\mathbb{R}^{2})}\gamma_{\pm}(x)\,.
\label{gamma_+-_est}%
\end{equation}
Plugging these relations into (\ref{wave_op_exp1_1})-(\ref{wave_op_exp1_2})
and using $1_{x>b}(x)g_{+}(x,k)=0$, we get%
\begin{align}
I  &  \leq\int_{-\infty}^{b}dx\left\vert \int_{\mathbb{R}}dk\,e^{-i\left\vert
k\right\vert x}e^{-itk^{2}}\frac{\mathcal{O}\left(  k\right)  e^{i\left\vert
k\right\vert b}}{w(\left\vert k\right\vert )}g(k)\right\vert ^{2}%
+\int_{-\infty}^{b}dx\left\vert \int_{\mathbb{R}}dk\,\ell_{-}(x,\left\vert
k\right\vert )e^{-itk^{2}}\frac{\mathcal{O}\left(  k\right)  e^{i\left\vert
k\right\vert b}}{w(\left\vert k\right\vert )}g(k)\right\vert ^{2}%
\,,\label{wave_op_exp1_1_exp}\\
& \nonumber\\
II  &  =\int_{b}^{+\infty}dx\left\vert \int_{\mathbb{R}}dk\,e^{i\left\vert
k\right\vert x}e^{-itk^{2}}\frac{\mathcal{O}\left(  k\right)  }{w(\left\vert
k\right\vert )}\chi_{-}(b,\left\vert k\right\vert )g(k)\right\vert ^{2}\,.
\label{wave_op_exp1_2_1}%
\end{align}
With the change of variable: $s=-x+b$, the first integral at the r.h.s. of
(\ref{wave_op_exp1_1_exp}) writes as%
\begin{equation}
\int_{-\infty}^{b}dx\left\vert \int_{\mathbb{R}}dk\,e^{-i\left\vert
k\right\vert x}e^{-itk^{2}}\frac{\mathcal{O}\left(  k\right)  e^{i\left\vert
k\right\vert b}}{w(\left\vert k\right\vert )}g(k)\right\vert ^{2}=\int
_{0}^{+\infty}ds\left\vert \int_{\mathbb{R}}dk\,e^{i\left\vert k\right\vert
s-itk^{2}}\left(  \frac{\mathcal{O}\left(  k\right)  }{w(\left\vert
k\right\vert )}g(k)\right)  \right\vert ^{2}\,,
\end{equation}
while, setting: $s=x-b$, the identity (\ref{wave_op_exp1_2_1}) rephrases as%
\begin{equation}
II=\int_{0}^{+\infty}ds\left\vert \int_{\mathbb{R}}dk\,e^{i\left\vert
k\right\vert s}e^{-itk^{2}}\left(  \frac{\mathcal{O}\left(  k\right)
e^{i\left\vert k\right\vert b}}{w(\left\vert k\right\vert )}\chi
_{-}(b,\left\vert k\right\vert )g(k)\right)  \right\vert ^{2}\,.
\end{equation}
Due to the relation $w(k)=\mathcal{O}\left(  1+\left\vert k\right\vert
\right)  $ (see the proof of Lemma \ref{Lemma_Green_ker})), the above
appearing functions,%
\[
\frac{\mathcal{O}\left(  k\right)  }{w(k)}g\left(  k\right)  \text{\quad
and\quad}\chi_{-}(b,\left\vert k\right\vert )g(k)\,,
\]
both belong to $L_{k}^{2}(\mathbb{R})$. Moreover, as a consequence of
(\ref{gamma_+-_est}), it results%
\begin{equation}
\left\vert \ell_{-}(x,\left\vert k\right\vert )\frac{\mathcal{O}\left(
k\right)  e^{i\left\vert k\right\vert b}}{w(\left\vert k\right\vert
)}\right\vert \leq\left\Vert \chi_{-}\right\Vert _{L_{x,k}^{\infty}%
(\mathbb{R}^{2})}\left\vert \frac{\mathcal{O}\left(  k\right)  }{kw(\left\vert
k\right\vert )}\right\vert \gamma_{-}(x)\lesssim\gamma_{-}(x)\in L^{2}\left(
\left(  -\infty,b\right)  \right)  \,,
\end{equation}
We get%
\begin{align}
I  &  \leq\int_{0}^{+\infty}ds\left\vert \int_{\mathbb{R}}dk\,e^{i\left\vert
k\right\vert s-itk^{2}}q_{1}(k)\right\vert ^{2}+\int_{-\infty}^{b}dx\left\vert
\int_{\mathbb{R}}dk\,e^{-itk^{2}}q_{2}(k,x)g(k)\right\vert ^{2}\,,\\
& \nonumber\\
II  &  =\int_{0}^{+\infty}ds\left\vert \int_{\mathbb{R}}dk\,e^{i\left\vert
k\right\vert s-itk^{2}}q_{3}(k)\right\vert ^{2}\,.
\end{align}
where, according to the previous remarks, $q_{1},q_{3}\in L_{k}^{2}%
(\mathbb{R})$, while $q_{2}$ allows the estimate%
\begin{equation}
\left\vert q_{2}\left(  k,x\right)  \right\vert \leq f(x)\in L^{2}\left(
\left(  -\infty,b\right)  \right)  \,.
\end{equation}
Then, it follows from an application of the Lemma 2.6.4 of \cite{Yafa1} that%
\begin{equation}
\lim_{t\rightarrow-\infty}\int_{0}^{+\infty}ds\left\vert \int_{\mathbb{R}%
}dk\,e^{i\left\vert k\right\vert s-itk^{2}}q_{j}(k)\right\vert ^{2}=0\,,\quad
j=1,3\,, \label{wave_op_exp1_1_lim1}%
\end{equation}
Moreover, for $g\in\mathcal{C}_{0}^{\infty}\left(  \mathbb{R}\right)  $, the
Riemann-Lebesgue Lemma implies%
\begin{equation}
\lim_{\left\vert t\right\vert \rightarrow\infty}\int_{\mathbb{R}%
}dk\,e^{-itk^{2}}q_{2}(k,x)g(k)=0\,;
\end{equation}
thus, using the estimate%
\[
\left\vert \int_{\mathbb{R}}dk\,e^{-itk^{2}}q_{2}(k,x)\,g(k)\right\vert
\lesssim\left\vert f(x)\right\vert \left(  \int_{\mathbb{R}}dk\,\left\vert
g(k)\right\vert \right)  \lesssim\left\vert f(x)\right\vert \,,
\]
and the dominated convergence theorem, we get%
\begin{equation}
\lim_{\left\vert t\right\vert \rightarrow\infty}\left\Vert \int_{\mathbb{R}%
}dk\,e^{-itk^{2}}q_{2}(k,x)\,g(k)\right\Vert _{L^{2}\left(  -\infty,b\right)
}=0\,. \label{wave_op_exp1_1_lim2}%
\end{equation}

\end{proof}

The above result exploits the condition: $\left(  \theta_{1},\theta
_{2}\right)  \in\mathcal{B}_{\delta}(\left(  0,0\right)  )$ which has been
previously introduced to identify the spectra of the operators $Q_{\theta
_{1},\theta_{2}}(\mathcal{V})$ and $Q_{0,0}(\mathcal{V})$, and to compare the
corresponding quantum dynamics. Nevertheless a question is left open: is a
small parameter condition necessary in order that the pair $\left\{
Q_{\theta_{1},\theta_{2}}(\mathcal{V}),Q_{0,0}(\mathcal{V})\right\}  $ forms a
complete scattering system ? Actually this restriction does not seems to be
necessary. It has been shown, indeed, that a key point in the development of
the scattering theory for the possibly non-selfadjoint pair $\left\{
H_{0},H_{1}\right\}  $ is the existence of the strong limit on the real axis
of the characteristic functions associated with $H_{i=0,1}$ (e.g. in
\cite{Ryz01} and \cite{Ryz1}). In particular, the Theorem 4.1 in \cite{Ryz01}
makes use of this assumption to study the existence of the related wave
operators. According to \cite{Ryz0}, the resolvent formula (\ref{krein_1})
implies that, for any $\left(  \theta_{1},\theta_{2}\right)  $, the
characteristic function of the operator $Q_{\theta_{1},\theta_{2}}%
(\mathcal{V})$ has boundary values a.e. on the real axis (for this point we
refer to the last Proposition in \cite{Ryz0} and to the references therein).
This suggests the possibility of defining the scattering system $\left\{
Q_{\theta_{1},\theta_{2}}(\mathcal{V}),Q_{0,0}(\mathcal{V})\right\}  $ without
restrictions on $\left(  \theta_{1},\theta_{2}\right)  $.

A slightly different approach to the scattering problem consists in
characterizing the scattering matrix for $\left\{  Q_{\theta_{1},\theta_{2}%
}(\mathcal{V}),Q_{0,0}(\mathcal{V})\right\}  $. In the case of selfadjoint
extensions of a symmetric operator, a relation between the scattering matrix
and the Weyl function, associated with a boundary triple, have been
established in \cite{BeMaNe}, while extensions of results from \cite{BeMaNe}
to certain non-selfadjoint situations (dissipative/accumulative) have been
presented in, \cite{BeMaNe1}, \cite{BeMaNe2}. A generalization to the case of
$Q_{\theta_{1},\theta_{2}}(\mathcal{V})$ would represents a useful insight in
the study of the scattering properties of the system $\left\{  Q_{\theta
_{1},\theta_{2}}(\mathcal{V}),Q_{0,0}(\mathcal{V})\right\}  $.

\section{\label{Sec_small_h}Further perspectives: the regime of quantum wells
in a semiclassical island.}

Let introduce the modified operators $Q_{\theta_{1},\theta_{2}}^{h}%
(\mathcal{V})$, depending on the small parameter $h\in\left(  0,h_{0}\right)
$, $h_{0}>0$, and defined according to%
\begin{equation}
Q_{\theta_{1},\theta_{2}}^{h}(\mathcal{V}):\left\{
\begin{array}
[c]{l}%
D\left(  Q_{\theta_{1},\theta_{2}}^{h}(\mathcal{V})\right)  =\left\{  u\in
H^{2}\left(  \mathbb{R}\backslash\left\{  a,b\right\}  \right)  \,\left\vert
\ \text{(\emph{\ref{B_C_1}}) holds}\right.  \right\}  \,,\\
\\
\left(  Q_{\theta_{1},\theta_{2}}^{h}(\mathcal{V}^{h})\,u\right)
(x)=-h^{2}u^{\prime\prime}(x)+\mathcal{V}^{h}(x)\,u(x)\,,\qquad x\in
\mathbb{R}\backslash\left\{  a,b\right\}  \,.
\end{array}
\right.  \label{Q_teta_h}%
\end{equation}
with $\mathcal{V}^{h}$ $h$-dependent and locally supported on $\left[
a,b\right]  $. In the applications perspectives, rather relevant is the case
of a positive and bounded $\mathcal{V}^{h}$ formed by the superposition of a
potential barrier and a collection of potential wells supported on a region of
size $h$ (these are usually referred to as \emph{quantum wells}). Hamiltonians
of this type have been introduced in \cite{FMN2}, where the case: $\theta
_{2}=3\theta_{1}$ is considered, with the purpose of realizing models of
electronic transverse transport through resonant heterostructures. When the
initial state describes incoming charge carriers in the conduction band, the
quantum dynamics of such systems is expected be driven by the, possibly
non-linear, adiabatic evolution of a finite number of resonant states related
to the shape resonances. This picture, arising in the physical literature, is
confirmed by the analysis presented in \cite{JoPrSj}, \cite{PrSj} concerning
the case of a 1D Schr\"{o}dinger-Poisson selfadjoint model with a double
barrier, and in \cite{FMN3}, where an application involving Hamiltonians of
the type $Q_{\theta_{1},\theta_{2}}^{h}(\mathcal{V})$ is considered.

As previously remarked, the artificial interface conditions allow to develop
an alternative approach to the adiabatic evolution for shape resonances. In
the particular case where $\theta_{2}=3\theta_{1}$, using suitable exterior
dilations, it is possible to write the evolution problem for the resonant
states of the modified operators $Q_{\theta_{1},\theta_{2}}^{h}(\mathcal{V})$
as a dynamical system of contractions. Then the adiabatic approximations can
be obtained by using a 'standard' approach and reasonably weak assumptions on
the regularity-in-time of the potential, as it has been shown in \cite{FMN2},
while a similar strategy would not work in the selfadjoint case, due to the
lack of accretivity of the corresponding complex deformed operator. This
justifies the interest in the operators $Q_{\theta_{1},\theta_{2}}%
^{h}(\mathcal{V})$ as models for the electronic quantum transport in the
regime of quantum wells in a semiclassical island, and motivates an extension
of the previous analysis taking into account the role of the scaling parameter
$h$ in the definition of the modified dynamics. Our aim is to obtain, in the
$h$-dependent case, a comparison between the dynamical system modified by
non-mixed interface conditions and the unitary dynamics generated by the
selfadjoint Hamiltonians $Q_{0,0}^{h}(\mathcal{V}^{h})$. Proceeding in this
direction, a first step consists in the study of the Jost's solutions, the
generalized eigenfunctions and the Green's kernel associated to operators of
the type $Q_{0,0}^{h}(\mathcal{V})$, which, according to the definition
(\ref{Q_teta_h}), is given by%
\begin{equation}
Q_{0,0}^{h}(\mathcal{V}):\left\{
\begin{array}
[c]{l}%
D\left(  Q_{0,0}^{h}(\mathcal{V})\right)  =H^{2}\left(  \mathbb{R}\right)
\,\,,\\
\\
\left(  Q_{0,0}^{h}(\mathcal{V})\,u\right)  (x)=-h^{2}u^{\prime\prime
}(x)+\mathcal{V}(x)\,u(x)\,,\qquad x\in\mathbb{R}\,,
\end{array}
\right.  \label{Q_0_h}%
\end{equation}
with $\mathcal{V}$, possibly depending on $h$, locally supported on $\left[
a,b\right]  $. In what follows, $\chi_{\pm}^{h}\left(  \cdot,\zeta
,\mathcal{V}\right)  $ denote the solutions of the equation%
\begin{equation}
\left(  -h^{2}\partial_{x}^{2}+\mathcal{V}\right)  u=\zeta^{2}u\,,
\label{Jost_eq_h1}%
\end{equation}
fulfilling the conditions%
\begin{equation}
\left.  \chi_{+}^{h}\left(  \cdot,\zeta,\mathcal{V}\right)  \right\vert
_{x>b}=e^{i\frac{\zeta}{h}x}\,,\qquad\left.  \chi_{-}^{h}\left(  \cdot
,\zeta,\mathcal{V}\right)  \right\vert _{x<a}=e^{-i\frac{\zeta}{h}x}\,.
\label{Jost_sol_ext_h}%
\end{equation}
It is worthwhile to notice that, in the attempt of extending our approach to
this new setting, all the estimates involved in the proofs depend on $h$ and
exhibit exponential bounds w.r.t. to the small parameter. To fix this point,
let $h>0$ and consider the functions $\chi_{\pm}^{h}$. The rescaled functions
$b_{\pm}^{h}\left(  x,\zeta\right)  =e^{\mp i\frac{\zeta}{h}x}\chi_{\pm}%
^{h}\left(  x,\zeta\right)  $ are defined through a Picard iteration procedure
(see the e.g. in the eq. (\ref{Jost_Picard})). Taking into account the
small-$h$ behviour of the corresponding rescaled kernels, it results%
\begin{equation}
\sup_{x\in\mathbb{R}\,,\ \zeta\in\overline{\mathbb{C}^{+}}}\left\vert b_{\pm
}(x,\zeta)\right\vert \leq e^{\frac{C_{0}}{h^{2}}}\,,\qquad\sup_{x\in
\mathbb{R}\,,\ \zeta\in\overline{\mathbb{C}^{+}}}\left\vert b_{\pm}^{\prime
}(x,\zeta)\right\vert \leq e^{\frac{C_{1}}{h^{2}}}\,,
\end{equation}
where the coefficients $C_{i}$, $i=0,1$, possibly depend on the data $a$, $b$,
and $\left\Vert \mathcal{V}\right\Vert _{L^{1}(a,b)}$. Then, a suitable
rewriting of the Krein's formula (\ref{krein_1}) and of the results of the
Lemmata \ref{Lemma_Green_ker} and \ref{Lemma_Krein_coeff} in the $h$-dependent
case would allows to express the $h$-dependent wave operators as:
$\mathcal{W}_{\theta_{1},\theta_{2}}^{h}=1+\mathcal{O}\left(  h\right)
+\mathcal{O}\left(  h\right)  \,$, provided that%
\begin{equation}
\left(  \theta_{1},\theta_{2}\right)  \in\mathcal{B}_{\rho(h)}\left(  \left(
0,0\right)  \right)  \,,\quad\text{with }\rho(h)=he^{-\frac{\tilde{C}}{h^{2}}}
\label{teta_h}%
\end{equation}
and $\tilde{C}>0$ large enough.

Operators defined with the prescription (\ref{teta_h}) appear to be of small
interest in the applications perspective. At this concern we recall that the
adiabatic theorem obtained in \cite{FMN2} applies with: $\theta_{i}%
=c_{i}h^{N_{0}}$, $i=1,2$, for some $N_{0}\in\mathbb{N}$ (see Theorem $7.1$ in
\cite{FMN2}). In this connection, it is important to relax the constraint
expressed by (\ref{teta_h}) in order to obtain small-$h$ expansions of the
waves operators, holding at least in a suitable subspace of $L^{2}$, when the
parameters are assumed to be only polynomially smalls w.r.t. $h$.

According to formulas of Section \ref{Section_Resolvent_2}, the key to obtain
small-$\theta_{i=1,2}$ expansions of the generalized eigenfunctions, and then
of the wave operators, consists in controlling the boundary values of Green's
functions as $z$ approaches the continuous spectrum. Introducing quantum wells
in the model, produces resonances with exponentially small imaginary parts as
$h\rightarrow0$. This means that, the Green's functions, which are expressed
in terms of the Jost solutions and the Jost function, will be exponentially
large w.r.t. $h$ somewhere in the potential structure when $z$ is close to the
corresponding energies. Nevertheless, their values on the boundary of the
potential support are expected to be, at most, of order $\mathcal{O}\left(
\frac{1}{h^{N_{0}}}\right)  $ for a suitable $N_{0}\in\mathbb{N}$; an explicit
example of this mechanism can be found in \cite{FMN3}. Studying the Green
function around a resonant energy requires the introduction of a Dirichlet
problem in order to resolve the spectral singularity and to match the complete
problem with some combination of this spectral problem with the filled wells
spectral problem. Following \cite{HeSj1}, \cite{Hel}, the Grushin technique
can be used for handling this matching and obtain resolvent approximations.
Developing this approach is a further perspective of our work.

\begin{description}
\item[Akcnowledgement] This work arises from a question addressed by A. Teta
and has largely profited from useful discussions with F. Nier and C.A. Pillet.
The author is also indebted to S. Naboko and H. Neidhardt for their important remarks.
\end{description}

\end{document}